\documentclass[a4paper,11pt]{article}

\makeatletter
\@addtoreset{equation}{section}

\makeatother

\usepackage{multirow}
\usepackage{booktabs}
\usepackage[dvips]{graphicx}
\usepackage{amsmath}
\usepackage{amsthm}
\usepackage{amssymb}
\usepackage{eepic}
\usepackage{amscd}
\usepackage{longtable}
\usepackage{array}

\newtheorem{theorem}{Theorem}[section]
\newtheorem{lemma}[theorem]{Lemma}
\newtheorem{conjecture}[theorem]{Conjecture}
\newtheorem{proposition}[theorem]{Proposition}
\newtheorem{corollary}[theorem]{Corollary}

\theoremstyle{definition}

\newtheorem{example}[theorem]{Example}
\newtheorem{remark}[theorem]{Remark}

\newcommand{\ot}{\otimes}
\newcommand{\ol}[1]{\overline{#1}}
\newcommand{\Z}{{\mathbb Z}}

\newcommand{\C}{{\mathbb C}}

\newcommand{\mn}[2]{
\left(\!\!\begin{array}{c} {#1}  \\  \! {#2}\!
\end{array}\!\!\right)}
\newcommand{\mc}[2]{
\left\langle\!\!\begin{array}{c} {#1}  \\  \! {#2}\!
\end{array}\!\!\right\rangle}

\newcommand{\nc}{\newcommand}
\nc{\rnc}{\renewcommand}
\nc{\nn}{\nonumber}
\nc{\bra}{\langle}
\nc{\ket}{\rangle}
\nc{\T}{\mathcal{T}}
\rnc{\i}{{\rm i}}
\rnc{\d}{{\rm d}}
\nc{\vac}{{\rm vac}}
\nc{\e}{{\rm e}}
\nc{\tr}{{\rm tr}}
\nc{\mr}{\midrule}
\nc{\br}{\bottomrule}

\DeclareMathOperator{\sh}{sh}
\DeclareMathOperator{\ch}{ch}
\DeclareMathOperator{\End}{End}
\DeclareMathOperator{\Spec}{Spec}

\def\s{{\mathfrak s}}

\title{Spectrum in multi-species asymmetric simple exclusion process on a ring}

\author{
Chikashi Arita\thanks{Faculty of Mathematics, Kyushu University},
Atsuo Kuniba\thanks{Institute of Physics, University of Tokyo},
Kazumitsu Sakai$^\dag$ 
and
Tsuyoshi Sawabe\thanks{Data Mining Division,
Mathematical Systems, Inc.}
}

\date{April 9, 2009}

\begin{document}

\maketitle

\begin{abstract}
The spectrum of Hamiltonian (Markov matrix) 
of a multi-species asymmetric simple
exclusion process on a ring is studied.
The dynamical exponent concerning 
the relaxation time is found to coincide with 
the one-species case.
It implies that the system belongs
to the Kardar-Parisi-Zhang or Edwards-Wilkinson
universality classes depending on whether 
the hopping rate is asymmetric or symmetric,
respectively.
Our derivation exploits a poset structure of 
the particle sectors, leading to a new 
spectral duality and inclusion relations. 
The Bethe ansatz integrability is also demonstrated.
\end{abstract}

\section{Introduction}\label{sec:1}

In recent years, intensive studies on non-equilibrium phenomena
have been undertaken through a variety of 
stochastic process models of many particle systems \cite{S1}.
Typical examples are driven lattice gas systems
with a simple but nonlinear interaction among 
constituent particles.
The asymmetric simple exclusion process (ASEP) 
is one of the simplest driven lattice gas models 
proposed originally to describe the dynamics 
of ribosome along RNA \cite{MGP}. 
In the ASEP, each site is occupied by at most one particle. 
Each particle is allowed to hop to its nearest neighbor right (left) site 
with the rate $p$ $(q)$ if it is empty. 
The ASEP admits exact analyses of 
non-equilibrium properties by the matrix product ansatz 
and the Bethe ansatz.
The matrix product form of the stationary state 
was first found in the open boundary case \cite{DEHP}. 
Similar results have been obtained in  
various driven lattice gas systems 
in one dimension with both open and periodic 
boundary conditions \cite{BE}. 
Applications of the Bethe ansatz \cite{B} to the ASEP
have also been successful.
See for example the works \cite{GS,K}, \cite{S2} and \cite{DE}
for the studies of the ASEP under periodic, infinite and open 
boundary conditions, respectively.
In the periodic boundary case, 
the dynamical exponent $z$ of the relaxation time 
to the stationary state is 
$z=3/2$ if $p\neq q$ and $z=2$ if $p=q$
according to the analysis of the Bethe equation.
This implies that the ASEP belongs to the Kardar-Parisi-Zhang (KPZ)
or Edwards-Wilkinson (EW) universality classes depending on 
whether the hopping rate is asymmetric or symmetric, respectively. 

In this paper we consider the following multi-species generalization 
of the ASEP on a ring of length $L$.
A local state on each site of the ring assumes $N$ states
$1,2,\ldots, N$.
A nearest neighbor pair of the states is interchanged with the transition rate
\begin{align*}
\alpha \beta \to \beta \alpha \;\;\begin{cases}
p& \text{if } \alpha > \beta,\\
q& \text{if } \alpha < \beta.
\end{cases}
\end{align*}
We regard the local state $\alpha=1$ as a vacant site and 
$\alpha=2,\ldots, N$ as a site occupied by a particle of the $\alpha$th kind. 
The dynamics is formulated in terms of the master equation
$\frac{d}{dt}|P(t)\rangle=H|P(t)\rangle$ 
on the probability vector $|P(t)\rangle$.
We call the model the ($N\!-\!1$)-species ASEP or 
simply the multi-species ASEP.
The usual ASEP corresponds to $N=2$.

The $N^L \times N^L$ Markov matrix $H$ will be called 
the ``Hamiltonian" although it is not Hermitian for $p\neq q$.
(At $p=q$, it is Hermitian and
coincides with the $sl(N)$ invariant Heisenberg Hamiltonian.) 
Its eigenvalue with the largest real part is 0 
corresponding to the stationary state.
The other eigenvalues contribute to the
relaxation behavior through
$|P(t)\rangle = \e^{tH}|P(0)\rangle$.
Especially those with the second largest real part 
determine the relaxation time $\tau$ and 
the dynamical exponent $z$ by the scaling
behavior $\tau \sim L^z$.
Our finding is that $z$ remains the same 
with the usual ASEP $N=2$.
The key to the result is that the spectrum of $H$ 
in any multi-species particle sector includes that for an $N=2$ sector.
We systematize such spectral inclusion relations 
by exploiting the poset structure of particle sectors.
As a byproduct we find a duality, a global aspect in the 
spectrum of the Hamiltonian,
which is new even at the Heisenberg point $p=q$.

There are some other multi-species models  
with different hopping rules 
such as ABC model \cite{EKKM} and AHR model \cite{AHR,RSS}, etc.
The recent paper \cite{KN} says
that the AHR model still has the dynamical exponent
$z=3/2$.
The multi-species ASEP in this paper is a 
most standard generalization of the one-species ASEP
allowing the application of the Bethe ansatz.
Let us comment on this point and related works   
to clarify the origin of integrability of the model.
In \cite{AR}, a formulation by Hecke algebra was given
for a wider class of stochastic models 
including reaction-diffusion systems.  
A Bethe ansatz treatment was presented in \cite{AB}.
As is well known \cite{Ba}, however, 
there underlies a two dimensional integrable vertex model
behind a Bethe ansatz solvable Hamiltonian.
In the present model, the relevant vertex model is a special case of 
the Perk-Schultz model (called ``second class of models" in \cite{PS}).
It should be fair to say that 
the nested Bethe ansatz for the present model  
is originally due to Schultz \cite{Sc}.
For an account of the Perk-Schultz $R$-matrix in the framework of  
multi-parameter quantum group, see \cite{OY} and reference therein.

The layout of the paper is as follows.
In section \ref{a:sec:2},
we introduce the multi-species ASEP Hamiltonian 
together with its basic properties.
In section \ref{sec:3},
we explain how the spectral gap responsible for the 
relaxation time is reduced to the one-species case and 
determine the dynamical exponent as in \eqref{k:eq:tl}.
Our argument is based on 
spectral inclusion property \eqref{a:increl}
and conjecture \ref{k:con:sp} supported by numerical analyses.
In section \ref{sec:4}, 
we elucidate a new duality of the spectrum of the Hamiltonian
in theorem \ref{k:th:sd}.
An intriguing feature is that it emerges only 
by dealing with all the basic sectors
of the  $N=L$-state model on the length $L$ ring.
(The term ``basic sector" will be defined in
section \ref{a:sec:2}.)
The argument of section \ref{sec:4} is independent from the Bethe ansatz. 
In section \ref{sec:5},  we discuss the 
Bethe ansatz integrability of the multi-species ASEP and the 
underlying vertex model including the completeness issue.
We execute the nested algebraic Bethe ansatz 
in an arbitrary nesting order, which leads to an 
alternative explanation of the spectral inclusion property.
Except the arbitrariness of the nesting order,
section \ref{k:sec:tme} is a review.

Appendix \ref{app:dd} gives a proof of 
the dimensional duality (theorem \ref{k:th:dd})
needed in section \ref{sec:4}.
Appendix \ref{sec:stationary}
is a sketch of a derivation of the stationary state
by the Bethe ansatz.
Appendix \ref{sec:appC} contains the complete spectra of 
transfer matrix and Hamiltonian in the basic sectors
with the corresponding Bethe roots for 
$(p,q)=(\frac{2}{3}, \frac{1}{3})$ and  
$L=4$. They agree with the completeness conjectures
in section \ref{k:subs:cp}.

\section{Multi-species ASEP}\label{a:sec:2}

\subsection{Master equation}\label{k:subsec:model}
Consider an $L$-site ring ${\mathbb Z}_L$ where 
each site $i \in {\mathbb Z}_L$, 
is assigned with a variable (local state)
$k_i \in \{1, \ldots, N\}\, (N \ge 1)$.
We introduce a stochastic model on ${\mathbb Z}_L$ such that 
nearest neighbor pairs of local states 
$(\alpha, \beta)=(k_i, k_{i+1})$
are interchanged with the transition rate: 
\begin{align}\label{a:rule}
\alpha \beta \to \beta \alpha \;\;\begin{cases}
p& \text{if } \alpha > \beta,\\
q& \text{if } \alpha < \beta,
\end{cases}
\end{align}
where $p$ and $q$ are real nonnegative parameters.
More precisely, the dynamics is 
formulated in terms of the continuous-time master equation
on the probability of finding the configuration
$(k_1,\dots,k_L)$ at time $t$:
\begin{align}\label{a:master1}
\begin{split}
\frac{d}{dt}P(k_1,\dots,k_L;t)
&=\sum_{i \in {\mathbb Z}_L}\Theta(k_{i+1}-k_i)
P(k_1 ,\dots,k_{i-1},k_{i+1},k_i,k_{i+2},
\dots, k_L;t) \\
&-\sum_{i \in {\mathbb Z}_L}\Theta(k_i-k_{i+1})
P(k_1 ,\dots, k_L;t),
\end{split}
\end{align}
where $\Theta$ is a step function defined as
\begin{align}\label{k:eq:step}
\Theta(x)
=\begin{cases}
p & (x>0), \\
0 & (x=0), \\
q & (x<0).
\end{cases}
\end{align}
(Actually, $\Theta(0)$ can be set to any value.)

Our model can be regarded as an 
interacting multi-species particle system on the ring.
We interpret the local state $k_i=\alpha$ as representing 
the site $i$ occupied by a particle of the $\alpha$th kind.
The transition (\ref{a:rule}) is viewed as a local 
hopping process of particles.
We identify the first kind particles with vacancies. 
Note that a particle $\alpha$ has been assumed to overtake any 
$\beta (< \alpha)$
with the same rate $p$ 
as the vacancy $\beta=1$. 

We call this model the multi-species ASEP or more specifically 
the $(N\!-\!1)$-species ASEP.
The usual ASEP corresponds to $N=2$.
We will be formally concerned with 
the zero-species ASEP ($N=1$) as well.
The case $p=q$ will be called 
the multi-species symmetric simple exclusion process (SSEP).

Let $ |1\rangle, \ldots,  |N\rangle$ be the basis of the 
single-site space $\C^N$  
and represent a particle configuration 
$(k_1,\dots,k_L)$ as the ket vector
$ |k_1, \dots, k_L\rangle
 = |k_1\rangle \otimes \cdots \otimes |k_L\rangle
\in (\C^N)^{\otimes L}$.
In terms of the probability vector
\begin{align}
 |P(t)\rangle = \sum_{1\le k_i \le N}
 P(k_1,\dots,k_L; t)|k_1, \dots, k_L\rangle,
\end{align}
the master equation \eqref{a:master1} is expressed as 
\begin{align}\label{a:master2}
\frac{d}{dt}|P(t)\rangle = H|P(t)\rangle,
\end{align}
where the linear operator $H$ has the form
\begin{align}
 H&= \sum_{i \in {\mathbb Z}_L} h_{i, i+1},
\label{k:hdef}\\
h &= \sum_{1\le \alpha<\beta\le N}
(-pE_{\beta \beta}\otimes E_{\alpha \alpha}
-q E_{\alpha \alpha}\otimes E_{\beta \beta}
+ pE_{\alpha \beta}\otimes E_{\beta \alpha}
+ qE_{\beta \alpha}\otimes E_{\alpha \beta}).
\label{k:eq:H}
\end{align}
Here $h_{i, i+1}$ acts 
on the $i$th and the $(i\!+\!1)$th components
of the tensor product as $h$ and as the identity elsewhere.
$E_{\alpha \beta}$ denotes the $N$ by $N$ 
matrix unit sending $|\gamma\rangle$ to 
$\delta_{\gamma \beta}|\alpha\rangle$.

The equation \eqref{a:master2} has the form of 
the Schr\"odinger equation with imaginary time and thus provides 
our multi-species ASEP with a quantum Hamiltonian formalism.
Of course in the present case,
$P(k_1,\dots,k_L; t)$ itself gives the probability distribution
unlike the squared wave functions 
in the case of quantum mechanics.
Nevertheless we call the 
matrix $H$ the {\em Hamiltonian}\footnote{
A more proper terminology is 
Markov matrix.} in this paper by the abuse of language.

\subsection{Basic properties of Hamiltonian}\label{k:subsec:bph}

Our Hamiltonian 
$H$ is an $N^L$ by $N^L$ matrix 
whose off-diagonal elements are $p, q$ or $0$, 
and diagonal elements belong to 
$p{\mathbb Z}_{\le 0}+q{\mathbb Z}_{\le 0}$.
Each column of $H$ sums up to $0$ 
assuring the conservation of the total probability
$\sum_{1\le k_i \le N}P(k_1,\dots,k_L; t)$.
It enjoys the symmetries
\begin{align}
[H, C] = 0,\quad
R H R^{-1} = Q H Q^{-1} 
= H\vert_{p \leftrightarrow q}.
\end{align}
where $C, R$ and $Q$ are linear operators defined by
\begin{alignat}{2}
C \vert k_1, \ldots, k_L\rangle 
&= \vert k_L, k_1, \ldots, k_{L-1}\rangle&\quad&
(\text{cyclic shift}),\label{k:eq:cs}\\
R \vert k_1, \ldots, k_L\rangle 
&= \vert k_L,  k_{L-1},\ldots, k_1\rangle&&
(\text{reflection}),\\
Q \vert k_1, \ldots, k_L\rangle 
&= \vert N\!+\!1\!-\!k_1,\ldots, N\!+\!1\!-\!k_L\rangle&&
(\text{``charge conjugation"})\label{k:eq:cc}
\end{alignat}
satisfying $C^L=R^2=Q^2=1$.
$H$ is Hermitian only at 
$p=q$, where it becomes the 
Hamiltonian of the 
$sl(N)$-invariant Heisenberg spin chain
$H = p\sum_{i \in {\mathbb Z}_L}
(\sigma_{i, i+1}-1)$ with $\sigma_{i,i+1}$ being the 
transposition of the local states $k_i \leftrightarrow k_{i+1}$.
In general, imaginary eigenvalues of $H$ 
form complex conjugate pairs.
In section \ref{sec:5}, Bethe ansatz 
integrability of $H$ for general $p, q$ will be demonstrated.
Although $H$ is not normal in general, we expect that 
it is diagonalizable based on 
conjecture \ref{k:con:dep} and the remark following it.
This fact will be used only in corollary \ref{k:cor:ome}.

In view of the transition rule \eqref{a:rule},
our Hamiltonian $H$ obviously preserves 
the number of particles of each kind.
It follows that the space of states $(\C^N)^{\otimes L}$
splits into a direct sum of sectors and 
$H$ has the block diagonal structure:
\begin{align}
(\C^N)^{\otimes L} &= \bigoplus_m V(m),\quad
V(m) =\bigoplus_{\{k_i\} \text{ in sector }\, m}
\C| k_1,\ldots, k_L\rangle,
\label{k:eq:m-sum}\\
H&= \bigoplus_{m}H(m),\quad
H(m) \in {\rm End }\,V(m).
\label{k:eq:hm}
\end{align}
Here the direct sums $\bigoplus_m$ extend over 
$m=(m_1,\ldots, m_N)\in {\mathbb Z}_{\ge 0}^N$ 
such that $m_1+\cdots + m_N=L$.
The array $m$ labels a sector $V(m)$, which is a 
subspace of $(\C^N)^{\otimes L}$ 
specified by the multiplicities of particles of each kind.
Namely, the latter sum in \eqref{k:eq:m-sum} is taken over all 
$(k_1,\ldots, k_L) \in \{1,\ldots, N\}^L$ such that
\begin{align}\label{k:eq:sort}
\text{Sort}(k_1,\ldots, k_L)
= \overbrace{1 . . . 1}^{m_1}
\overbrace{2 . . . 2}^{m_2}\cdots
\overbrace{N . . .  N}^{m_N},
\end{align}
where $\text{Sort}$ stands for the 
ordering non decreasing to the right.
In other words, the $\{k_i\}$
in \eqref{k:eq:m-sum} runs over 
all the permutations of the right hand side of \eqref{k:eq:sort}.
The array $m$ itself will also be called a sector.
A sector,  $m=(2,0,4,1,0,3)$ for instance,  
will also be referred in terms of the Sort sequence as
$1^23^4 46^3=1133334666$. 
By the definition $\dim V(m) = L!/\prod_{i=1}^Nm_i!$.

By now we have separated the master equation 
\eqref{a:master2} into the ones in each sector 
$\frac{d}{dt}|P(t)\rangle = H(m) |P(t)\rangle$
with $|P(t)\rangle \in V(m)$.  
Since the transition rule \eqref{a:rule}
only refers to the alternatives $\alpha > \beta$ or 
$\alpha < \beta$,
the master equation in the sector 
$1^23^4 46^3$ for example is equivalent to 
that in $a^2 b^4 c\, d^3$ 
for any $1\!\le\! a\!<\! b\!<\! c\!<\! d\! \le\! N$.

Henceforth without loss of generality 
we take $N = L$ and restrict our consideration to 
the {\em basic sectors} that have the form
$m=(m_1,\ldots, m_n,0,\ldots,0)$ for some $n$
with $m_1, \ldots, m_n$ being all positive.
The basic sectors 
are labeled with the elements of the set
\begin{align}
{\mathcal M} &= \{(m_1,\ldots, m_n) \in \Z_{\ge 1}^n\mid
1 \le n \le L, \, m_1+\cdots + m_n=L\}. \label{k:eq:M}
\end{align}
In this convention, which will be  employed in the rest of the paper 
except section \ref{sec:5}, 
$n$ plays the role of $N$ in the sense that 
$H(m)$ for 
$m = (m_1,\ldots, m_n) \in {\mathcal M}$
is equivalent to the Hamiltonian of the $(n-1)$-species ASEP.

In section \ref{sec:3} we study specific eigenvalues of $H$
that are relevant to the leading behavior of the relaxation. 
In section \ref{sec:4} we elucidate a
spectral duality, a new global aspect of the spectrum of $H$,
which has escaped a notice in earlier works mostly devoted to 
the studies of the thermodynamic limit 
$L \rightarrow \infty$ under a fixed $N$.

\section{Relaxation to the stationary state}\label{sec:3}

\subsection{General remarks}\label{k:subsec:rem}
The initial value problem of 
the master equation
$\frac{d}{dt}|P(t)\rangle = H(m) |P(t)\rangle$
with $|P(t)\rangle \in V(m)$
in the sector $m =(m_1,\ldots, m_n)
\in {\mathcal M}$ \eqref{k:eq:M}
is formally solved as 
\begin{align}\label{k:eq:sol}
|P(t)\rangle = \e^{tH(m)} |P(0)\rangle .
\end{align}
There is a unique stationary state corresponding to
the zero eigenvalue of $H(m)$.
The stationary state has a non-uniform  
probability distribution \cite{PEM} 
except for the zero-species, the one-species
and the SSEP cases, where
$P(k_1,\dots,k_L)=1/\dim V(m)$. 
All the other eigenvalues of $H(m)$ have 
strictly negative real parts, which are responsible 
for various relaxation modes to the stationary state.
(The associated eigenvectors themselves are not
physical probability vectors having non-negative components.)
Let us denote by $|P(\infty)\rangle$
the stationary state.
In general, the system exhibits the long time behavior 
\begin{align}\label{k:eq:ppt}
|P(t)\ket -  |P(\infty)\ket \sim \e^{-t/\tau}
\quad(t\rightarrow \infty),
\end{align}
where $\tau$ is the relaxation time.
An important characteristic of the non-equilibrium
dynamics is the scaling property of the relaxation time 
with respect to the system size:  
\begin{align}\label{k:eq:tauL}
\tau \sim L^z \quad (L \rightarrow \infty),
\end{align}
where $z$ is the dynamical exponent.
The thermodynamic limit $L\rightarrow \infty$ 
is to be taken under the fixed
densities $\rho_j = m_j/L$ for $j=1,\ldots, n$. 
(Recall $L=m_1+\cdots+m_n$ in \eqref{k:eq:M}, therefore
$\rho_1+\cdots + \rho_n=1$.)

The spectrum in a sector $m \in {\mathcal M}$ will be denoted by  
\begin{align}\label{k:eq:spm}
{\rm Spec}(m)=\text{multiset of eigenvalues of}\, H(m),
\end{align}
where the multiplicities counts the degrees of degeneracy.
${\rm Spec}(m)$ is invariant under complex conjugation.
The charge conjugation property \eqref{k:eq:cc} 
implies the symmetry
\begin{align}\label{k:eq:ccs}
{\rm Spec}(m_1,\ldots, m_n) 
= {\rm Spec}(m_n,\ldots, m_1).
\end{align}
We say that a complex eigenvalue $x$ of 
$H(m)$  is larger (smaller) than 
$y$ if ${\rm Re}(x) > {\rm Re}(y)$
(${\rm Re}(x) < {\rm Re}(y)$). 
${\rm Spec}(m)$ contains $0$  
as the unique largest eigenvalue.
For a finite $L$, we say an eigenvalue $E$ of $H(m)$ is 
{\em second largest} if 
\begin{align}
{\rm Re}\,E =
\max {\rm Re}({\rm Spec}(m)\setminus \{0\}).
\end{align}
{}From (\ref{k:eq:sol}) and  \eqref{k:eq:ppt}, 
the scaling property \eqref{k:eq:tauL} is equivalent to 
the following behavior of the second largest eigenvalues
\begin{align}\label{k:eq:mrs}
\max {\rm Re}({\rm Spec}(m)\setminus \{0\})
= -c L^{-z}+ o(L^{-z})\quad (L \rightarrow \infty)
\end{align}
if the initial condition $|P(0)\rangle$ is generic.
Here $c>0$ is an ``amplitude" which can 
depend on $\rho_1,\ldots, \rho_n$ in general, 
but not on $L$.

In the remainder of this section, we 
derive the exponent $z$ of the multi-species ASEP based on  
\eqref{k:eq:mrs}.
Our argument reduces the problem essentially to the 
one-species case and is partly based on a conjecture 
supported by numerical analyses.

\subsection{Known results on the one-species ASEP}
\label{a:1sASEP}
In this subsection we review 
the known results on the one-species ASEP.
Thus we shall exclusively consider the sector of the form 
$m=(m_1,m_2)\in {\mathcal M}$, and 
regard the local states $1$ and $2$ as vacancies 
and the particles of one kind, respectively. 
Recall also that $L=m_1+m_2$.

The second largest eigenvalues 
are known to form a complex-conjugate pair,
which will be denoted by $E^{\pm}(m)$.
See figure \ref{plot_34}.
When $m_1=m_2$ or $p=q$, the 
degeneracy $E^+(m)=E^-(m) \in {\mathbb R}$ occurs.

\begin{figure}[h]
\begin{center}
 \includegraphics{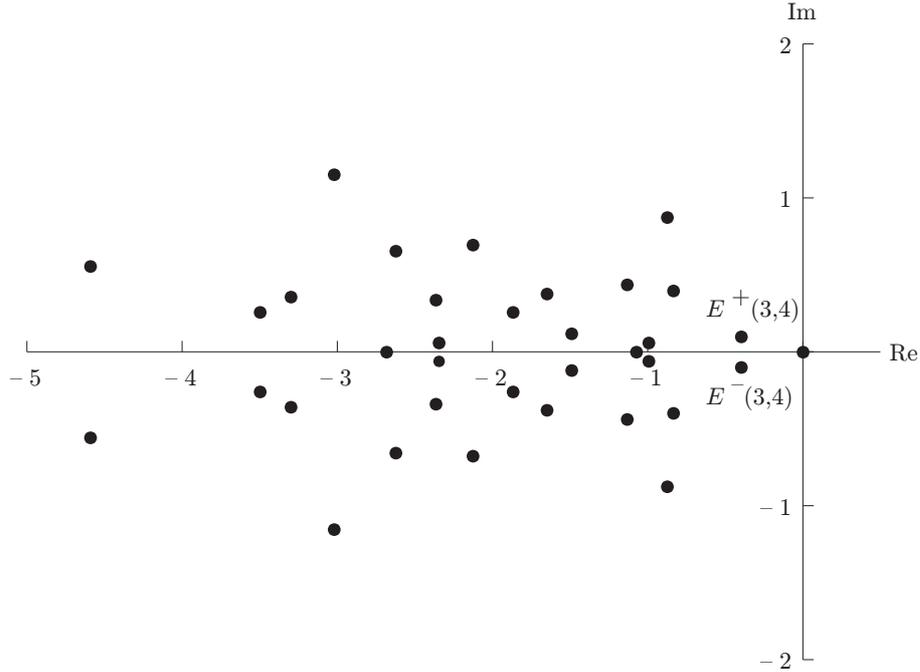}
 \caption{${\rm Spec}(3,4)$ for $(p,q)=(0.8,0.2).$}
 \label{plot_34}
\end{center}
\end{figure}

In \cite{GS,K,GM}, 
the large $L$ asymptotic form 
\begin{align}\label{a:1KPZ}
 E^{\pm} ( (1-\rho)L,\rho L )
= \pm 2\pi \i|(p-q)(1-2\rho)| L^{-1} 
 - 2C|p-q|\sqrt{\rho(1-\rho)}L^{-\frac{3}{2}}
 + O (L^{-2})
\end{align}
with a fixed particle density $\rho=m_2/L$
was derived for $p\neq q$
by an analysis of the Bethe equation.
The two terms are both invariant under 
$\rho \leftrightarrow 1-\rho$.
The constant $C$ has been numerically evaluated as
 $C = 6.50918933794\ldots$ .
Thus from \eqref{k:eq:mrs} 
the one-species ASEP for $p\neq q$
has the dynamical exponent $z=\frac{3}{2}$,
which is a characteristic value for the 
Kardar-Parisi-Zhang universality class
\cite{KPZ}. 

In the SSEP case $p=q$,
the Hamiltonian $H(m_1,m_2)$ is Hermitian,
hence all the eigenvalues are real.
The system relaxes
to the equilibrium stationary state.
For a finite $L$, 
the second largest eigenvalues 
take the simple form
\begin{align}
 E^+(m_1,m_2) = E^-(m_1,m_2)
= - 4p\sin^2(\frac{\pi}{L})\quad
(0 < m_2 < L),
\end{align}
which is 
independent of the density $\rho = m_2/L$ as long as 
$0 < \rho < 1$.
The asymptotic behavior in $L \to \infty$ is
easily determined as
\begin{align}\label{a:1EW}
 E^{\pm}((1-\rho)L, \rho L)
=  - 4\pi ^2pL^{-2} + O(L^{-4}),
\end{align}
which is free from a contribution 
of order $L^{-\frac{3}{2}}$.
From \eqref{a:1EW}, we find the dynamical exponent
$z=2$, which is the characteristic value for
the Edwards-Wilkinson universality class
\cite{EW}.

\subsection{\mathversion{bold} 
Eigenvalues of $H(2,1,3,1)$ -an example}

Let us proceed to the multi-species ASEP.
Before considering a general sector in the next subsection,
we illustrate characteristic features of 
the spectrum along an example. 
Figure \ref{plot_2131} (a) is a plot (black dots) 
of the spectrum ${\rm Spec}(2,1,3,1)$ on the complex plane.
We recall that the sector $(2,1,3,1)$ means  
the ring of length 7 populated with 4 kinds of particles
with multiplicities $2,1,3$ and $1$, among which 
the first kind ones are regarded as vacancies.

For comparison, we have also included 
the plot of the spectra in the one-species sectors 
$(2,5)$, $(3,4)$ and $(6,1)$ in different colors and shapes.
These one-species sectors are related to the 
multi-species sector $(2,1,3,1)$ as follows.
The sector $(2,5)=(2,1+3+1)$ is obtained by
identification of all kinds of particles
(except for vacancies) as one kind of particles.
The sector $(3,4)=(2+1,3+1)$ 
is obtained by identification of the 
second kind particles as vacancies
and the rest of particles as one kind of particles.
The sector $(6,1)=(2+1+3,1)$ 
is obtained by identification of
the second and the third kinds of particles as vacancies.
In figure \ref{plot_2131} (a), 
we observe that all the colored dots overlap  
the black dots.
Namely, 
${\rm Spec(2,5)},\,{\rm Spec(3,4)},\,{\rm Spec(6,1)}$
are totally embedded into ${\rm Spec}(2,1,3,1)$.

Figure \ref{plot_2131} (b) 
shows that the second largest eigenvalues 
(denoted by $E^\pm_j(2,1,3,1)$) 
in those one-species sectors
form a string within ${\rm Spec}(2,1,3,1)$ near the origin.  
Although their real parts are not strictly the same, 
there is no black dot between
the string and the origin.
More precisely, there is no 
eigenvalue in the sector $(2,1,3,1)$ which is nonzero and 
larger than 
any second largest eigenvalues in the one-species sectors
$(2,5)$, $(3,4)$ and $(6,1)$.
This property is a key to our argument in the sequel.

\begin{figure}
\begin{center}
 \includegraphics{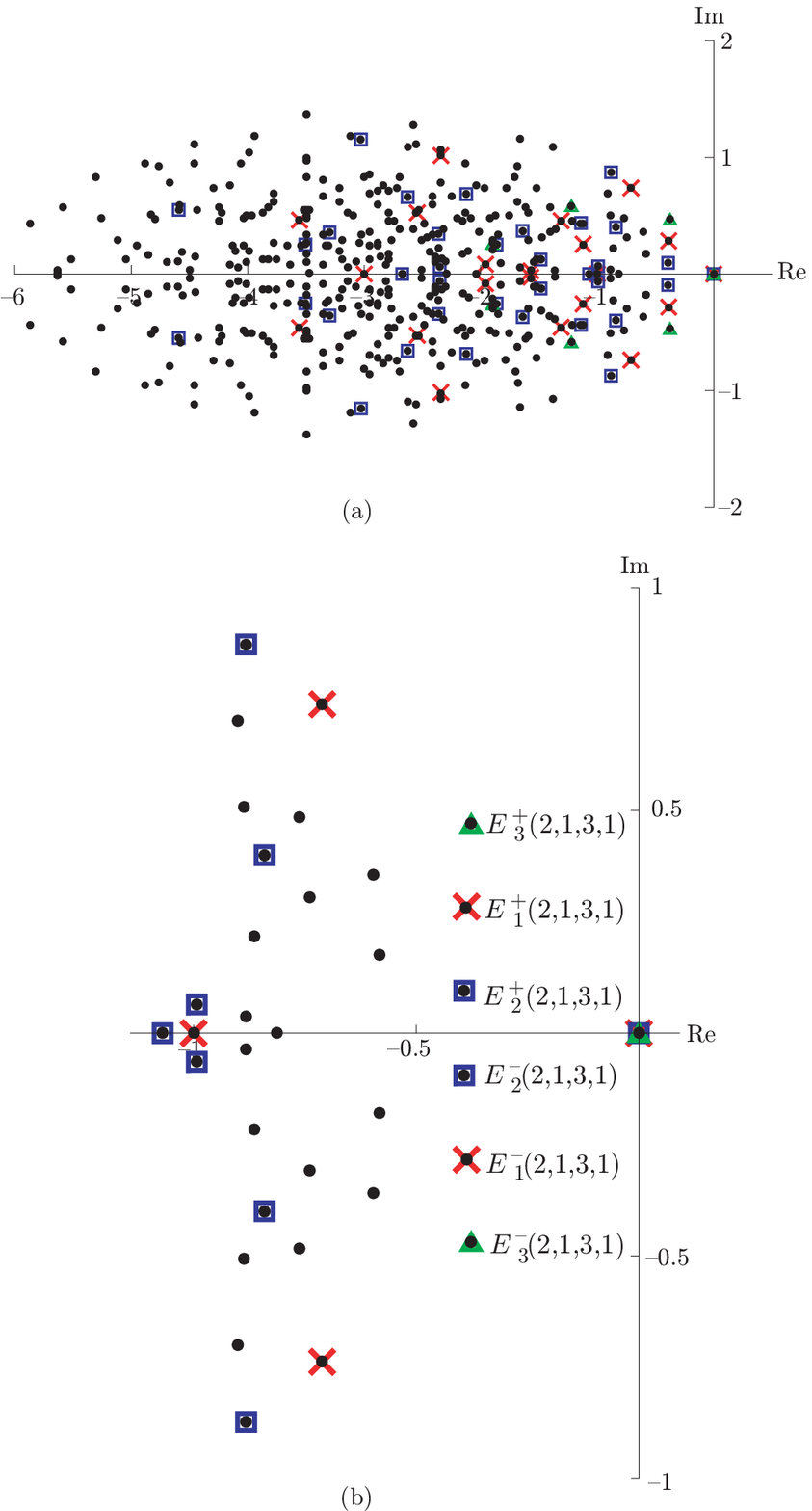}
 \caption{(a) ${\rm Spec}(2,1,3,1)$ (black dot),
${\rm Spec}(2,5)$ (red $\times$), 
${\rm Spec}(3,4)$ (blue square)
and ${\rm Spec}(6,1)$ (green triangle)
with $(p, q)=(0.8, 0.2)$.
(b) An enlarged view near the origin.}
 \label{plot_2131}
\end{center}
\end{figure}

\subsection{Eigenvalues of multi-species ASEP}

Let us systematize the observations made in the previous subsection.
First we claim that the following inclusion relation holds generally:
\begin{align}\label{a:increl}
 {\rm Spec}(m_1,\dots ,m_n)
 \supset {\rm Spec}(m_1+\cdots+m_j, m_{j+1}+\cdots+m_n)
\quad (1\le j \le n-1).
\end{align}
Each one-species sector appearing in the
right-hand side of \eqref{a:increl} 
is obtained by the identification
similar to the previous subsection:
\begin{align*}
\underbrace{
\overbrace{\circ \cdots \circ}^{m_1}|\cdots|
\overbrace{\circ \cdots \circ}^{m_j}}_{m_1+\cdots+m_j}|
\underbrace{
\overbrace{\circ \cdots \circ}^{m_{j+1}}|\cdots|
\overbrace{\circ \cdots \circ}^{m_n}}_{m_{j+1}+\cdots+m_n} .
\end{align*}
The relation \eqref{a:increl} is a special case 
of the more general statement in theorem \ref{k:cor:spe}.
See also section \ref{sec:5-2-2} for an account 
from the nested Bethe ansatz.

Next we introduce a class of eigenvalues of $H(m)$ for 
a multi-species sector
$m =(m_1,\ldots, m_n) \in {\mathcal M}$ by
\begin{align}\label{k:eq:epmj}
E^{\pm}_j(m)=E^{\pm}(m_1+\cdots+m_j, m_{j+1}+\cdots+m_n)
\quad (1\le j \le n-1),
\end{align} 
where $E^\pm$ in the right hand side 
are the second largest eigenvalues 
in the one-species sector introduced in section \ref{a:1sASEP}.
In view of \eqref{a:increl}, we know 
$E^\pm_j(m) \in {\rm Spec}(m)$.
Note that ${\rm Re}(E^+_j(m))={\rm Re}(E^-_j(m))$,
but the subscript $j$ does not necessarily reflect 
the ordering of the eigenvalues 
with respect to their real parts.
Generalizing the previous 
observation on figure \ref{plot_2131},
we make 

\vspace{0.3cm}\noindent
\begin{conjecture}\label{k:con:sp}
In any sector $m=(m_1,\dots,m_n) \in {\mathcal M}$,
there is no eigenvalue $E \in {\rm Spec}(m)$ such that
\begin{align}
\max\{{\rm Re}\,E^{\pm}_1(m),\ldots, 
{\rm Re}\,E^{\pm}_{n-1}(m)\} < {\rm Re}\,E < 0.
\end{align}
\end{conjecture}

\vspace{0.3cm}
The one-species case $n=2$ is trivially true by the definition.
So far the conjecture has been checked in all the sectors $m$
satisfying $\dim V(m) <8000$.

Admitting the conjecture, we are able to claim that 
the second largest eigenvalues in 
${\rm Spec}(m)$ are equal to 
$E^\pm_j(m)$ for some $1 \le j \le n-1$.
(Such $j$ may not be unique.)
The asymptotic behavior of $E^\pm_j(m)$ is 
derived from  \eqref{a:1KPZ} and \eqref{k:eq:epmj} as
\begin{align}\label{a:KPZ}
\begin{split}
&E^\pm_j (\rho_1 L,\dots,\rho_n L) \\
&=
 \pm 2\pi \i|(p-q)(1-2r_j)| L^{-1} 
 - 2C|p-q| \sqrt{r_j(1-r_j)}
 L^{-\frac{3}{2}} +  O (L^{-2})
\end{split}
\end{align}
for $p\neq q$,  where $r_j = \rho_1+\cdots+\rho_j$ is fixed.
We remark that the leading terms in \eqref{a:KPZ}
depend on $j$ only through $r_j$ in the amplitudes.
We call the eigenvalues  
$E^\pm_1(m),\ldots, E^\pm_{n-1}(m)$ 
{\em next leading}. 
Thus the second largest eigenvalues are next leading.
All the next leading eigenvalues possess 
the same asymptotic behavior 
as the second largest ones up to the amplitudes
as far as the first 2 leading terms in \eqref{a:KPZ}
are concerned. 

With regard to the SSEP case $p=q$,
the stationary state is an equilibrium state.
$H(m)$ is Hermitian and  $E^{\pm}_j(m)$ is real.
We have the following explicit form
as in the one-species case:
\begin{align}
E_1^{\pm}(m)= \cdots = E_{n-1}^{\pm}(m)
= - 4p\sin^2(\frac{\pi}{L}).
\end{align}
In other words,
the next leading eigenvalues $E_j^{\pm}(m)$
are degenerated in the SSEP limit $p-q\to 0$.
See figure \ref{degen},
where the string of the next leading eigenvalues 
shrinks to a point on the real axis as $p-q$ approaches $0$.

\begin{figure}[h]
\begin{center}
\includegraphics{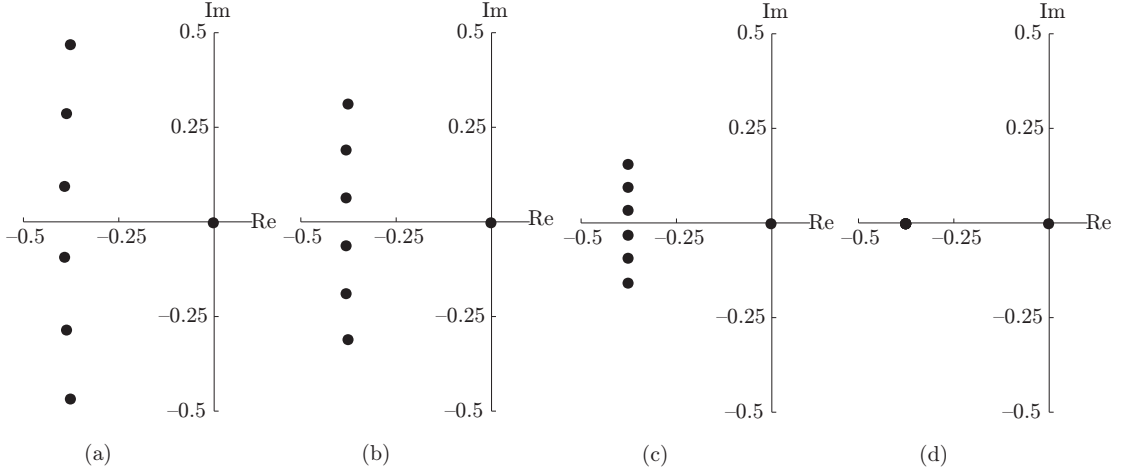}
\label{degen}
\caption{Degeneracy in the SSEP limit $p-q\to 0$ 
in the sector $(2,1,3,1)$.
$(p,q)$ is taken as 
(a) $(0.8,0.2)$, (b) $(0.7,0.3)$, (c) $(0.6,0.4)$
and (d) $(0.5, 0.5)$.}
\end{center}
\end{figure}
As the one-species case \eqref{a:1EW}, we find
\begin{align}\label{a:EW}
E^{\pm}_j(\rho_1 L,\dots,\rho_n L)
= - 4\pi ^2pL^{-2} + O(L^{-4})\ .
\end{align}

To summarize, the results \eqref{a:KPZ} and \eqref{a:EW} 
lead to the following behavior of the relaxation time $\tau$:
\begin{align}\label{k:eq:tl}
\tau \sim
\begin{cases}
 L^{\frac{3}{2}} & \text{for }\; p\neq q, \\
 L^{2}&  \text{for }\; p=q,
\end{cases}\quad (L \rightarrow \infty).
\end{align}
Therefore we conclude that 
the dynamical exponent of 
the multi-species ASEP is  
independent of the number of species.
It belongs to 
the KPZ universality class ($z=\frac{3}{2}$) 
for $p\neq q$, 
and to the EW universality class ($z=2$) for $p=q$.

We leave it as a future study to investigate the
gap between the next leading eigenvalues 
and further smaller eigenvalues, which governs 
the pre-asymptotic behavior of the multi-species ASEP.

\section{Duality in Spectrum}\label{sec:4}

Throughout this section, 
a sector means a basic sector as promised in 
section \ref{k:subsec:bph}.
We fix the number of sites in the ring 
$L \in \Z_{\ge 2}$.
Our goal is to prove 
theorem \ref{k:th:sd}, which exhibits 
a duality in the spectrum of Hamiltonian.

\subsection{Another label of sectors}\label{subsec:sector}

Set
\begin{align}
\Omega &=\{1,2,\ldots, L-1\},\label{k:eq:ome}\\
{\mathcal S}&= \text{the power set of } \Omega.
\label{k:eq:I}
\end{align}
Recall that the sectors in the length $L$ chain
are labeled with the set ${\mathcal M}$ \eqref{k:eq:M}.
We identify ${\mathcal M}$ with 
${\mathcal S}$ by the one to one correspondence:
\begin{align}\label{k:eq:bi}
{\mathcal M} \ni m=(m_1,\ldots, m_n) \longleftrightarrow 
\{s_1 < \cdots < s_{n-1}\}={\mathfrak s} \in {\mathcal S}
\end{align}
specified via $s_j = m_1+ m_2 + \cdots + m_j$, namely,
\begin{align*}
\underbrace{
\underbrace{
\underbrace{
\overbrace{\circ \cdots \circ}^{m_1}}_{s_1}\vert 
\overbrace{\circ \cdots \circ}^{m_2}}_{s_2}\vert\cdots }_{s_{n-1}}
\vert\overbrace{\circ \cdots \circ}^{m_n},
\end{align*}
where the numbers of the symbols $\circ$ and $\vert$ are 
$L$ and $n-1$, respectively.
For example the identification
${\mathcal M} \leftrightarrow {\mathcal S}$ 
for $L=4$ is given as follows:
\begin{center}
$(1,1,1,1)\leftrightarrow \{1,2,3\}$\\
\vspace{0.2cm}
$(1,1,2)\leftrightarrow \{1,2\}
\quad (1,2,1)\leftrightarrow \{1,3\}\quad 
(2,1,1)\leftrightarrow \{2,3\}$\\
\vspace{0.2cm}
$(1,3)\leftrightarrow \{1\}\quad (2,2)\leftrightarrow\{2\} 
\quad (3,1)\leftrightarrow\{3\}$\\
\vspace{0.2cm}
$(4) \leftrightarrow\emptyset$.
\end{center}

An element of ${\mathcal S}$ will also be called a sector.
In the remainder of this section we will mostly work with 
the label ${\mathcal S}$ instead of ${\mathcal M}$.
There are $\sharp{\mathcal S} = 2^{L-1}$ distinct sectors.
We employ the notation:
\begin{align}\label{k:eq:cmp}
\overline{\mathfrak{s}} = 
\Omega \setminus \mathfrak{s} 
=\text{complement sector of }\, \mathfrak{s}.
\end{align}
For a sector 
$\mathfrak{s}=\{s_1<\cdots <s_{n-1}\} \in {\mathcal S}$, 
we introduce the set ${\mathcal P}(\mathfrak{s})$ by
(see \eqref{k:eq:sort})
\begin{align}\label{k:eq:pdef}
{\mathcal P}(\mathfrak{s}) = \{k=(k_1,\ldots, k_L)\mid
\text{Sort}(k)
= \overbrace{1 . . . 1}^{s_1}
\overbrace{2 . . . 2}^{s_2-s_1}\cdots
\overbrace{n\!-\!1 . . .  n\!-\!1}^{s_{n-1}-s_{n-2}}\!
\overbrace{n . . .  n}^{L-s_{n-1}}\},
\end{align}
where Sort stands for the ordering non decreasing to the right as
in (\ref{k:eq:sort}).
For $\mathfrak{s}=\emptyset \in {\mathcal S}$, this definition should 
be understood as 
${\mathcal P}(\emptyset)=\{(1,\ldots, 1)\}$.

To each sector we associate the bra and ket vector spaces 
\begin{align}
V_\mathfrak{s}^\ast = \bigoplus_{k \in {\mathcal P}(\mathfrak{s})}
\C \langle k_1,\ldots, k_L\vert,
\quad
V_\mathfrak{s} = \bigoplus_{k \in {\mathcal P}(\mathfrak{s})}
\C \vert k_1,\ldots, k_L\rangle.
\end{align}
Here $k_1, \ldots, k_L \in \{1,\ldots, L\}$ stand for local states.
For example if $L=3$, one has
\begin{align*}
V_\emptyset &= \C \vert 111\rangle,\\
V_{\{1\}}&= \C\vert 122\rangle \oplus 
\C\vert 212\rangle \oplus \C\vert 221\rangle,\\
V_{\{2\}}&=\C\vert 112\rangle \oplus \C\vert 121\rangle
\oplus \C\vert 211\rangle,\\
V_{\scriptscriptstyle \Omega} =V_{\{1,2\}}&= 
\C\vert 123\rangle \oplus \C\vert 132\rangle \oplus
\C\vert 213\rangle \oplus \C\vert 231\rangle \oplus
\C\vert 312\rangle \oplus \C\vert 321\rangle.
\end{align*}
Note that the vectors like $\vert 2 2 2 \rangle$ 
and $\vert 113 \rangle$ are {\em not} included 
in any $V_\mathfrak{s}$ because 
we are concerned with basic sectors only.
See (\ref{k:eq:pdef}).  
In general, one has
\begin{align}\label{k:eq:ddv}
\dim V_\mathfrak{s} = \dim V_\mathfrak{s}^\ast 
= \frac{L!}{s_1!(s_2\!-\!s_1)! \cdots(s_{n-1}\!-\!s_{n-2})!
(L\!-\!s_{n-1})!}
\end{align}
for $\mathfrak{s}=\{s_1<\cdots <s_{n-1}\} \in {\mathcal S}$.

Suppose ${\mathcal M}\ni m 
\leftrightarrow {\mathfrak s} \in {\mathcal S}$ 
under the correspondence (\ref{k:eq:bi}).
We renew the symbols $H(m)$ and $V(m)$ 
in \eqref{k:eq:m-sum}--\eqref{k:eq:hm} as
\begin{align}
V_{\mathfrak s} = V(m),\quad 
H_{\mathfrak s}=H(m).\label{k:eq:hh}
\end{align}

The set ${\mathcal S}$  
is equipped with the natural poset
(partially ordered set) structure 
with respect to $\subseteq$.
The poset structure is encoded 
in the Hasse diagram \cite{St}, which is useful 
in our working below. 
In the present case, it is just 
the $L\!-\!1$ dimensional hypercube,  
where each vertex corresponds to a sector.
Sectors are so arranged that every edge of the hypercube 
becomes an arrow $\mathfrak{s} \rightarrow \mathfrak{t}$ 
meaning that $\mathfrak{s} \subset \mathfrak{t}$ and 
$\sharp \mathfrak{t} = \sharp \mathfrak{s} + 1$.  
There is the unique sink corresponding to the maximal sector 
$\Omega \in {\mathcal S}$ and the unique source 
corresponding to the minimal sector 
$\emptyset \in {\mathcal S}$.
See figure \ref{Hasse}.

\begin{figure}[h]
\begin{center}
\includegraphics{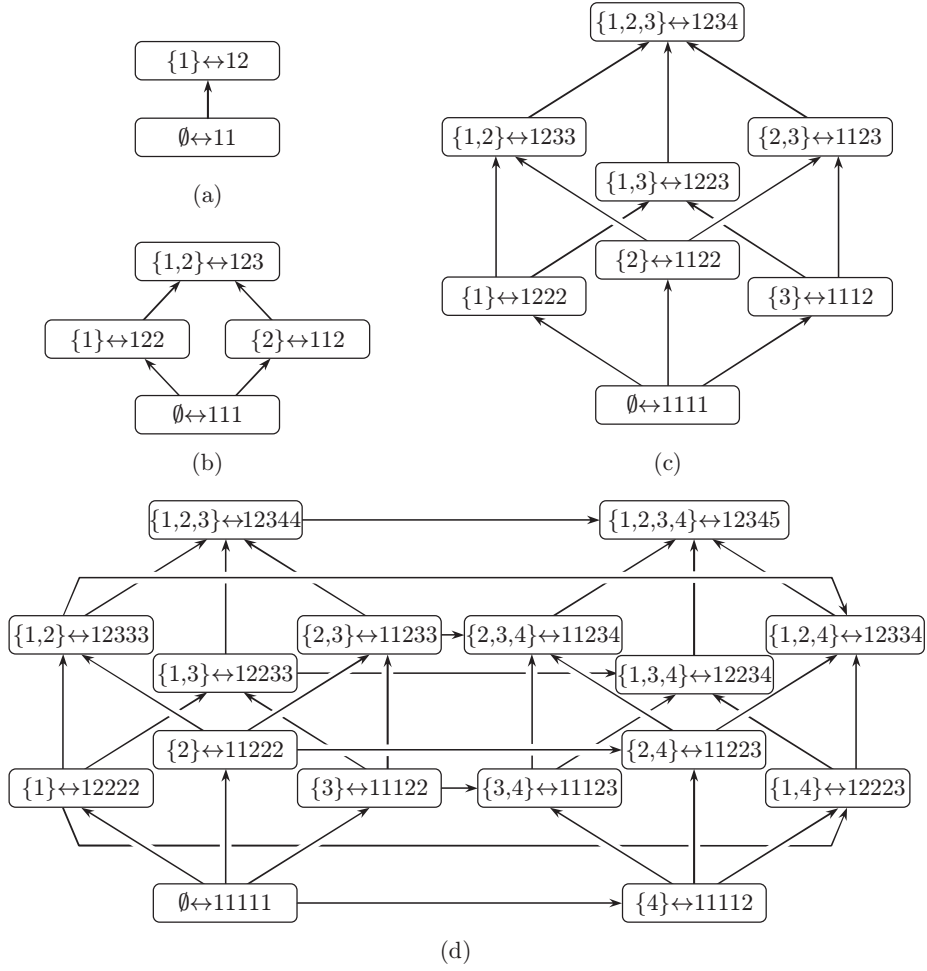}
 \caption{Hasse diagrams 
(a), (b), (c) and (d) for $L=2,3,4$ and $5$, respectively.
Sectors are labeled by ${\mathcal S}$ (\ref{k:eq:I}) 
as well as the sequence ${\rm Sort}(k)$ 
of local states as in (\ref{k:eq:pdef}).}
 \label{Hasse}
\end{center}
\end{figure}

We introduce the natural bilinear pairing between the bra and ket vectors by 
\begin{align}\label{k:eq:bp}
\langle k_1,\ldots, k_L \vert j_1,\ldots, j_L\rangle = 
\delta_{k_1, j_1}\cdots \delta_{k_L, j_L}.
\end{align}
With respect to the pairing, 
$V_\mathfrak{s}^\ast$ and $V_\mathfrak{t}$ 
are dual if $\mathfrak{s}=\mathfrak{t}$ and orthogonal 
if $\mathfrak{s}\neq \mathfrak{t}$. 

Any linear operator $\overset{\rightarrow}{G}$ acting on ket vectors 
give rise to the unique linear operator 
$\overset{\leftarrow}{G}$ acting on bra vectors via
$\left(\langle k_1,\ldots, k_L \vert \overset{\leftarrow}{G}\right)
\vert j_1,\ldots, j_L\rangle
= \langle k_1,\ldots, k_L \vert \left(\overset{\rightarrow}{G}
\vert j_1,\ldots, j_L\rangle\right)$
and vice versa.
We write this quantity simply as 
$\langle k_1,\ldots, k_L \vert G \vert j_1,\ldots, j_L\rangle$ as usual,
and omit $\overset{\leftarrow}{\phantom{e}}$ and 
$\overset{\rightarrow}{\phantom{e}}$ unless 
an emphasis is preferable.
The transpose $G^{\rm T}$ of $G$ is defined by
$\langle k \vert G^{\rm T} \vert j \rangle = 
\langle j \vert G \vert k \rangle \;(=:G_{jk})$
for any $k=(k_1,\ldots, k_L)$ and $j=(j_1,\ldots, j_L)$.
Of course $\overset{\rightarrow}{G}{}^{\rm T}$ 
is equivalent to $\overset{\leftarrow}{G}$ in the sense that
\begin{align}\label{k:eq:gt}
G^{\rm T}\vert k \rangle = \sum_j G_{k j}\vert j \rangle,\quad
\langle k \vert G = \sum_j G_{k j}\langle j \vert.
\end{align}

\subsection{\mathversion{bold} 
Operator $\varphi_{\mathfrak{s},\mathfrak{t}}$}
\label{subsec:vphi}

Let $\mathfrak{s}, \mathfrak{t} \in {\mathcal S}$ 
be sectors such that $\mathfrak{s} \subseteq \mathfrak{t}$.
We introduce a $\C$-linear operator $\varphi_{\mathfrak{s,t}}$ 
in terms of its action on ket vectors 
$\overset{\rightarrow}{\varphi}_{\mathfrak{s},\mathfrak{t}}:
V_{\mathfrak{t}} \rightarrow V_{\mathfrak{s}}$.
We define $\varphi_{\mathfrak{s}, \mathfrak{s}}$ 
to be the identity operator for any sector $\mathfrak{s}$.
Before giving the general definition of the case 
$\mathfrak{s} \subset \mathfrak{t}$,
we illustrate it with the example
$\mathfrak{s}=\{2,5\} \subset \mathfrak{t}=\{2,3,5,8\}$ 
with $L=9$.
The Sort sequence of the local states 
in the sense of (\ref{k:eq:pdef})
for ${\mathcal P}(\mathfrak{t})$ 
and ${\mathcal P}(\mathfrak{s})$ read as follows:
\begin{align}\label{k:eq:ppex}
\begin{split}
{\mathcal P}(\mathfrak{t}=\{2,3,5,8\}): 
&\quad11\overset{2}{\vert}2\overset{3}{\vert}33
\overset{5}{\vert}444\overset{8}{\vert}5,\\
{\mathcal P}(\mathfrak{s}=\{2,5\}):
&\quad11\overset{2}{\vert}2\, 2\, 2 \overset{5}{\vert}
3\,3\,3\,3.
\end{split}
\end{align}
According to these lists, we define 
${\varphi}_{\mathfrak{s},\mathfrak{t}}$ to be the operator 
replacing the local states as
$3 \rightarrow 2, 4 \rightarrow 3, 5 \rightarrow 3$ 
(keeping $1$ and $2$ unchanged)
within all the ket vectors $\vert k_1,\ldots, k_L\rangle$ in 
$V_\mathfrak{t}$.

General definition of 
${\varphi}_{\mathfrak{s},\mathfrak{t}}$ 
is similar and goes as follows.
Suppose $\mathfrak{t}=\{t_1<\cdots < t_n\}$ and 
$\mathfrak{s} = \mathfrak{t}\setminus\{t_{i_1}, \ldots, t_{i_l}\}$.
Then ${\varphi}_{\mathfrak{s},\mathfrak{t}}$ is 
a $\C$-linear operator determined by its action on 
base vectors as follows: 
\begin{align}\label{k:eq:pvv}
\begin{array}{cccc}
\overset{\rightarrow}{\varphi}_{\mathfrak{s}, \mathfrak{t}}:
& V_\mathfrak{t}               & \longrightarrow & V_\mathfrak{s}\\
& \vert k_1,\ldots, k_L\rangle & \mapsto & \vert k'_1, \ldots, k'_L\rangle,
\end{array}
\end{align}
where $x' = x-\sharp\{i_j \mid i_j <x\}$.

\begin{example}\label{k:ex:vp}
\begin{align*}
\vert \phi\rangle :&= \vert 21433\rangle - \vert 12343\rangle 
\in V_{\mathfrak{t}} \;\;\;(\mathfrak{t}=\{1,2,4\}, \;L=5),\\
{\varphi}_{12, 124}\vert \phi\rangle
&=\vert 21333\rangle - \vert 12333\rangle,\;\;
\begin{cases}{\varphi}_{1, 12}\,
{\varphi}_{12, 124}\vert \phi\rangle
=\vert 21222\rangle - \vert 12222\rangle,\\
{\varphi}_{2, 12}\,
{\varphi}_{12, 124}\vert \phi\rangle
=0,
\end{cases}\\
{\varphi}_{14, 124}\vert \phi\rangle
&=\vert 21322\rangle - \vert 12232\rangle,\;\;
\begin{cases}{\varphi}_{1, 14}\,
{\varphi}_{14, 124}\vert \phi\rangle
=\vert 21222\rangle - \vert 12222\rangle,\\
{\varphi}_{4, 14}\,
{\varphi}_{14, 124}\vert \phi\rangle
=\vert 11211\rangle - \vert 11121\rangle,
\end{cases}\\
{\varphi}_{24, 124}\vert \phi\rangle
&=\vert 11322\rangle - \vert 11232\rangle,\;\;
\begin{cases}{\varphi}_{2, 24}\,
{\varphi}_{24, 124}\vert \phi\rangle
=0,\\
{\varphi}_{4, 24}\,
{\varphi}_{24, 124}\vert \phi\rangle
=\vert 11211\rangle - \vert 11121\rangle,
\end{cases}\\
{\varphi}_{1, 124}\vert \phi\rangle
&=\vert 21222\rangle - \vert 12222\rangle,\\
{\varphi}_{4, 124}\vert \phi\rangle
&=\vert 11211\rangle - \vert 11121\rangle,\\
{\varphi}_{2, 124}\vert \phi\rangle
&=0,\quad
{\varphi}_{\emptyset, 124}\vert \phi\rangle =0,
\end{align*}
where ${\varphi}_{12, 124}$
is an abbreviation of 
${\varphi}_{\{1,2\}, \{1,2,4\}}$,  etc.
\end{example}

The following property 
of ${\varphi}_{\mathfrak{s},\mathfrak{t}}$
is a direct consequence of the definition.
\begin{lemma}\label{k:le:fac}
For a pair of sectors $\mathfrak{s} \subset \mathfrak{t}$, 
let $\mathfrak{s}_0 \subset \mathfrak{s}_1 
\subset \cdots \subset \mathfrak{s}_l$
be any sectors such that 
$\mathfrak{s}_0 = \mathfrak{s}, \mathfrak{s}_l = \mathfrak{t}$ and 
$\sharp \mathfrak{s}_{j+1}=\sharp \mathfrak{s}_j + 1$ 
for all $0 \le j <l$.
Then,  
\begin{align*}
{\varphi}_{\mathfrak{s},\mathfrak{t}}
={\varphi}_{\mathfrak{s}, \mathfrak{s}_1}\,
{\varphi}_{\mathfrak{s}_1, \mathfrak{s}_2}\cdots
{\varphi}_{\mathfrak{s}_{l-1},\mathfrak{t}}.
\end{align*}
In particular, the composition in the right hand side is 
independent of the choice of the intermediate 
sectors $\mathfrak{s}_1, \ldots, \mathfrak{s}_{l-1}$.
\end{lemma}
In example \ref{k:ex:vp}, one can observe, for instance, 
${\varphi}_{4, 124}\vert \phi\rangle
={\varphi}_{4, 14}\,{\varphi}_{14, 124}\vert \phi\rangle
={\varphi}_{4, 24}\,{\varphi}_{24, 124}\vert \phi\rangle$.

Let us turn to the transpose 
$\varphi_{\mathfrak{s},\mathfrak{t}}^{\rm T}$.
By the definition (see (\ref{k:eq:gt})), we have
\begin{align}\label{k:eq:pvt}
\begin{array}{cccc}
\overset{\rightarrow}{\varphi}_{\mathfrak{s}, \mathfrak{t}}\!\!\!\!{}^{\rm T}:
 & V_\mathfrak{s} & \longrightarrow & V_\mathfrak{t}\\
 & \vert k_1,\ldots, k_L\rangle & \mapsto & 
\displaystyle
\sum_{j \in {\mathcal P}(\mathfrak{t})}\!\!{}^\prime \;
\vert j_1, \ldots, j_L\rangle,
\end{array}
\end{align}
where $\Sigma^\prime$ extends over 
those $j=(j_1,\ldots, j_L)\in {\mathcal P}(\mathfrak{t})$ such that 
${\varphi}_{\mathfrak{s},\mathfrak{t}}\vert j_1, \ldots, j_L\rangle
= \vert k_1,\ldots, k_L\rangle$.
For example in example \ref{k:ex:vp}, one has
\begin{align*}
\varphi^{\rm T}_{14, 124}\vert 21322\rangle
=\vert 21433\rangle + \vert 31423\rangle + \vert 31432\rangle.
\end{align*}
{}From (\ref{k:eq:pvt}) and (\ref{k:eq:pvv}) it follows that
$\varphi_{\mathfrak{s}, \mathfrak{t}}\,
\varphi^{\rm T}_{\mathfrak{s}, \mathfrak{t}} = 
\frac{\dim V_\mathfrak{t}}{\dim V_\mathfrak{s}}\,{\rm Id}$,
which actually means
\begin{align}\label{k:eq:pp}
\overset{\rightarrow}{\varphi}_{\mathfrak{s}, \mathfrak{t}}\,
\overset{\rightarrow}{\varphi}_{\mathfrak{s}, \mathfrak{t}}
\!\!\!\!{}^{\rm T}\,\,\,\,\, 
= \frac{\dim V_\mathfrak{t}}
{\dim V_\mathfrak{s}}\,{\rm Id}_{V_\mathfrak{s}},\quad
\overset{\leftarrow}{\varphi}_{\mathfrak{s}, \mathfrak{t}}\,
\overset{\leftarrow}{\varphi}_{\mathfrak{s}, \mathfrak{t}}
\!\!\!\!{}^{\rm T}\,\,\,\,\, 
= \frac{\dim V_\mathfrak{t}}{\dim V_\mathfrak{s}}
\,{\rm Id}_{V_\mathfrak{s}^\ast}
\end{align}
for any sectors $\mathfrak{s} \subset \mathfrak{t}$.
As a result, we obtain
\begin{lemma}\label{k:le:si}
Let $\mathfrak{s} \subset \mathfrak{t}$ be any sectors.

(1)\; $\overset{\rightarrow}
{\varphi}_{\mathfrak{s}, \mathfrak{t}}: 
V_\mathfrak{t} \rightarrow V_\mathfrak{s}$ is 
surjective.

(2)\; $\overset{\leftarrow}{\varphi}_{\mathfrak{s}, \mathfrak{t}}:
V_\mathfrak{s}^\ast \rightarrow V_\mathfrak{t}^\ast$ is 
injective.
\end{lemma}

The kernel of 
$\overset{\rightarrow}{\varphi}_{\mathfrak{s}, \mathfrak{t}}$ and 
the cokernel of 
$\overset{\leftarrow}{\varphi}_{\mathfrak{s}, \mathfrak{t}}$
will be the key in our 
derivation of the 
spectral duality in section \ref{subsec:genuine}.

By now it should be clear that 
$\overset{\rightarrow}
{\varphi}_{\mathfrak{t}\setminus \{n\}, \mathfrak{t}}$
kills ket vectors in a sector $\mathfrak{t}$ 
or send them to the neighboring smaller sector 
$\mathfrak{t}\setminus \{n\}$ 
in the Hasse diagram against one of the arrows.
Similarly, 
$\overset{\leftarrow}
{\varphi}_{\mathfrak{s}, \mathfrak{s}\cup \{n\}}$
never kills bra vectors in a sector $\mathfrak{s}$ and send them 
to the neighboring larger sector
$\mathfrak{s}\cup \{n\}$ in the Hasse diagram 
along one of the arrows.

\subsection{\mathversion{bold}
Commutativity of $\varphi_{\mathfrak{s}, \mathfrak{t}}$ and 
Hamiltonian}
\label{subsec:ph}

The action 
$H_\mathfrak{s}: V_\mathfrak{s} \rightarrow V_\mathfrak{s}$ 
of our Hamiltonian (\ref{k:eq:hh}) is 
specified by \eqref{k:hdef}--\eqref{k:eq:H} as
\begin{align}
H_\mathfrak{s} \vert k_1,\ldots, k_L\rangle = 
\sum_{i \in {\mathbb Z}_L}\Theta(k_i-k_{i+1})
\bigl(\vert k_1\ldots  k_{i+1}, k_i \ldots  k_L\rangle - 
\vert k_1\ldots k_i , k_{i+1}\ldots  k_L\rangle \bigr).
\end{align}

\begin{proposition}\label{k:pr:com}
$\varphi_{\mathfrak{s}, \mathfrak{t}}$ is spectrum preserving. 
Namely,  $\varphi_{\mathfrak{s}, \mathfrak{t}}\,
H_\mathfrak{t} = H_\mathfrak{s}\,
\varphi_{\mathfrak{s}, \mathfrak{t}}$
holds for any sectors $\mathfrak{s} \subset \mathfrak{t}$.
\end{proposition}
\begin{proof}
Consider the actions on the ket vector
$\vert k \rangle = \vert k_1,\ldots, k_L\rangle \in V_{\mathfrak{t}}$:
\begin{align}
\label{action1}
\varphi_{\mathfrak{s}, \mathfrak{t}}H_\mathfrak{t} \vert k \rangle 
&= \sum_{i \in {\mathbb Z}_L}\Theta(k_i-k_{i+1})
\bigl(\vert k'_1\ldots  k'_{i+1}, k'_i \ldots  k'_L\rangle - 
\vert k'_1\ldots k'_i , k'_{i+1}\ldots  k'_L\rangle \bigr),\\
\label{action2}
H_\mathfrak{s} \varphi_{\mathfrak{s}, \mathfrak{t}}\vert k \rangle 
&= \sum_{i \in {\mathbb Z}_L}\Theta(k'_i-k'_{i+1})
\bigl(\vert k'_1\ldots  k'_{i+1}, k'_i \ldots  k'_L\rangle - 
\vert k'_1\ldots k'_i , k'_{i+1}\ldots  k'_L\rangle \bigr),
\end{align}
where $x'$ is the one specified in (\ref{k:eq:pvv}).
For simplicity, let us write $(k_i, k_{i+1})$ as $(x,y)$.
{}From (\ref{k:eq:pvv}), we see that
$x>y$ implies $x'\ge y'$, and similarly 
$x<y$ implies $x'\le y'$.
{}From this fact and 
the definition of $\Theta$ in \eqref{k:eq:step},
the discrepancy of the coefficients
$\Theta(x-y)$ and $\Theta(x'-y')$ 
in the above two formulas can possibly make difference only when 
($x>y$ and $x'=y'$) or ($x<y$ and $x'=y'$).
But in the both cases, the vector
$\vert . . .\,  y', x' . . . \rangle -  
\vert . . .\,  x', y' . . . \rangle$ is zero.
Thus the right-hind sides of \eqref{action1}
and \eqref{action2} are the same.
\end{proof}
Our Hamiltonian arises as an expansion coefficient 
of a commuting transfer matrix $T(\lambda)$ with respect to 
the spectral parameter $\lambda$.  See \eqref{s:baxter}. 
However, the commutativity 
$\varphi_{\mathfrak{s}, \mathfrak{t}}
T(\lambda)_{\mathfrak{t}}
=T(\lambda)_{\mathfrak{s}}
\varphi_{\mathfrak{s}, \mathfrak{t}}$ 
does not hold in general.
 
To each sector 
$\mathfrak{s}=\{s_1<\cdots < s_{n-1}\} 
\in {\mathcal S}$, we associate 
\begin{align}
{\rm Spec}\,(\mathfrak{s})=
\text{multiset of eigenvalues of }\, H_\mathfrak{s},
\end{align}
where the multiplicity of an element represents, of course,  
the degree of its degeneracy.
This definition  is just a translation of \eqref{k:eq:spm}
into the notation \eqref{k:eq:hh}.
The property \eqref{k:eq:ccs} reads
\begin{align}\label{k:eq:scc}
{\rm Spec}(s_1,\ldots, s_{n-1})
={\rm Spec}(L-s_{n-1},\ldots, L-s_1).
\end{align}
One has $\sharp {\rm Spec}\,(\mathfrak{s}) 
= \dim V_\mathfrak{s}= \dim V_\mathfrak{s}^\ast$.
Lemma \ref{k:le:si} (2) and proposition \ref{k:pr:com} lead to
\begin{theorem}\label{k:cor:spe}
There is an embedding of the spectrum 
${\rm Spec}(\mathfrak{s}) \hookrightarrow 
{\rm Spec}(\mathfrak{t})$ 
for any pair of sectors such that $\mathfrak{s} \subset \mathfrak{t}$.
In particular, ${\rm Spec}(\Omega)$ contains 
the eigenvalues of the Hamiltonian 
$H_{\mathfrak s}$ of all the sectors 
${\mathfrak s} \in {\mathcal S}$.
\end{theorem}
See figure \ref{spec} for example.

\begin{figure}[h]
\begin{center}
\includegraphics{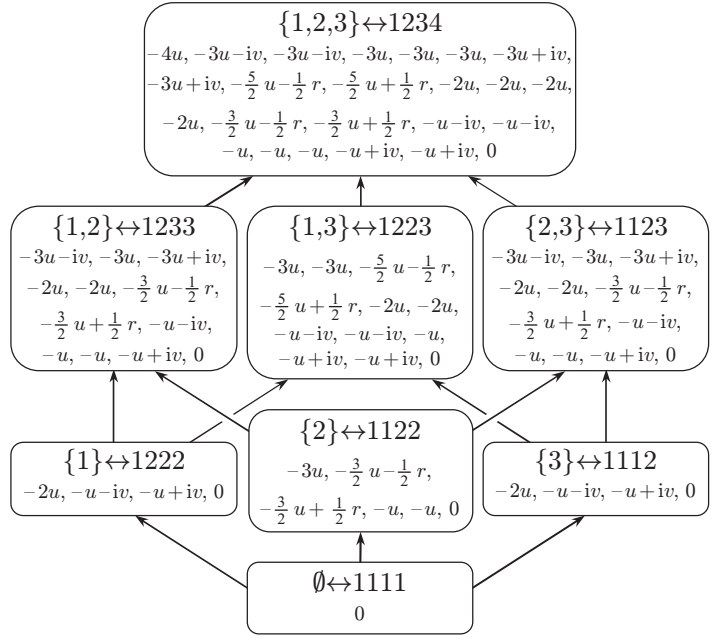}
 \caption{${\rm Spec}({\mathfrak s})$ for $L=4$.
$u=p+q,\ v=p-q,\ r=\sqrt{-7p^2+18pq-7q^2}$.
The symmetry \eqref{k:eq:scc} can be also observed.}
 \label{spec}
\end{center}
\end{figure}

\subsection{\mathversion{bold}
Spectral duality in the maximal sector $\Omega$}
\label{subsec:maxsec}

As indicated in theorem \ref{k:cor:spe},
the structure of the spectrum in 
the maximal sector $\Omega \in {\mathcal S}$
is of basic importance.
In this subsection we concentrate on this
sector and elucidate a duality.

Define a $\C$-linear map $\omega$ by
\begin{align}\label{k:eq:oma}
\begin{array}{cccc}
\omega: &  
V_{\scriptscriptstyle \Omega}^\ast & \overset{\sim}{\longrightarrow} 
& V_{\scriptscriptstyle \Omega} \\
&\langle k_1,\ldots, k_L\vert & \mapsto &
{\rm sgn}(k)\;
\vert k_L, \ldots, k_1\rangle,
\end{array}
\end{align}
where ${\rm sgn}(k)={\rm sgn}(k_1,\ldots, k_L)$ stands for the 
signature of the permutation.
(Note that 
${\mathcal P}(\Omega)$ is the set of permutations of 
$(1, 2,\ldots, L)$.)
Obviously, $\omega$ is bijective.

It turns out that $\omega$ interchanges the eigenvalues of Hamiltonian 
as $E \leftrightarrow -L(p+q)-E$.

\begin{theorem}\label{k:pr:ten}
Let $\langle \phi \vert \in V_{\scriptscriptstyle \Omega}^\ast$ be an 
eigenvector such that
$\langle \phi \vert H_{\scriptscriptstyle \Omega} 
= E\langle \phi \vert$.
Set $\vert \psi \rangle 
= \omega(\langle \phi \vert) \in V_{\scriptscriptstyle \Omega}$.
Then $H_{\scriptscriptstyle \Omega} \vert \psi \rangle
=(-L(p+q)-E)\vert \psi \rangle$ holds.
\end{theorem}
\begin{proof}
Let $\langle \phi \vert 
= \sum_{k \in {\mathcal P}(\Omega)}
f(k_1,\ldots, k_L)\langle k_1,\ldots, k_L \vert$.
Then $\langle \phi \vert H_{\scriptscriptstyle \Omega} 
= E\langle \phi \vert$ 
is expressed as
\begin{align*}
\sum_{i \in {\mathbb Z}_L}\Theta(k_i-k_{i+1})
\bigl(f(k^{(i)})-f(k)\bigr)= Ef(k),
\end{align*}
where we have used the shorthand
$k=(k_1,\ldots, k_i, k_{i+1},\ldots, k_L)$ and 
$k^{(i)}=(k_1,\ldots, k_{i+1}, k_i, \ldots, k_L)$.
Adding $(p+q)L f(k)$ to the both sides we get
\begin{align*}
\sum_{i \in {\mathbb Z}_L}\Theta(k_i-k_{i+1})f(k^{(i)})
+\sum_{i \in {\mathbb Z}_L}(p+q-\Theta(k_i-k_{i+1}))f(k)
= (E+ L(p+q))f(k),
\end{align*}
Since $k_i$'s are all distinct in the sector $\Omega$ 
under consideration, 
the coefficient in the second term equals $\Theta(k_{i+1}-k_i)$.
Multiplication of $-{\rm sgn}(k) = {\rm sgn}(k^{(i)})$ on 
the both sides leads to
\begin{align*}
&\sum_{i \in {\mathbb Z}_L}\Theta(k_i-k_{i+1})
{\rm sgn}(k^{(i)})f(k^{(i)})
-\sum_{i \in {\mathbb Z}_L}\Theta(k_{i+1}-k_i)
{\rm sgn}(k)f(k)\\
&= (-E-L(p+q)){\rm sgn}(k)f(k).
\end{align*}
This coincides with the equation 
$H_{\scriptscriptstyle \Omega} \vert \psi \rangle
=(-L(p+q)-E)\vert \psi \rangle$ on 

\noindent
$\vert \psi \rangle = \sum_{k \in {\mathcal P}(\Omega)}
{\rm sgn}(k)f(k_1,\ldots, k_L)\vert k_L,\ldots, k_1\rangle$.
\end{proof}

\begin{remark}
It is easy to see that 
$\langle \phi | = \sum_{k\in {\mathcal P}(\Omega)}
\langle k_1,\ldots, k_L| \in 
V_{\scriptscriptstyle \Omega}^\ast$ is the
eigen bra vector with the largest eigenvalue $E=0$.
It follows that
$\omega(\langle \phi | ) \in V_{\scriptscriptstyle \Omega}$ is 
the eigen ket vector with the smallest eigenvalue $-L(p+q)$.
Namely, one has
\begin{align}
\left(H_{\scriptscriptstyle \Omega}+L(p+q)\right)
\sum_{k\in {\mathcal P}(\Omega)}
{\rm sgn}(k)|k_L,\ldots, k_1\rangle = 0.
\end{align}
\end{remark}

In view of conjecture \ref{k:con:dep} and the remark following it,
we assume the diagonalizability 
of the Hamiltonian $H_\Omega$\footnote{
Theorem \ref{k:cor:spe} is derived on the basis of generalized 
eigenvectors hence its validity is independent of the 
diagonalizability of the Hamiltonian.} .
Then every eigenvalue in ${\rm Spec}(\Omega)$ 
is associated with an eigenvector in 
$V_{\scriptscriptstyle \Omega}^\ast$.
Therefore theorem \ref{k:pr:ten} implies
\begin{corollary}\label{k:cor:ome}
${\rm Spec}(\Omega) = -L(p+q)-{\rm Spec}(\Omega)$.
\end{corollary}

\begin{figure}[h]
\begin{center}
\includegraphics{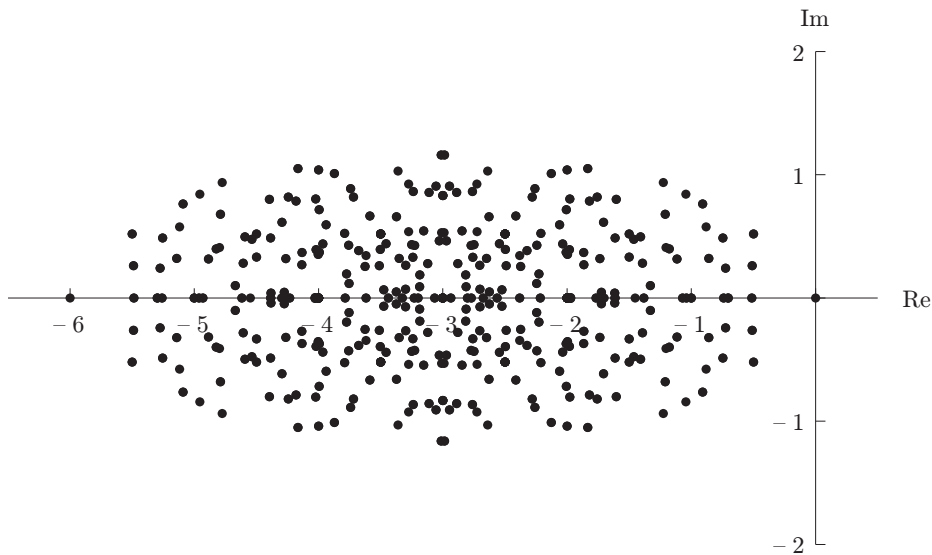}
 \caption{${\rm Spec}(\Omega)$ for 
$L=6$ and $(p,q)=(0.8, 0.2)$.
The symmetry with respect to 
$-L(p+q)/2=-3$ can be observed.}
\label{1-6}
\end{center}
\end{figure}

Figure \ref{1-6} is a plot showing this property.
The property of 
interchanging the eigenvalues of Hamiltonian 
$E \leftrightarrow -L(p+q)-E$ will be 
referred as {\em spectrum reversing}.
Our main task in the sequel is to extend $\omega$ to
a spectrum reversing operator between 
general sectors, and to identify the 
``genuine components" that are 
in bijective correspondence thereunder.
This will be achieved as $\omega^\circ$ 
in theorem \ref{k:th:sd}.

\subsection{\mathversion{bold} Genuine components
$X_\mathfrak{s}^\ast$ and $Y_\mathfrak{s}$}
\label{subsec:genuine}

Theorem \ref{k:cor:spe} motivates us to classify the
eigenvalues ${\rm Spec}\, (\mathfrak{s})$ 
in a sector $\mathfrak{s}$ into two kinds.
One is those coming from the smaller sectors 
$\mathfrak{u} \subset \mathfrak{s}$
through the embedding 
${\rm Spec}(\mathfrak{u}) 
\hookrightarrow {\rm Spec}(\mathfrak{s})$.
The other is the {\em genuine eigenvalues} that 
are born at $\mathfrak{s}$ without such an origin. 
Having this feature in mind 
we introduce a quotient $X_\mathfrak{s}^\ast$ 
of $V_\mathfrak{s}^\ast$ and 
a subspace $Y_\mathfrak{s}$ of $V_\mathfrak{s}$ as
\begin{align}\label{k:eq:xy}
X^\ast_\mathfrak{s} 
= V_\mathfrak{s}^\ast/
\sum_{\mathfrak{u} \subset \mathfrak{s}}{\rm Im }\, 
\overset{\leftarrow}{\varphi}_{\mathfrak{u},\mathfrak{s}},
\quad
Y_\mathfrak{s} 
= \bigcap_{\mathfrak{u}\subset \mathfrak{s}}{\rm Ker }\,
\overset{\rightarrow}{\varphi}_{\mathfrak{u}, \mathfrak{s}}.
\end{align}
We call $X_\mathfrak{s}^\ast$ and $Y_\mathfrak{s}$ 
the {\em genuine component} of 
$V_\mathfrak{s}^\ast$ and $V_\mathfrak{s}$, respectively.
(We set $X_\emptyset^\ast = V_\emptyset^\ast 
= \C \langle 1,\ldots, 1 \vert$
and 
$Y_\emptyset = V_\emptyset = \C\vert 1, \ldots, 1 \rangle$.)
The Hamiltonian $H_\mathfrak{s}$ acts on each 
$X_\mathfrak{s}^\ast$ and $Y_\mathfrak{s}$ 
owing to proposition \ref{k:pr:com}.
The vector spaces 
$X_\mathfrak{s}^\ast$ and $Y_\mathfrak{s}$ 
are dual to each other canonically,  
therefore 
\begin{align}\label{k:eq:dxy}
\dim X_\mathfrak{s}^\ast = \dim Y_\mathfrak{s}.
\end{align}

We wish to focus on the spectrum 
that are left after excluding the embedding structure
explained above and in theorem \ref{k:cor:spe}.
This leads us to define 
the set of genuine eigenvalues of a sector $\mathfrak{s}$ as
\begin{align}\label{k:eq:sp0}
\begin{split}
{\rm Spec}^\circ(\mathfrak{s})&=
\text{multiset of eigenvalues of }\, H_\mathfrak{s}
\vert_{X_\mathfrak{s}^\ast}\\
&=
\text{multiset of eigenvalues of }\, H_\mathfrak{s}
\vert_{Y_\mathfrak{s}}.
\end{split}
\end{align}

Let us write the image of $\langle\phi| \in V^\ast_{\mathfrak s}$ in 
$X^\ast_{\mathfrak s} $ under the natural projection
by $[\langle\phi|]$.
Fix an embedding of 
$X^\ast_{\mathfrak s}$ into $V^\ast_{\mathfrak s}$
sending each eigenvector $[\langle\phi|] \in X^\ast_{\mathfrak s}$
to an eigenvector $\langle\phi'|\in V^\ast_{\mathfrak s}$
with the same eigenvalue
satisfying $[\langle\phi|]=[\langle\phi'|]$.
The image of the embedding is complementary to 
$\sum_{\mathfrak{u} \subset \mathfrak{s}}{\rm Im }\, 
\overset{\leftarrow}{\varphi}_{\mathfrak{u},\mathfrak{s}}$,
therefore we can treat the first relation in \eqref{k:eq:xy} as 
$V^\ast_{\mathfrak s}=X^\ast_{\mathfrak s}\oplus
 \sum_{\mathfrak{u} \subset \mathfrak{s}}{\rm Im }\, 
\overset{\leftarrow}{\varphi}_{\mathfrak{u},\mathfrak{s}}$.
Then the following decomposition holds:
\begin{align}
V_\mathfrak{s}^\ast 
= \bigoplus_{\mathfrak{u}\subseteq \mathfrak{s}}
X_\mathfrak{u}^\ast \,
\overset{\leftarrow}\varphi_{\mathfrak{u},\mathfrak{s}}.
\label{k:eq:vss}
\end{align}
From theorem \ref{k:cor:spe} and 
(\ref{k:eq:vss}) we have
\begin{align}\label{k:eq:sus}
{\rm Spec}(\mathfrak{s}) 
= \bigcup_{\mathfrak{u}\subseteq \mathfrak{s}}
{\rm Spec}^\circ(\mathfrak{u}),
\end{align}
where the multiplicity is taken into account 
for the union of the multisets.
In terms of the cardinality, this amounts to
\begin{align}\label{k:eq:dvx}
\dim V_\mathfrak{s}^\ast 
= \sum_{\mathfrak{u} \subseteq \mathfrak{s}}
\dim X_\mathfrak{u}^\ast.
\end{align}

\begin{theorem}[Dimensional duality]
\label{k:th:dd}
For any sector $\mathfrak{s} \in {\mathcal S}$, the 
following equality is valid:
\begin{align*}
\dim X^\ast_\mathfrak{s} = \dim X^\ast_{\ol{\mathfrak{s}}},
\end{align*}
or equivalently 
$\sharp{\rm Spec}^\circ(\mathfrak{s})
=\sharp{\rm Spec}^\circ(\overline{\mathfrak{s}})$.
Here $\overline{\mathfrak{s}}$ denotes the 
complement sector \eqref{k:eq:cmp}.
\end{theorem}
See figure \ref{dimdata} for example with $L=4$.
The proof is due to the standard M\"obius inversion
in the poset ${\mathcal S}$ 
and available in appendix \ref{app:dd}.

\begin{figure}[h]
\begin{center}
\includegraphics{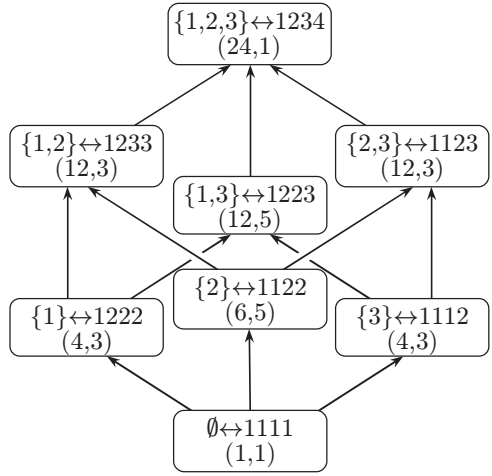}
\caption{The data 
$(\dim V_\mathfrak{s}^\ast, \dim X_\mathfrak{s}^\ast) = 
(\sharp {\rm Spec}(\mathfrak{s}), 
\sharp {\rm Spec}^\circ(\mathfrak{s}))$ is presented 
for each $\mathfrak{s}$ in the same Hasse diagram (c)
in figure \ref{Hasse}.
The dimensional duality (theorem \ref{k:th:dd}) 
can be observed.
For a systematic calculation of these data, see appendix \ref{app:dd}.
}
\label{dimdata}
\end{center}
\end{figure}

The following lemma, although slightly technical, plays a key role in our
subsequent argument.

\begin{lemma}\label{k:le:key}
$ $

(1) $\overset{\rightarrow}{\varphi}_{
{\scriptscriptstyle \Omega}\setminus\{r\}, 
{\scriptscriptstyle \Omega}}\,
\omega({\rm Im} 
\overset{\leftarrow}{\varphi}_{
{\scriptscriptstyle \Omega}\setminus\{r\}, 
{\scriptscriptstyle \Omega}})=0$ for any 
$r \in \Omega$.

(2) $\overset{\rightarrow}{\varphi}_{\mathfrak{s}, 
{\scriptscriptstyle \Omega}}\,
\omega({\rm Im} 
\overset{\leftarrow}{\varphi}_{\mathfrak{u}, 
{\scriptscriptstyle \Omega}})=0$ unless
$\mathfrak{u} \supseteq \overline{\mathfrak{s}}$.
\end{lemma}
\begin{proof}
(1) For brevity we write $\Omega_r=\Omega\setminus\{r\}$.
We illustrate an example $L=5, \Omega=\{1,2,3,4\}, 
\Omega_2=\{1,3,4\}$, 
from which the general case is easily understood.
Recall the scheme as in (\ref{k:eq:ppex}):
\begin{align*}
\begin{split}
{\mathcal P}(\Omega=\{1,2,3,4\}): 
&\quad1\overset{1}{\vert}2\overset{2}{\vert}3\overset{3}{\vert}
4\overset{4}{\vert}5,\\
{\mathcal P}(\Omega_2=\{1,3,4\}):
&\quad1\overset{1}{\vert}\,2\,2
\overset{3}{\vert}3\overset{4}{\vert}4.
\end{split}
\end{align*}
Thus $\overset{\leftarrow}
{\varphi}_{{\scriptscriptstyle \Omega}_2, 
{\scriptscriptstyle \Omega}}$ is 
the operator replacing the local states
$3 \rightarrow 4, 4 \rightarrow 5$ and moreover changes 
$\langle . . . 2, . . . , 2, . . . \vert$ into the symmetric sum 
$\langle . . . 3, . . . , 2, . . . \vert + \langle . . . 2, . . . , 3, . . . \vert$. 
At the next stage, 
$\omega$ in (\ref{k:eq:oma})
attaches the factor ${\rm sgn}(k)$ which 
makes the above sum antisymmetric.
Finally, $\overset{\rightarrow}
{\varphi}_{{\scriptscriptstyle \Omega}_2, 
{\scriptscriptstyle \Omega}}$
makes the antisymmetrized letters $2$ and $3$ merge into $2$ again 
(and also does $4\rightarrow 3, 5\rightarrow 4$),
which therefore kills the vector. For example, 
\begin{align*}
\langle 4 2 3 1 2 \vert 
&\overset{\overset{\leftarrow}
{\varphi}_{{\scriptscriptstyle \Omega}_2, {\scriptscriptstyle \Omega}}}
{\longmapsto}
\langle 5 2 4 1 3 \vert + \langle 5 3 4 1 2 \vert\\
&\;\,\overset{\omega}{\longmapsto}
-\vert 31425 \rangle + \vert 2 1 4 3 5 \rangle\\
&\overset{\overset{\rightarrow}
{\varphi}_{{\scriptscriptstyle \Omega}_2, {\scriptscriptstyle \Omega}}}
{\longmapsto}
-\vert 21324 \rangle + \vert 2 1 3 2 4 \rangle = 0.
\end{align*}

(2) Note that
${\rm Im} \overset{\leftarrow}{\varphi}_{\mathfrak{u}, 
{\scriptscriptstyle \Omega}}
= V_\mathfrak{u}^\ast \overset{\leftarrow}
{\varphi}_{\mathfrak{u}, {\scriptscriptstyle \Omega}}$.
Thus we are to ask when 
$\overset{\rightarrow}{\varphi}_{\mathfrak{s}, 
{\scriptscriptstyle \Omega}}\,
\omega(
V_\mathfrak{u}^\ast \overset{\leftarrow}
{\varphi}_{\mathfrak{u}, {\scriptscriptstyle \Omega}})$
vanishes.
It is helpful to view this as a process in the Hasse diagram going 
from $V_\mathfrak{u}^\ast$ to $V_\mathfrak{s}$ via the maximal sector 
$\Omega$ as in figure \ref{shcheme},
where
$\overline{\mathfrak{u}}=
\{\overline{u}_1,\ldots, \overline{u}_a\}$ and 
$\overline{\mathfrak{s}}=
\{\overline{s}_1,\ldots, \overline{s}_b\}$. 
\begin{figure}
\begin{center}
\includegraphics{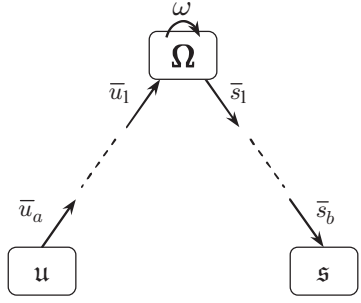}
\caption{A conceptual scheme of the proof of lemma \ref{k:le:key} (2).}
\label{shcheme}
\end{center}
\end{figure}
In figure \ref{shcheme}, the arrows $\nearrow$ represent 
the factorization 
$\overset{\leftarrow}{\varphi}_{\mathfrak{u}, 
{\scriptscriptstyle \Omega}}
=\overset{\leftarrow}{\varphi}_{\mathfrak{u},
\mathfrak{u}\cup\{\overline{u}_a\}}
\cdots\overset{\leftarrow}{\varphi}_{
{\scriptscriptstyle \Omega}\setminus\{\overline{u}_1\}, 
{\scriptscriptstyle \Omega}}$
due to lemma \ref{k:le:fac}
growing $\mathfrak{u}$ up to $\Omega$ by adding $\overline{u}_i$'s 
one by one.
Similarly the arrows $\searrow$ stand for
$\overset{\rightarrow}{\varphi}_{\mathfrak{s}, 
{\scriptscriptstyle \Omega}}
=\overset{\rightarrow}{\varphi}_{
\mathfrak{s},\mathfrak{s}\cup\{\overline{s}_b\}} \cdots
\overset{\rightarrow}{\varphi}_{
{\scriptscriptstyle \Omega}\setminus\{\overline{s}_1\}, 
{\scriptscriptstyle \Omega}}$
shrinking $\Omega$ down to $\mathfrak{s}$ by removing 
$\overline{s}_i$'s one by one.
(The arrows attached to $\overline{u}_i$ ($\overline{s}_i$)
are the same (opposite) 
as those in the Hasse diagram.) 
In this way 
\begin{align*}
\overset{\rightarrow}{\varphi}_{\mathfrak{s}, 
{\scriptscriptstyle \Omega}}\,
\omega(
V_\mathfrak{u}^\ast \overset{\leftarrow}
{\varphi}_{\mathfrak{u}, {\scriptscriptstyle \Omega}})
&= \cdots 
\overset{\rightarrow}{\varphi}_{
{\scriptscriptstyle \Omega}\setminus\{\overline{s}_1\},
{\scriptscriptstyle \Omega}}\,
\omega(\cdots 
\overset{\leftarrow}{\varphi}_{
{\scriptscriptstyle \Omega}\setminus\{\overline{u}_1\}, 
{\scriptscriptstyle \Omega}})\\
&= \cdots 
\overset{\rightarrow}{\varphi}_{
{\scriptscriptstyle \Omega}\setminus\{\overline{s}_i\}, 
{\scriptscriptstyle \Omega}}\,
\omega(\cdots 
\overset{\leftarrow}{\varphi}_{
{\scriptscriptstyle \Omega}\setminus\{\overline{u}_j\}, 
{\scriptscriptstyle \Omega}})\quad
\text{for any } 1\!\le\!i\!\le\!b,\; 1\!\le\!j\!\le\!a,
\end{align*}
where the second equality is due to 
lemma \ref{k:le:fac} which assures that 
the factorization is possible in arbitrary orders.
{}From the assertion (1) we thus find that  
this vanishes if 
$\overline{\mathfrak{s}} \cap \overline{\mathfrak{u}}
\neq \emptyset$.
In other words, 
$\overset{\rightarrow}{\varphi}_{\mathfrak{s}, 
{\scriptscriptstyle \Omega}}\,
\omega(
V_\mathfrak{u}^\ast \overset{\leftarrow}{\varphi}_{\mathfrak{u}, 
{\scriptscriptstyle \Omega}})
=0$ unless $\emptyset = 
\overline{\mathfrak{s}} \cap \overline{\mathfrak{u}}$, or
equivalently $\mathfrak{u} \supseteq \overline{\mathfrak{s}}$.
\end{proof}

\begin{proposition}\label{k:pr:vy}
 $ $ 

(1)
$V_\mathfrak{s} = \bigoplus_{\mathfrak{u}\supseteq 
\overline{\mathfrak{s}}}
\overset{\rightarrow}{\varphi}_{\mathfrak{s},
{\scriptscriptstyle \Omega}}\,
\omega(X_\mathfrak{u}^\ast
\overset{\leftarrow}{\varphi}_{\mathfrak{u},
{\scriptscriptstyle \Omega}})$.

(2)
$Y_\mathfrak{s} =
\overset{\rightarrow}{\varphi}_{\mathfrak{s},
{\scriptscriptstyle \Omega}}\,
\omega(X_{\overline{\mathfrak{s}}}^\ast\,
\overset{\leftarrow}{\varphi}_{
\overline{\mathfrak{s}}, {\scriptscriptstyle \Omega}})$.
\end{proposition}
\begin{proof}
(1)
\begin{align}\label{k:eq:ren}
\begin{split}
V_\mathfrak{s} &\overset{\text{Lem}. \ref{k:le:si}(1)}{=}
\overset{\rightarrow}{\varphi}_{\mathfrak{s},
{\scriptscriptstyle \Omega}}\,
V_{\scriptscriptstyle \Omega}
= \overset{\rightarrow}{\varphi}_{\mathfrak{s},
{\scriptscriptstyle \Omega}}\,
\omega(V_{\scriptscriptstyle \Omega}^\ast)
\overset{(\ref{k:eq:vss})}{=} 
\overset{\rightarrow}{\varphi}_{\mathfrak{s},
{\scriptscriptstyle \Omega}}\,
\omega\left(
\bigoplus_{\mathfrak{u}}X_\mathfrak{u}^\ast \,
\overset{\leftarrow}\varphi_{\mathfrak{u},
{\scriptscriptstyle \Omega}}\right)\\
&=\sum_{\mathfrak{u}}
\overset{\rightarrow}{\varphi}_{\mathfrak{s},
{\scriptscriptstyle \Omega}}\,
\omega\left(X_\mathfrak{u}^\ast \,
\overset{\leftarrow}\varphi_{\mathfrak{u},
{\scriptscriptstyle \Omega}}\right)
\overset{\text{Lem}. \ref{k:le:key} (2)}{=}
\sum_{\mathfrak{u} \supseteq \overline{\mathfrak{s}}}
\overset{\rightarrow}{\varphi}_{\mathfrak{s},
{\scriptscriptstyle \Omega}}\,
\omega\left(X_\mathfrak{u}^\ast \,
\overset{\leftarrow}\varphi_{\mathfrak{u},
{\scriptscriptstyle \Omega}}\right).
\end{split}
\end{align}
Taking the dimensions, we have
\begin{align*}
\begin{split}
&\dim V_\mathfrak{s} = \dim \sum_{\mathfrak{u} \supseteq 
\overline{\mathfrak{s}}}
\overset{\rightarrow}{\varphi}_{\mathfrak{s},
{\scriptscriptstyle \Omega}}\,
\omega\left(X_\mathfrak{u}^\ast \,
\overset{\leftarrow}\varphi_{\mathfrak{u},
{\scriptscriptstyle \Omega}}\right)
\le \sum_{\mathfrak{u} \supseteq \overline{\mathfrak{s}}}\dim 
\overset{\rightarrow}{\varphi}_{\mathfrak{s},
{\scriptscriptstyle \Omega}}\,
\omega\left(X_\mathfrak{u}^\ast \,
\overset{\leftarrow}\varphi_{\mathfrak{u},
{\scriptscriptstyle \Omega}}\right)\\
&\le \sum_{\mathfrak{u} \supseteq 
\overline{\mathfrak{s}}}\dim X_\mathfrak{u}^\ast
\overset{{\rm Th}. \ref{k:th:dd}}{=}
\sum_{\mathfrak{u} \supseteq \overline{\mathfrak{s}}}
\dim X_{\overline{\mathfrak{u}}}^\ast
=\sum_{\overline{\mathfrak{u}} \subseteq \mathfrak{s}}
\dim X_{\overline{\mathfrak{u}}}^\ast
\overset{(\ref{k:eq:dvx})}{=} \dim V_{\mathfrak{s}}^\ast
= \dim V_{\mathfrak{s}}.
\end{split}
\end{align*}
Thus all the inequalities $\le$ here are actually the equality $=$.
Moreover, all the sums $\sum$ in (\ref{k:eq:ren}) must be 
the direct sum $\oplus$, finishing the proof.

(2)
   Let 
${\tilde Y}_\mathfrak{s} = \overset{\rightarrow}
{\varphi}_{\mathfrak{s},{\scriptscriptstyle \Omega}}\,
\omega(X_{\overline{\mathfrak{s}}}^\ast\,
\overset{\leftarrow}{\varphi}_{
\overline{\mathfrak{s}}, {\scriptscriptstyle \Omega}})$.
By an argument similar to the proof of lemma \ref{k:le:key} (2),
one can easily show that ${\tilde Y}_\mathfrak{s}$ is killed by 
$\overset{\rightarrow}{\varphi}_{\mathfrak{s}\setminus\{n\},
\mathfrak{s}}$
for any $n \in \mathfrak{s}$.
In view of lemma \ref{k:le:fac}, this implies 
${\tilde Y}_\mathfrak{s} \subseteq Y_\mathfrak{s}$.
The proof is finished by noting 
$\dim {\tilde Y}_\mathfrak{s} 
= \dim \overset{\rightarrow}{\varphi}_{\mathfrak{s},
{\scriptscriptstyle \Omega}}\,
\omega(X_{\overline{\mathfrak{s}}}^\ast\,
\overset{\leftarrow}{\varphi}_{
\overline{\mathfrak{s}}, {\scriptscriptstyle \Omega}})
\overset{(1)}{=} \dim X_{\overline{\mathfrak{s}}}^\ast 
\overset{{\rm Th. } \ref{k:th:dd}}{=}
\dim X_\s^\ast \overset{(\ref{k:eq:dxy})}{=}
\dim Y_\s$.
\end{proof}

Combining proposition \ref{k:pr:vy} (2) and theorem \ref{k:pr:ten}, 
we arrive at our main result in this section.
\begin{theorem}[Spectral duality]
\label{k:th:sd}
For any sector $\mathfrak{s} \in {\mathcal S}$ and its 
complementary sector $\overline{\mathfrak{s}}$,  
there is a spectrum reversing bijection $\omega^\circ$ between 
their genuine components: 
\begin{align}\label{k:eq:sd}
\begin{array}{cccc}
\omega^\circ:
& X_{\overline{\mathfrak{s}}}^\ast &
\overset{\sim}{\longrightarrow} & Y_\mathfrak{s}\\
 &\langle \phi \vert & \mapsto &
\overset{\rightarrow}{\varphi}_{\mathfrak{s},
{\scriptscriptstyle \Omega}}\,
\omega(\langle \phi \vert \,
\overset{\leftarrow}{\varphi}_{\overline{\mathfrak{s}}, 
{\scriptscriptstyle \Omega}}\,).
\end{array}
\end{align}
In particular, the genuine spectrum 
enjoys the following duality: 
\begin{align}\label{k:eq:kore}
{\rm Spec}^\circ(\overline{\mathfrak{s}}) = -L(p+q)
-{\rm Spec}^\circ(\mathfrak{s}).
\end{align}
\end{theorem}
This relation 
is a refinement of theorem \ref{k:th:dd}.

\begin{example}\label{k:ex:has3}
Figure \ref{specmaru} presents 
${\rm Spec}^\circ({\mathfrak s})$ for $L=4$ 
in the same format as figure \ref{spec}. 
All the genuine eigenvalues form pairs with those 
in the complementary sectors 
to add up to $-L(p+q)=-4u$ including the multiplicity.
The full spectrum ${\rm Spec}(\mathfrak{s})$ 
in figure \ref{spec} is reproduced 
from the data in figure \ref{specmaru}
 and (\ref{k:eq:sus}).

\begin{figure}[h]
\begin{center}
\includegraphics{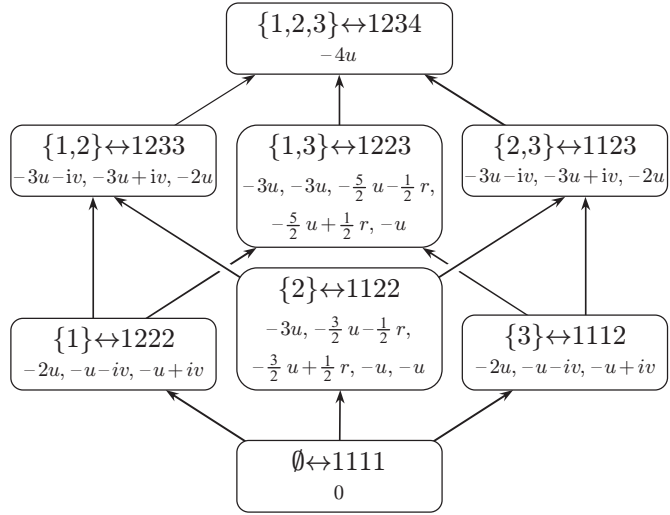}
 \caption{${\rm Spec}^\circ({\mathfrak s})$ 
for $L=4$. $u=p+q,\ v=p-q,\ r=\sqrt{-7p^2+18pq-7q^2}$.}
\label{specmaru}
\end{center}
\end{figure}

\end{example}

\begin{remark}
The genuine spectrum ${\rm Spec}^{\circ}$ also 
enjoys the symmetry \eqref{k:eq:scc}.
It follows that 
if a sector $\mathfrak{t}$ satisfies
$\mathfrak{t}\supset \s, \tilde{\s}$
with $\s =(s_1,\dots,s_{n-1})$
and $\tilde{\s}=(L-s_{n-1},\dots,L-s_1)$,
then $H_\mathfrak{t}$ is degenerated
because of 
$\Spec(\mathfrak{t})\supset\Spec^{\circ} (\s)
\cup\Spec^{\circ} (\tilde{\s})$
and $\Spec^{\circ}(\s)=\Spec^{\circ} (\tilde{\s})$.
\end{remark}

\section{Integrability of the model}\label{sec:5}
Our multi-species ASEP
is integrable in the sense that 
the eigenvalue formula of the Hamiltonian
can be derived by a nested Bethe ansatz \cite{Sc}. 
See also \cite{AB, BDV}.  

As mentioned in section \ref{sec:1},
our Hamiltonian is associated with 
the transfer matrix of the Perk-Schultz vertex model \cite{PS}.
In section~\ref{sec:5-1},
we derive the eigenvalues of 
the transfer matrix in a slightly more general way than \cite{Sc}. 
Namely we execute the nested Bethe ansatz in an arbitrary ``nesting order".
In section~\ref{sec:5-2}, we utilize it 
to give an alternative account of the spectral inclusion property 
(theorem~\ref{k:cor:spe}) in the Bethe ansatz framework. 
We also recall the original derivation of the asymptotic form of the
spectrum following \cite{K}.
In section~\ref{sec:5-3}, the Bethe ansatz results 
are presented in a more conventional parameterization 
with the spectral parameter having a difference property.

\subsection{Nested algebraic Bethe ansatz}\label{sec:5-1}

\subsubsection{Transfer matrix and eigenvalue formula}\label{k:sec:tme}
Let us derive the eigenvalues of the Hamiltonian $H$ \eqref{k:hdef}
for  the $(N-1)$-species ASEP on the ring ${\mathbb Z}_L$ 
by using the nested  algebraic Bethe ansatz. Let $W_j$ be a
vector space $W=\C^N$ at the $j$th site of the ring.
We define a matrix $R_{jk}(\lambda)\in 
\End(W_j\ot W_k)$ as
\begin{align}
R_{j k}(\lambda)=P_{j k}(1+\lambda h_{j k}),
\end{align} 
where $P_{j k}$  and $h_{jk}$ are, respectively, the 
permutation operator and the local Hamiltonian \eqref{k:eq:H} 
acting non-trivially on $W_j\ot W_k$. The non-zero 
elements are explicitly given by 
\begin{align}
R_{\alpha\alpha}^{\alpha\alpha}(\lambda)=1,
\,\,
R_{\alpha\beta}^{\alpha\beta}(\lambda)=
    \begin{cases}
   q\lambda  &\text{ for $\alpha<\beta$}, \\
   p \lambda &\text{ for $\alpha>\beta$},
    \end{cases}
\,\,
R_{\alpha\beta}^{\beta\alpha}(\lambda)=
    \begin{cases}
     1-q\lambda   &\text{ for $\alpha<\beta$}, \\
     1-p \lambda  &\text{ for $\alpha>\beta$}.
    \end{cases}
\label{s:R}
\end{align}
Here $\alpha, \beta\in\{1,2,\dots,N\}$, and 
$R_{\alpha\beta}^{\gamma \delta}(\lambda)$
stands for $R_{jk}(\lambda)(|\alpha\ket_j \otimes |\beta\ket_k)=
|\gamma\ket_j \otimes|\delta\ket_k 
R_{\alpha\beta}^{\gamma\delta}(\lambda)$
(summation over repeated indices will always be assumed).
The above  $R$-matrix satisfies the Yang-Baxter equation \cite{Ba}
\begin{align}
R_{23}(\lambda_2)R_{13}(\lambda_1)R_{12}(\lambda)=
     R_{12}(\lambda)R_{13}(\lambda_1)R_{23}(\lambda_2),
\label{s:YBE}
\end{align}
where the parameter $\lambda$ is given by
\begin{align}
\lambda=\xi(\lambda_1,\lambda_2)=
\frac{\lambda_1-\lambda_2}{1-(p+q)\lambda_2+p q \lambda_1 \lambda_2}.
\end{align}
This is not a simple difference $\lambda_1 - \lambda_2$. 
However, one can restore the difference property 
by changing variables as in section~\ref{sec:5-3}.
Thanks to \eqref{s:YBE}, the transfer matrix 
$T(\lambda) \in \End(W^{\ot L})$
\begin{align}
T(\lambda)=\tr_{W_0} [R_{0 L}(\lambda)\cdots R_{0 1}(\lambda)]
\label{s:TM}
\end{align}
constitutes a one-parameter commuting family
\begin{align}
[T(\lambda_1),T(\lambda_2)]=0. 
\end{align}
It means that $T(\lambda)$ is a generating function for
a set of mutually commuting 
``quantum integrals of motion" $\mathcal{I}_j$
($j=1,2,\dots$):
\begin{align}
\mathcal{I}_j=\left(\frac{\partial}{\partial \lambda}\right)^j
\ln T(\lambda) \biggr|_{\lambda=0}.
\label{s:IM}
\end{align}
$\mathcal{I}_0$ is the momentum operator 
related to the shift operator $C$ \eqref{k:eq:cs} 
by $C=\exp \mathcal{I}_0$.
$\mathcal{I}_1$ yields the ASEP Hamiltonian $H$ \eqref{k:hdef}:
\begin{align}
\mathcal{I}_1=\sum_{j \in {\mathbb Z}_L} R_{j j+1}(0) R_{jj+1}'(0)
=\sum_{j \in {\mathbb Z}_L} h_{j j+1}=H.
\label{s:baxter}
\end{align}
Thus the eigenvalue problem of  $H$
is contained in that of $T(\lambda)$.
To find the eigenvalues of $T(\lambda)$, 
we introduce the monodromy matrix $\mathcal{T(\lambda)} \in 
\End(W_0\ot W^{\ot L})$ by
\begin{align}
\mathcal{T(\lambda)}=R_{0 L}(\lambda)\cdots R_{0 1}(\lambda).
\end{align}
Its trace over the auxiliary space $W_0$ reproduces 
the transfer matrix \eqref{s:TM}
\begin{align}
T(\lambda)=\tr_{W_0} \mathcal{T(\lambda)}.
\end{align}
From the Yang-Baxter equation \eqref{s:YBE}, one sees the following is
valid:
\begin{align}
\T_2(\lambda_2)\T_1(\lambda_1)R_{12}(\lambda)=R_{12}(\lambda)\T_1(\lambda_1)\T_2(\lambda_2),
\label{s:YBR}
\end{align}
where $R_{12}(\lambda)$ here acts on the tensor product of 
two auxiliary spaces.

Let us  define the elements of the monodromy matrix in the auxiliary space
as $\mathcal{T}(\lambda)|\alpha\ket_0
=\mathcal{T}_{\alpha}^{\beta}(\lambda)|\beta\ket_0$,
where $\mathcal{T}_\alpha^\beta(\lambda)$ 
acts on the quantum space $W^{\ot L}$.
More explicitly,
\begin{align}
&\T(\lambda)=\begin{pmatrix}
              \T_{a_1}^{a_1}(\lambda) & B_{a_2}(\lambda) &\cdots & B_{a_N}(\lambda)       \\
               C^{a_2}(\lambda)   & \T_{a_2}^{a_2}(\lambda)  &\cdots & \T_{a_N}^{a_2}(\lambda)  \\
               \vdots      &   \vdots          &\ddots   &\vdots           \\
               C^{a_N}(\lambda)   &  \T_{a_2}^{a_N}(\lambda) & \cdots  &  \T_{a_N}^{a_N}(\lambda)
            \end{pmatrix}, \nn \\
&B_{a_j}(\lambda):=\mathcal{T}_{a_j}^{a_1}(\lambda), \quad
C^{a_j}(\lambda):=\mathcal{T}_{a_1}^{a_j}(\lambda) \quad \text{ for $j\in\{2,\dots,N\}$},
\label{s:B-op1}
\end{align}
Here we have introduced 
the indices $a_1, \ldots, a_N$ that are arbitrary as long as 
$\{a_j\}_{j=1}^{N}=\{1,\dots,N\}$.
They specify the nesting order 
$|a_1\ket,|a_2\ket,\dots, |a_N\ket$
\footnote{In the standard nested algebraic Bethe ansatz, the 
nesting order is chosen as $a_j=j$.}.

Let $|\vac\ket:=|a_1\ket_1\ot|a_1\ket_2\ot \cdots \ot |a_1\ket_L$ be
the ``vacuum state" in the quantum space. It immediately follows that 
the action of $\T(\lambda)$ on $|\vac \ket$ is given by
\begin{align}
\T(\lambda)|\vac \ket=\begin{pmatrix}
             1      & B_{a_2}(\lambda) &\cdots      & B_{a_N}(\lambda)       \\
             0      & d(\lambda)(q/p)^{L \theta_{12}}   &\cdots & 0\\
             \vdots & \vdots          &\ddots    & \vdots  \\
             0      & 0      & \cdots &      d(\lambda) (q/p)^{L\theta_{1N}}
            \end{pmatrix}
         |\vac \ket,
\label{s:vacuum}
\end{align}
where 
\begin{align}
d(\lambda):=(p \lambda)^L, \quad 
\theta_{ij}:=\theta(a_i-a_j):=
   \begin{cases}
       0  &\text{ for $a_i<a_j$}, \\
       1  &\text{ for $a_i>a_j$}.
   \end{cases}
\end{align}
Using the relation \eqref{s:YBR}, 
we can verify the following commutation relations:
\begin{align}
&B_{\alpha}(\lambda)B_{\beta}(\lambda')=\begin{cases}
                                    B_{\beta}(\lambda')B_{\alpha}(\lambda) & 
                                                       \text{ for $\alpha=\beta$}, \\     
                                    R_{\alpha \beta}^{\gamma \delta}(\xi(\lambda,\lambda'))
                                                                    B_{\delta}(\lambda')
                                                                    B_{\gamma}(\lambda)&
                                                       \text{ for $\alpha\ne\beta$},
                                  \end{cases} \nn \\
&\T_{a_1}^{a_1}(\lambda)B_{\alpha}(\lambda')                  
 =
    \begin{cases}
        f(\lambda',\lambda)B_{\alpha}(\lambda')\T_{a_1}^{a_1}(\lambda)+
                                g(\lambda,\lambda')B_{\alpha}(\lambda)\T_{a_1}^{a_1}(\lambda')
                                                         &\text{ for $a_1 <\alpha$}, \\
        \overline{f}(\lambda',\lambda)B_{\alpha}(\lambda')\T_{a_1}^{a_1}(\lambda)+
                                \overline{g}(\lambda,\lambda')B_{\alpha}(\lambda)
                                                        \T_{a_1}^{a_1}(\lambda')
                                                          &\text{ for $a_1 >\alpha$},
    \end{cases} \nn \\
&\T_{\gamma}^{\beta}(\lambda)B_{\alpha}(\lambda') 
  =
       \begin{cases}
                      f(\lambda,\lambda')
                          R_{\gamma \alpha}^{\varepsilon\delta}(\xi(\lambda,\lambda'))
                          B_{\delta}(\lambda')\T_{\varepsilon}^{\beta}(\lambda)-
            g(\lambda,\lambda')B_{\gamma}(\lambda)\T_{\alpha}^{\beta}(\lambda')  
                                                             &\text{for $a_1<\beta$}, \\
                     \overline{f}(\lambda,\lambda')
                          R_{\gamma \alpha}^{\varepsilon\delta}(\xi(\lambda,\lambda'))
                          B_{\delta}(\lambda')\T_{\varepsilon}^{\beta}(\lambda)-
 \overline{g}(\lambda,\lambda')B_{\gamma}(\lambda)\T_{\alpha}^{\beta}(\lambda')  
                                                               &\text{for $a_1>\beta$}.
       \end{cases}
\label{s:relation}
\end{align}
Here  $\alpha,\beta,\gamma,\delta\in\{a_j\}_{j=2}^N$, and the functions
$f$, $g$, $\overline{f}$ and $\overline{g}$ are defined by
\begin{align}
&f(\lambda,\mu)=\frac{1}{p \xi(\lambda,\mu)}=\frac{1-(p+q)\mu+ p q \lambda \mu}{p(\lambda-\mu)}, \nn \\
&g(\lambda,\mu)=1-f(\mu,\lambda)=f(\lambda,\mu)-\frac{q}{p}=
             \frac{(1-q \lambda)(1-p\mu)}{p(\lambda-\mu)}, \nn \\
&\overline{f}(\lambda,\mu)=f(\lambda,\mu)|_{p \leftrightarrow q}
        =\frac{p}{q}f(\lambda,\mu),  \quad
\overline{g}(\lambda,\mu)=g(\lambda,\mu)|_{p \leftrightarrow q}.
\end{align}

Consider the following state with the number of 
particles of the  $a_j$th kind being $m_{a_j}$:
\begin{align}
|\{\lambda^{(1)}\}\ket=F^{\alpha_1\cdots\alpha_{n_1}}
                   B_{\alpha_1}(\lambda_1^{(1)})
                            \cdots B_{\alpha_{n_1}}(\lambda_{n_1}^{(1)})|\vac \ket,
\;\;
\alpha_j\in \{a_2,\dots,a_N\} \quad (1\le j \le n_1),
\label{s:BV}
\end{align}
where 
\begin{align}
n_k:=\sum_{j=k+1}^N m_{a_j} \quad (1\le k \le N-1;\,n_0=L).
\label{s:m-number}
\end{align}
The sum over repeated indices in \eqref{s:BV} is restricted by the condition
\begin{align}
\sharp \{j|1\le j \le n_1, \alpha_j=a_k\}=m_{a_k} \quad (2\le k \le N).
\end{align}
Then the action of $T(\lambda)$ on $|\{\lambda^{(1)}\}\ket$ is calculated by 
using the relations
\eqref{s:relation} and \eqref{s:vacuum}:
\begin{align}
&T(\lambda)|\{\lambda^{(1)}\}\ket=\left[\T_{a_1}^{a_1}(\lambda)+\sum_{\alpha=2}^N       
                          \T_{a_\alpha}^{a_\alpha}(\lambda) \right]|\{\lambda^{(1)}\}\ket   
        =\left(\frac{p}{q}\right)^{\overline{n}_1}\prod_{j=1}^{n_1} f(\lambda_j^{(1)},\lambda) 
                                                             |\{\lambda^{(1)}\}\ket   \nn \\
&\qquad \quad +F^{\alpha_1\cdots\alpha_{n_1}}
  T_{\quad\,\alpha_1,\dots,\alpha_{n_1}}^{(1)\beta_1,\dots,\beta_{n_1}}
                                       (\lambda|\{\lambda^{(1)}\})
d(\lambda) \prod_{j=1}^{n_1}f(\lambda,\lambda_j^{(1)})
\prod_{j=1}^{n_1}B_{\beta_j}(\lambda_j^{(1)})|\vac \ket
+{\rm u.t.},
\label{s:nest1}
\end{align}
where the product $\prod_{j=1}^{n_1}B_{\beta_j}(\lambda_j^{(1)})$ is ordered from
left to right with increasing $j$;
$\alpha_j, \beta_j \in\{a_2,\dots,a_N\}$; $\overline{n}_k$ is an integer
given by
\begin{align}
\overline{n}_k:=\sum_{j=k+1}^N m_{a_j}\theta_{k j} \quad (1\le k \le N-1),
\end{align}
and 
$T_{\quad\,\alpha_1,\dots,\alpha_{n_1}}^{(1)\beta_1,\dots,\beta_{n_1}}
(\lambda|\{\lambda^{(1)}\})$ is a matrix element of 
$T^{(1)}(\lambda|\{\lambda^{(1)}\}) \in \End ((\mathbb{C}^{N-1})^{\otimes L})$ 
defined by
\begin{align}
T^{(1)}(\lambda|\{\lambda^{(1)}\})
=\sum_{\alpha=2}^N\left(\frac{q}{p}\right)^{(L-n_1)\theta_{1\alpha}}
                \left[R_{0 n_1}(\xi(\lambda,\lambda_{n_1}^{(1)}))\cdots
              R_{0 1}(\xi(\lambda,\lambda_1^{(1)}))\right]_{a_{\alpha}}^{a_{\alpha}}.
\end{align}
Note that the sum corresponds to the trace  over the $(N-1)$-dimensional
auxiliary space spanned by  the basis vectors $|a_j \ket_0$ ($2\le j \le N$), and
the quantum space acted on by $T^{(1)}(\lambda|\{\lambda^{(1)}\})$ is
spanned by the vector $\otimes_{j=1}^{n_1}|\alpha_j\ket_j$ where
$\alpha_j\in\{a_k\}_{k=2}^N$.

If we set $F^{\alpha_1\cdots \alpha_{n_1}}$ as the elements of
the eigenstate for $T^{(1)}(\lambda|\{\lambda^{(1)}\})$ and
choose the set of unknown numbers $\{\lambda^{(1)}_j\}_{j=1}^{n_1}$
so that the unwanted terms (u.t.) in \eqref{s:nest1} 
become zero ($\rm{u.t.}=0$),  the eigenvalue, written $\Lambda(\lambda)$,
of the transfer matrix $T(\lambda)$ is expressed as
\begin{align}
\Lambda(\lambda)=\left(\frac{p}{q}\right)^{\overline{n}_1}
                        \prod_{j=1}^{n_1} f(\lambda_j^{(1)},\lambda)+
     \Lambda^{(1)}(\lambda|\{\lambda^{(1)}\}) d(\lambda) 
                      \prod_{j=1}^{n_1}f(\lambda,\lambda_j^{(1)}).
\end{align}
Here $\Lambda^{(1)}(\lambda|\{\lambda^{(1)}\})$ is the eigenvalue of 
 $T^{(1)}(\lambda|\{\lambda^{(1)}\})$, which will be determined below.
Noting that
\begin{align}
\xi(\xi(\lambda_1,\mu),\xi(\lambda_2,\mu))=\xi(\lambda_1,\lambda_2),
\end{align}
and using the Yang-Baxter equation \eqref{s:YBE}, one finds that 
the transfer matrix
$T^{(1)}(\lambda|\{\lambda^{(1)}\})$ 
forms a commuting family
\begin{align}
[T(\lambda_1|\{\lambda^{(1)}\}),T(\lambda_2|\{\lambda^{(1)}\})]=0.
\end{align}
Hence the method similar to the above is also applicable to
the eigenvalue problem of $T^{(1)}(\lambda|\{\lambda^{(1)}\})$.
Namely constructing the state 
\begin{align}
|\{\lambda^{(2)}\}\ket =F^{(1)\alpha_1\cdots\alpha_{n_2}}B^{(1)}_{\alpha_1}(\lambda_1^{(2)})
\cdots B_{\alpha_{n_2}}^{(1)}(\lambda_{n_2}^{(2)}) |\vac^{(1)}\ket,\quad
|\vac^{(1)}\ket:=\bigotimes_{j=1}^{n_2}|a_2\ket _j,
\label{s:2nd-eigen}
\end{align}
where $\alpha_j\in\{3,\dots,N\}$ and
\begin{align}
B_{\alpha_j}^{(1)}(\lambda)= \left[R_{0 n_1}(\xi(\lambda,\lambda_{n_1}^{(1)}))\cdots
              R_{0 1}(\xi(\lambda,\lambda_1^{(1)}))\right]_{\alpha_j}^{a_2},
\label{s:B-op2}
\end{align}
we obtain 
\begin{align}
\Lambda^{(1)}(\lambda|\{\lambda_j^{(1)}\})
=&\left(\frac{q}{p}\right)^{(L-n_1)\theta_{12}}\left(\frac{p}{q}\right)^{\overline{n}_2}
     \prod_{j=1}^{n_2} f(\lambda_j^{(2)},\lambda)                      \nn \\
  &+
          \Lambda^{(2)}(\lambda|\{\lambda^{(2)}\}) 
               \prod_{j=1}^{n_1}\frac{1}{f(\lambda,\lambda_j^{(1)})}
                                       \prod_{j=1}^{n_2}f(\lambda,\lambda_j^{(2)}).
\end{align}
Note that the coefficient $F^{(1)\alpha_1\cdots\alpha_{n_2}}$ in \eqref{s:2nd-eigen}
and 
$ \Lambda^{(2)}(\lambda|\{\lambda^{(2)}\}) $ in the above are, respectively,
the elements of the eigenstate and the eigenvalue for the transfer matrix
\begin{align}
 T^{(2)}(\lambda|\{\lambda^{(2)}\}) 
   =\sum_{\alpha=3}^{N} 
    \left(\frac{q}{p}\right)^{(L-n_1)\theta_{1\alpha}+(n_1-n_2)\theta_{2 \alpha}}
             \left[R_{\alpha n_2} (\xi(\lambda,\lambda_{n_2}^{(2)}))
         \cdots
              R_{\alpha 1}(\xi(\lambda,\lambda_1^{(2)}))\right]_{a_{\alpha}}^{a_{\alpha}}.
\label{s:nest-2}
\end{align}
The sum corresponds to the trace  over  the $(N-2)$-dimensional auxiliary 
space spanned by  the basis vectors $|a_j \ket_0$ ($3\le j \le N$), and
the space acted on by $T^{(2)}(\lambda|\{\lambda^{(2)}\})$ is spanned
by $\otimes_{j=1}^{n_2}|\alpha_j\ket_j$ where
$\alpha_j\in\{a_k\}_{k=3}^N$.
Repeating this procedure, one obtains
\begin{align}
\Lambda^{(l)}(\lambda|\{\lambda^{(l)}\})=& 
      \left(\frac{q}{p}\right)^{\sum_{j=1}^{l}(n_{j-1}-n_j)\theta_{j l+1}-\overline{n}_{l+1}}
      \prod_{j=1}^{n_{l+1}} f(\lambda_j^{(l+1)},\lambda)                \nn \\
    &+
         \Lambda^{(l+1)}(\lambda|\{\lambda^{(l+1)}\}) 
               \prod_{j=1}^{n_l}\frac{1}{f(\lambda,\lambda_j^{(l)})}
                                       \prod_{j=1}^{n_{l+1}}f(\lambda,\lambda_j^{(l+1)}),
\end{align}
 for $2\le l< N-2$, and
\begin{align}
\Lambda^{(N-2)}(\lambda|\{\lambda^{(N-2)}\})=&
      \left(\frac{q}{p}\right)^{\sum_{j=1}^{N-2}(n_{j-1}-n_j)\theta_{j N-1}-\overline{n}_{N-1}}
      \prod_{j=1}^{n_{N-1}} f(\lambda_j^{(N-1)},\lambda)                \nn \\
    & \hspace*{-1cm} +  
        \left(\frac{q}{p}\right)^{\sum_{j=1}^{N-1}(n_{j-1}-n_j)\theta_{j N}}
               \prod_{j=1}^{n_{N-2}}\frac{1}{f(\lambda,\lambda_j^{(N-2)})}
               \prod_{j=1}^{n_{N-1}}f(\lambda,\lambda_j^{(N-1)}).
\end{align}
Thus we finally arrive at the eigenvalue formula of the transfer matrix: 
\begin{align}
\Lambda(\lambda)=&\left(\frac{q}{p}\right)^{-\overline{n}_1}
                      \prod_{j=1}^{n_1}f(\lambda_j^{(1)},\lambda) 
\nn \\
        &+d(\lambda)\sum_{k=1}^{N-2}
          \left(\frac{q}{p}\right)^{\sum_{j=1}^{k}(n_{j-1}-n_j) \theta_{j k+1}-
                         \overline{n}_{k+1}}
  \prod_{j=1}^{n_k}f(\lambda,\lambda_j^{(k)})\prod_{j=1}^{n_{k+1}}f(\lambda_j^{(k+1)},\lambda) \nn \\
        &  +d(\lambda)
            \left(\frac{q}{p}\right)^{\sum_{j=1}^{N-1}(n_{j-1}-n_j)\theta_{j N}}
            \prod_{j=1}^{n_{N-1}}f(\lambda,\lambda_j^{(N-1)}).
\label{s:DVF}
\end{align}
The unwanted terms disappear when 
the set of  unknown numbers $\{\lambda^{(n)}_l\}_{j=1}^{n_l}$ 
($1\le l \le N-1$) satisfy the following Bethe equations, 
which are also derived
by imposing the pole free conditions on 
the eigenvalue formula:
\begin{align}
& \left(\frac{q}{p}\right)^{(L-n_1)\theta_{12}}d(\lambda_j^{(1)}) =
    -\left(\frac{q}{p}\right)^{-\overline{n}_1+\overline{n}_2}\prod_{k=1}^{n_1} 
       \frac{f(\lambda_k^{(1)},\lambda_j^{(1)})}{f(\lambda_j^{(1)},\lambda_k^{(1)})}
       \prod_{k=1}^{n_2}\frac{1}{f(\lambda_k^{(2)},\lambda_j^{(1)})},   \nn \\
&\left(\frac{q}{p}\right)^{\sum_{j=1}^{l}(n_{j-1}-n_j)\theta_{j l+1}-
                           \sum_{j=1}^{l-1}(n_{j-1}-n_j)\theta_{jl}}   \nn \\
& \qquad \qquad \qquad=  
    -\left(\frac{q}{p}\right)^{-\overline{n}_{l}+\overline{n}_{l+1}}
    \prod_{k=1}^{n_l}\frac{f(\lambda_k^{(l)},\lambda_j^{(l)})}{f(\lambda_j^{(l)},\lambda_k^{(l)})}
    \frac{\prod_{k=1}^{n_{l-1}}f(\lambda_j^{(l)},\lambda_k^{(l-1)})} 
         {\prod_{k=1}^{n_{l+1}}f(\lambda_k^{(l+1)},\lambda_j^{(l)})}   \quad    
                                          \text{\,\, ($2\le l \le N-2$)}, \nn \\
&\left(\frac{q}{p}\right)^{\sum_{j=1}^{N-1}(n_{j-1}-n_j)\theta_{j N}-
                           \sum_{j=1}^{N-2}(n_{j-1}-n_j)\theta_{jN-1}} \nn \\
&\qquad \qquad \qquad  =-
            \left(\frac{q}{p}\right)^{-\overline{n}_{N-1}}
              \prod_{k=1}^{n_{N-1}} 
              \frac{f(\lambda_k^{(N-1)},\lambda_j^{(N-1)})}
                   {f(\lambda_j^{(N-1)},\lambda_k^{(N-1)})}
             \prod_{k=1}^{n_{N-2}}
                   f(\lambda_j^{(N-1)},\lambda_k^{(N-2)}).
\label{s:BAE}
\end{align}
Inserting the expression \eqref{s:DVF}
into \eqref{s:baxter}, one finds the 
spectrum of the Hamiltonian:
\begin{align}
E=\frac{\partial}{\partial \lambda}
\ln \Lambda(\lambda)\biggr|_{\lambda=0}=
\sum_{j=1}^{n_1}\frac{(1-p\lambda_j^{(1)})
(1-q\lambda_j^{(1)})}{\lambda_j^{(1)}}.
\label{s:spectrum}
\end{align}

Though the explicit form of the Bethe equation \eqref{s:BAE}
and its solutions depend on the nesting order in general,
the spectrum of the Hamiltonian \eqref{s:spectrum}, 
of course,  does not depend on it.

\subsubsection{Completeness of the Bethe ansatz}\label{k:subs:cp}

In this sub-subsection we 
exclusively consider the standard nesting order $a_j=j\,(1\le j \le N)$. 
Let us recall our setting and definitions.
We consider the transfer matrix $T(\lambda)$ \eqref{s:TM}
acting on the sector $V(m)$. See \eqref{k:eq:m-sum}.
The data $m=(m_1,\ldots, m_N) \in {\mathbb Z}^N_{\ge 0}$
specifies the number $m_j$ of the particles of the $j$th kind
and $m_1+\cdots + m_N = L$.
The sector 
$m$ is basic if it has the form 
$m=(m_1,\ldots, m_n,0,\ldots,0)$ with $m_1,\ldots, m_n$ 
all positive for some $1\le n\le L$.
The basic sectors are labeled either with 
${\mathcal M}$ \eqref{k:eq:M} or ${\mathcal S}$ \eqref{k:eq:I}
by the one to one correspondence 
${\mathcal M} \ni m \leftrightarrow {\mathfrak s}
\in {\mathcal S}$ \eqref{k:eq:bi}.
For a basic sector $m$, we write $V(m)$ also as $V_{\mathfrak s}$
as in \eqref{k:eq:hh}. 
$Y_{\mathfrak s}$ is the genuine component of $V_{\mathfrak s}$ 
defined in \eqref{k:eq:xy}.
${\rm Spec}^\circ({\mathfrak s})$ is the multiset of 
genuine eigenvalues of $H$ in $Y_{\mathfrak s}$ \eqref{k:eq:sp0}.

\begin{conjecture}\label{k:con:dep}
Suppose that $p \neq q$ are generic.
Then for any sector $V(m)$ which is not necessarily basic,
there exist $d=\dim V(m)$ distinct polynomials 
$\Lambda_1(\lambda), \ldots, \Lambda_d(\lambda)$ 
in $\lambda$ such that 
$\det(\zeta - T(\lambda))=\prod_{g=1}^d(\zeta - \Lambda_g(\lambda))$.
\end{conjecture}
We call $\Lambda_1(\lambda), \Lambda_2(\lambda), 
\ldots, \Lambda_d(\lambda)$ the {\em eigen-polynomials} of $T(\lambda)$.
(It should not be confused with the characteristic polynomial  
$\det(\zeta - T(\lambda))$.)
A direct consequence of conjecture \ref{k:con:dep} is 
that the transfer matrix $T(\lambda)$ hence the Hamiltonian $H$
are diagonalizable in arbitrary sectors.
(At $p=q$, the diagonalizability still holds but 
$\Lambda_g(\lambda)$'s are no longer distinct due to 
degeneracy caused by $sl(N)$-invariance.)

Now we turn to the completeness of the Bethe ansatz.
In the remainder of this subsection and section \ref{sec:5-2-2}, by 
the Bethe equations we mean 
the {\em polynomial equations} 
on $\{\lambda^{(l)}_j\mid 1 \le l \le N-1, 1 \le j \le n_j\}$ 
obtained from \eqref{s:BAE} 
by multiplying a polynomial in them so that 
the resulting two sides do not share a nontrivial common factor. 
We say that a set of complex numbers $\{\lambda^{(l)}_j\}$ is 
a Bethe root if it satisfies the Bethe equations.
Bethe roots $\{\lambda^{(l)}_j\}$ and 
$\{\lambda^{(l)\prime}_j\}$ are identified if 
$\lambda^{(l)}_j=\lambda^{(l)\prime}_{k_j}$ 
for some permutation $k_1, \ldots, k_{n_l}$ of 
$1,\ldots, n_l$ for each $l$. 
We say that a Bethe root $\{\lambda^{(l)}_j\}$ is {\em regular}
if none of them is equal to $1/p$ 
and two sides of any Bethe equation are nonzero.
Using the same notation as in conjecture \ref{k:con:dep}, we propose
\begin{conjecture}[Completeness]\label{k:con:cp}
$ $Suppose $p\neq q$ are generic.
\begin{itemize}
\item[(1)] 
For any sector, all the eigen-polynomials $\Lambda_g(\lambda)$ 
are expressed in the form \eqref{s:DVF} in 
terms of some Bethe root.
\item[(2)]   
For a basic sector ${\mathfrak s} \in {\mathcal S}$, 
there exist exactly $\dim Y_{\mathfrak s}$ regular Bethe roots and the associated $\dim Y_{\mathfrak s}$ eigen-polynomials
among $\Lambda_1(\lambda), \ldots, \Lambda_d(\lambda)$.
\item[(3)]The $\dim Y_{\mathfrak s}$ eigen-polynomials in (2) 
give ${\rm Spec}^\circ({\mathfrak s})$ 
by the logarithmic derivative \eqref{s:spectrum}.
\end{itemize}
\end{conjecture}

In view of section \ref{subsec:genuine}, 
it is natural to call the (conjectural) 
$\dim Y_{\mathfrak s}$ eigen-polynomials in 
conjecture \ref{k:con:cp} (2) the {\em genuine} eigen-polynomials
of the basic sector ${\mathfrak s}$.
Then conjecture \ref{k:con:cp} (3) is rephrased as claiming  
that the spectrum ${\rm Spec}({\mathfrak s})$ and 
the genuine spectrum 
${\rm Spec}^\circ({\mathfrak s})$ are obtained by the 
logarithmic derivatives of the $\dim V_{\mathfrak s}$ eigen-polynomials 
and the $\dim Y_{\mathfrak s}$ genuine eigen-polynomials, respectively.

Some examples supporting conjectures \ref{k:con:dep} and 
\ref{k:con:cp} are presented in appendix \ref{sec:appC}. 
In conjecture \ref{k:con:cp} (1),  
the Bethe roots corresponding to a non-genuine 
eigen-polynomial are not necessarily unique.
See the 2nd and the 3rd examples from the last in appendix \ref{sec:appC}. 
We expect that the Bethe vectors associated with 
the regular Bethe roots form a basis of $Y_{\mathfrak s}$. 

Theorem \ref{k:cor:spe} and conjecture \ref{k:con:cp} (2) 
bear some analogy with the $sl(N)$-invariant Heisenberg chain ($p=q$).
There, the number of the Bethe roots are conjecturally 
the Kostka numbers \cite{KKR} and 
the spectral embedding is induced by $sl(N)$ actions.
Here, the analogous roles are played by 
$\dim Y_{\mathfrak s}$ and 
$\overset{\leftarrow}{\varphi}_{\mathfrak{s}, \mathfrak{t}}$,
respectively.
%
\subsection{Properties of the spectrum}\label{sec:5-2}
%
Now we derive some consequences of the eigenvalue formula
\eqref{s:DVF}, \eqref{s:BAE} and \eqref{s:spectrum}. 
Sections \ref{KPZuc} and \ref{EWuc} are reviews of known 
derivation for reader's convenience.

\subsubsection{Spectral inclusion property}\label{sec:5-2-2}
First we rederive the spectral inclusion property
(theorem~\ref{k:cor:spe}) in the Bethe ansatz framework.
Consider the sector ${\mathfrak t} \in {\mathcal S}$
where the number of particles 
of the $j$th kind is $m_j \ge 1$ for any $j$.
In the notation \eqref{k:eq:bi}, ${\mathfrak t}$ reads
\begin{align}
\mathfrak{t}&=\{m_1,m_1+m_2,\dots,m_1+m_2+\dots+m_{N-1}\} \nn \\
            &\leftrightarrow 1^{m_1}2^{m_2}\cdots{N^{m_N}}
            =a_1^{m_{a_1}}\cdots a_{N}^{m_{a_N}},  
\label{s:sector-t}
\end{align}
where $m_1+\cdots+m_N=L\ge 2$.
Set
\begin{align}\label{k:aap}
a_{N}=a_{N-1}+1,  \quad 
\lambda_j^{(N-1)}\to
\frac{1}{p} \quad (1\le j \le n_{N-1}).
\end{align}
Due to the relations 
\begin{align}
\overline{n}_{N-1}=0, 
\quad f\left(\frac{1}{p},\lambda \right)
=\frac{p}{q}f\left(\lambda, \frac{1}{p} \right)=1
\,\,
\text{ for $\lambda \ne \frac{1}{p}$},
\end{align}
the following reduction relation holds:
\begin{align}
\Lambda(\lambda)=\overline{\Lambda}(\lambda)+d(\lambda)
           \left(\frac{q}{p}\right)^{\sum_{j=1}^{N-2}(n_{j-1}-n_j)\theta_{j N}+n_{N-1}}.
\label{s:inclusion}
\end{align}
Here $\overline{\Lambda}(\lambda)$ stands for the eigenvalue 
formula of the $(N-2)$-species case in the sector
\begin{align}
\tilde{a}_1^{m_{a_1}}\cdots \tilde{a}_{N-1}^{m_{a_{N-1}}+m_{a_{N}}}
            &= 1^{m_1}\cdots a_{N-1}^{m_{a_{N-1}}+m_{a_{N}}}
              a_N^{m_{a_N+1}}\cdots (N-1)^{m_{N}}                   \nn \\
            &\leftrightarrow \mathfrak{t}\setminus \{m_1+\cdots+m_{a_{N-1}}\}
            =\mathfrak{u},
\end{align}
where
\begin{align}
\tilde{a}_j=\begin{cases}
              a_j  &\text{ for $a_j< a_{N}$} \\
              a_j-1 &\text{ for $a_j \ge  a_{N}$}
            \end{cases}  \quad
(1\le j \le N-1).
\end{align}
If $\{\lambda^{(l)}_j \mid 1 \le j \le n_l,\, 1 \le l \le N-1\}$
is a solution of the Bethe equation in the sector 
${\mathfrak t}$,
so is $\{\lambda^{(l)}_j \mid 1 \le j \le n_l,\, 1 \le l \le N-2\}$
left after the substitution \eqref{k:aap}
in the sector ${\mathfrak u}$.
This is because 
the last Bethe equation in \eqref{s:BAE} 
for $(N-1)$-species case
becomes trivial, or alternatively 
one may say that the resulting $(N-2)$-species 
Bethe equation guarantees that
$\overline{\Lambda}(\lambda)$ is pole-free.

Inserting \eqref{s:inclusion} into \eqref{s:spectrum}, and
using $d(0)=d'(0)=0$, one thus finds the set of 
eigenvalues of the Hamiltonian for the sector 
$\mathfrak{t}$ includes that for the sector $\mathfrak{u}$.
Applying this argument repeatedly, one can see 
\begin{align}
\Spec(\mathfrak{s}) \subset \Spec(\mathfrak{t}) 
\quad \text{ for $\mathfrak{s}\subset \mathfrak{t}$}.
\end{align}
That is theorem~\ref{k:cor:spe}.
Since the solutions of the Bethe equation \eqref{s:BAE} 
depend on the nesting order, the set of solutions characterizing
the above $\Spec(\mathfrak{s})$ are, in general, not 
included in the original set of solutions 
characterizing $\Spec(\mathfrak{t})$.

\subsubsection{Stationary state}\label{sec:stationary-eigenvalue}
%
One of the direct consequences of the above 
property is that the stationary state 
$E=0$ for an arbitrary sector $\mathfrak{t}$ \eqref{s:sector-t} is given 
by setting all the Bethe roots to $1/p$, i.e.  
\begin{align}
\lambda_j^{(l)}\to \frac{1}{p}
\quad (1\le j \le n_l, \,\, 1 \le l \le N).
\end{align} 
It immediately follows that the eigen-polynomial of the stationary 
state is given by
\begin{align}
\Lambda(\lambda)=1+\sum_{k=1}^{N-1}\left(\frac{q}{p}\right)^{n_k}d(\lambda),
\end{align}
where $n_k$ is defined by \eqref{s:m-number}, and we consider
the standard nesting order $a_j=j\,(1\le j \le N)$.

On the other hand, in the framework of the Bethe ansatz,
the calculation of the corresponding eigenstate is rather
cumbersome.
It will be sketched in appendix \ref{sec:stationary}.
\subsubsection{KPZ universality class}\label{KPZuc}
%
From section \ref{sec:5-2-2}, one immediately
sees that the set of spectrum for the sector $\mathfrak{t}$ 
\eqref{s:sector-t} includes  those for 
sectors consisting of single particles:  
\begin{align}
\Spec(\mathfrak{s}_l)\subset \Spec(\mathfrak{t}), \quad 
\mathfrak{s}_l=\{m_1+\cdots+m_l \}\leftrightarrow 
1^{m_1+\cdots+m_l} 2^{m_{l+1}+\cdots+m_{N}}
\quad
(1\le l \le N-1).
\end{align}

As discussed in section \ref{sec:3}, the relaxation spectrum 
characterizing the universality class are 
the eigenvalues in the sector $\mathfrak{s}_l$, 
whose real parts have the second largest value. 
As described below or in section \ref{sec:3},
these eigenvalues form a complex-conjugate pair.
We denote them by $E_l^{\pm}$ hereafter.
The Bethe equation \eqref{s:BAE} describing 
$\Spec(\mathfrak{s}_l)$ reduces to
\begin{align}
(p\lambda_j)^L=(-1)^{n_1-1}\prod_{k=1}^{n_1}
  \frac{1-(p+q)\lambda_j+p q \lambda_j \lambda_k}
       {1-(p+q)\lambda_k+p q \lambda_j \lambda_k},
\quad n_1=L-(m_1+\cdots+m_l),
\label{s:BAE-2nd}
\end{align}
where we set the nesting order as $(a_1,a_2)=(1,2)$.
Since the spectrum of the Hamiltonian
is invariant under the change of the nesting order,
and under the transformation $p \leftrightarrow q$\footnote{
This can be easily seen from the
fact that the Bethe equation \eqref{s:BAE}
is invariant under the transformation $p\leftrightarrow q$
and $(a_1,\dots,a_N)\leftrightarrow(N+1-a_1,\dots,N+1-a_N)$.
},
it is enough to consider the case  $p \ge q$ and $n_1 \le L/2$

For $p\ne q$, $E_l^{\pm}$ characterizes the KPZ universality 
class. The corresponding solutions to \eqref{s:BAE-2nd} are
determined as follows \cite{K}. Changing the variable $\lambda$
as 
\begin{align}
p \lambda_j=\frac{1-x_j}{1-\e^{-2\eta} x_j}, \quad \frac{p}{q}=\e^{-2\eta},
\end{align}
we modify \eqref{s:BAE-2nd} and \eqref{s:spectrum}:
\begin{align}
\left(\frac{1-x_j}{1-\e^{-2\eta} x_j}\right)^L=
(-1)^{n_1-1}\prod_{k=1}^{n_1}\frac{x_j-\e^{-2\eta} x_k}
                                  {x_k-\e^{-2\eta} x_j}, \quad
 E=(p-q)\sum_{j=1}^{n_1}
  \left(\frac{x_j}{1-x_j}-\frac{\e^{-2\eta} x_j}{1-\e^{-2\eta}x_j}\right).
\end{align}
The meaning of $\eta$ in the above will be revealed in 
section~\ref{sec:5-3}.
Taking logarithm of both sides, one has
\begin{align}
\frac{L}{2 \pi \i}\log\frac{1-x_j}{x_j^{\rho}(1-\e^{-2\eta} x_j)}
=I_j-\frac{1}{2\pi \i}\sum_{k=1}^{n_1}\left(
     \log x_k-\log\frac{1-\e^{-2\eta}x_k/x_j}{1-\e^{-2\eta}x_j/x_k}\right),
\end{align}
where $\rho=n_1/L$ and $I_j \in {\mathbb Z}+\frac{1+(-1)^{n_1}}{4}$.
In fact, for sufficiently large $n_1$ and $L$, 
the following choice
\begin{align}
I_j^{-}=\begin{cases}
  -\frac{n_1+1}{2}+j \,\, &\text{ for $1\le j\le n_1-1$}\\
  \frac{n_1+1}{2} \,\,    &\text{ for $j=n_1$}
        \end{cases},
\quad
I_j^{+}=-I_j^{-}
\end{align}
gives the solution corresponding to $E_l^{\pm}$.

By carefully taking into account  finite size corrections, 
the asymptotic form of $E_l^{\pm}$ for $L\gg 1$ is determined as
\begin{align}
E_l^{\pm}=\pm 2|(p-q)(1-2\rho)|\pi \i L^{-1}-2 C |p-q|\sqrt{\rho(1-\rho)} 
             L^{-\frac{3}{2}}+O(L^{-2}),
\end{align}
where $C=6.50918933794\ldots$ \cite{K}.
Thus we conclude that the system for $p\ne q$ belongs to the
KPZ universality class whose dynamical exponent is $z=3/2$. 
%
\subsubsection{EW universality class}\label{EWuc}
%
For $p=q$, the set of eigenvalues of the 
Hamiltonian for an arbitrary sector $\mathfrak{t}$ contains 
the relaxation spectrum corresponding to the ``one-magnon" states.
This can be seen by setting all the roots in \eqref{s:BAE-2nd}
except for $\lambda_1$ to $1/p$. Thus  
$E_l^{\pm}$ are given by the second largest eigenvalues
for this one-magnon states,
and obviously do not  depend on $l$. The Bethe
ansatz equation determining the unknown 
$\lambda_1$ simply reduces to
\begin{align}
(p \lambda_1)^L=1.
\end{align}
Solving this and substituting the solutions 
\begin{align}
p \lambda_1=\exp\left(\pm \frac{2\pi k \i}{L}\right) 
\quad
1\le k \le \frac{L}{2}.
\end{align}
 into
\eqref{s:spectrum}, we have $E=-4 p \sin^2(2\pi k/L)$. Obviously
the case $k=1$ gives  the second largest eigenvalues:
\begin{align}
E_l^{\pm}=-4 p \sin^2\left(\frac{\pi}{L}\right) 
        = -4p \pi^2 L^{-2} +O(L^{-4}),
\end{align}
which gives the EW exponent $z=2$.

%
%
\subsection{Parameterization with difference property}\label{sec:5-3}
%
Here we present the Bethe ansatz results 
in a more conventional parameterization \cite{Sc,PS} 
with the spectral parameter having a difference property.

First we treat the one-species case
($N=2$) whose spectrum is given by \eqref{s:spectrum} via
the Bethe ansatz \eqref{s:BAE-2nd}, where $0 \le n_1 \le L$ and
the nesting order is $(a_1,a_2)=(1,2)$. 
Changing the variables  as 
\begin{align}
p \lambda_j^{(1)} \to \exp(\i p_j+\eta), \quad \frac{q}{p}=\e^{-2\eta},
\end{align}
we transform \eqref{s:spectrum} and \eqref{s:BAE-2nd} to
\begin{align}
&E=2\sqrt{p q}\sum_{j=1}^{n_1} (\cos p_j-\Delta), \nn \\
&\e^{\i L p_j}=(-1)^{n_1-1} \e^{-\eta L} \prod_{k=1}^{n_1}
         \frac{1+\e^{\i(p_j+p_k)}-2 \Delta \e^{\i p_j}}
              {1+\e^{\i(p_j+p_k)}-2 \Delta \e^{\i p_k}}, 
\quad \Delta=\cosh\eta.
\label{s:BAE-momentum}
\end{align} 
This is nothing but the eigenvalue of the Hamiltonian 
for the XXZ chain threaded by a ``magnetic flux" $-\i \eta L$:
\begin{align}
H=\sqrt{p q}\sum_{k\in {\mathbb Z}_L} 
\left\{\e^{\eta} \sigma_k^+ \sigma_{k+1}^-
  +\e^{-\eta} \sigma^+_{k+1}\sigma_k^-+\frac{\Delta}{2}(\sigma_k^z \sigma_{k+1}^z-1)
\right\}.
\label{s:h-flux}
\end{align}
The variable $p_j$ in \eqref{s:BAE-momentum} is called the quasi-momentum 
of the Bethe wave function.
Introducing the transformation (see \cite{T} for example)
\begin{alignat}{2}
&\widetilde{p}(u):=
 \frac{1}{\i}\log \frac{\sh\frac{\eta}{2} u}{\sh\frac{\eta}{2}(u+2)},
\quad &&
p_j=\widetilde{p}(\i u_j^{(1)}-1), \nn \\
&\widetilde{\lambda}(u):=\frac{1}{p}\exp(\i \widetilde{p}(u)+\eta), 
&&\lambda_j^{(1)}=\widetilde{\lambda}(\i u_j^{(1)}-1),
\label{s:CV}
\end{alignat}
we rewrite \eqref{s:BAE-momentum} in terms of the ``rapidities" $u_j^{(1)}$:
\begin{align}
E=2\sqrt{p q}\sum_{j=1}^{n_1}\frac{ \sh^2\eta}{\ch(\eta u_j^{(1)})-\ch \eta}, 
\quad
\phi(u_j^{(1)})=-\e^{-\eta L}
                \frac{q_1(u_j^{(1)}+2\i)}{q_1(u_j^{(1)}-2\i)},
\label{s:BAE-rapidity}
\end{align}
where the two functions $\phi(u)$ and $q_1(u)$ are defined by
\begin{align}
\phi(u)=
\left(
\frac{\sin\frac{\eta}{2}(u+\i)}{\sin\frac{\eta}{2}(u-\i)}
\right)^L,
\quad
q_i(u)=\prod_{j=1}^{n_i}\sin\frac{\eta}{2}(u-u_j^{(i)}).
\end{align}

Applying the momentum-rapidity transformation \eqref{s:CV} to 
the $R$-matrix \eqref{s:R} (we write $R(\lambda(u))=\widetilde{R}(u)$),
we find the non-zero elements of $\widetilde{R}(u)$
can be written as
\begin{align}
\widetilde{R}_{\alpha\alpha}^{\alpha\alpha}(u)=1,
\,\,
\widetilde{R}_{\alpha\beta}^{\alpha\beta}(u)=
    \begin{cases}
  \frac{ \e^{-\eta}\sh \frac{\eta}{2}u}{\sh\frac{\eta}{2}(u+2)}   
&\text{ for $\alpha<\beta$}, \\
  \frac{ \e^{\eta}\sh \frac{\eta}{2}u}{\sh\frac{\eta}{2}(u+2)}     
&\text{ for $\alpha>\beta$},
    \end{cases}
\,\,
\widetilde{R}_{\alpha\beta}^{\beta\alpha}(u)=
    \begin{cases}
     \frac{ \e^{\frac{\eta}{2}u}\sh \eta}{\sh\frac{\eta}{2}(u+2)}
                                             &\text{ for $\alpha<\beta$}, \\
     \frac{ \e^{-\frac{\eta}{2}u}\sh \eta}{\sh\frac{\eta}{2}(u+2)}
                               &\text{ for $\alpha>\beta$}.
    \end{cases}
\label{s:R2}
\end{align}
where $\alpha,\beta\in\{1,2\}$ in the present case. 
Up to the asymmetric factors $\e^{\pm \eta}$ and $\e^{\pm \eta u/2}$,
these are the Boltzmann weights for the well-known six vertex model \cite{Ba}
associated with the quantum group $U_q(\widehat{sl}(2))$.
For the gauge factors, see \cite{PS, OY}.
The $R$-matrix \eqref{s:R2} satisfies the Yang-Baxter equation
\begin{align}
\widetilde{R}_{23}(u_2)\widetilde{R}_{13}(u_1)
\widetilde{R}_{12}(u_1-u_2)=
     \widetilde{R}_{12}(u_1-u_2)\widetilde{R}_{13}(u_1)
\widetilde{R}_{23}(u_2),
\label{s:ybe2}
\end{align}
which possesses the difference property.
The Hamiltonian \eqref{s:h-flux} is expressed as the logarithmic
derivative of the transfer matrix $\widetilde{T}(u)$ (cf. \eqref{s:IM}):
\begin{align}
&\widetilde{T}(u)=\tr_{W_0} [\widetilde{R}_{0 L}(\i u-1)
              \cdots\widetilde{R}_{0 1}(\i u-1)], \nn \\
&H=-\frac{2\i\sqrt{p q}\sh \eta}{\eta} 
      \frac{\partial}{\partial u}\ln \widetilde{T}(u) \biggl|_{u=-\i},
\label{s:baxter2}
\end{align}
where $N=2$ in the present case.
Noting that
\begin{align}
&d(\lambda(\i u-1))=(p\lambda(\i u-1))^L=\e^{\eta L}\phi(u), \nn \\
&f(\lambda(\i u-1),\lambda(\i v-1))=\frac{1}{p \xi(\lambda(\i u-1),\lambda(\i v-1))}=
          \e^{-\eta} \frac{\sin\frac{\eta}{2}(u-v-2\i)}{\sin\frac{\eta}{2}(u-v)},
\end{align}
the eigenvalue of the $\widetilde{T}(u)$ for the 
nesting order  $(a_1,a_2)=(1,2)$ 
is given by
\begin{align}
\widetilde{\Lambda}(u)=\frac{q_1(u+2\i)}{q_1(u)}\e^{-\eta n_1 }+
      \phi(u) \frac{q_{1}(u-2\i)}{q_{1}(u)}\e^{\eta L-\eta n_1},
\label{s:DVF2}
\end{align}
via the Bethe equation \eqref{s:BAE-rapidity}.

The extension to the general $(N-1)$-species case is straightforward.
We just let the local states $\alpha, \beta$ in \eqref{s:R2} range over 
$\alpha,\beta\in\{1,\dots, N \}$.  
Finally, we write down the explicit form
of the eigenvalues for an arbitrary nesting order:
\begin{align}
\widetilde{\Lambda}(u)=&\frac{q_1(u+2\i)}{q_1(u)}\e^{-\eta (n_1-2\overline{n}_1) }  \nn \\
    &+\phi(u)
     \sum_{k=1}^{N-2}\frac{q_{k}(u-2\i)}{q_k(u)}
                     \frac{q_{k+1}(u+2\i)}{q_{k+1}(u)}
                \e^{\eta L-\eta(n_k+n_{k+1}+2\sum_{j=1}^{k}(n_{j-1}-n_j)\theta_{j k+1}-
                           2\overline{n}_{k+1})} \nn \\
    &+\phi(u) \frac{q_{N-1}(u-2\i)}{q_{N-1}(u)}
     \e^{\eta L-\eta( n_{N-1}+2\sum_{j=1}^{N-1}(n_{j-1}-n_j)\theta_{j N}}.
\end{align}
Correspondingly the
Bethe equation \eqref{s:BAE} is transformed to
\begin{align}
& \e^{-2\eta (L-n_1)\theta_{12}}\phi(u_j^{(1)}) 
                 =-\e^{-\eta L+\eta (n_2+2(\overline{n}_1-\overline{n}_2))} 
                                      \frac{q_1(u_j^{(1)}+2\i)}
                                           {q_1(u_j^{(1)}-2 \i)}
                                      \frac{q_2(u_j^{(1)})}{q_2(u_j^{(1)}+2\i)}, \nn \\
& \e^{-2\eta(\sum_{j=1}^{l}(n_{j-1}-n_j)\theta_{j l+1}-
                           \sum_{j=1}^{l-1}(n_{j-1}-n_j)\theta_{jl})}        \nn \\
   & \quad             =-\e^{\eta(n_{l+1}-n_{l-1}+2(\overline{n}_{l}-
                                                     \overline{n}_{l+1}))} 
                             \frac{q_l(u_j^{(l)}+2\i)}
                                     {q_l(u_j^{(l)}-2\i)}
                                \frac{q_{l-1}(u_j^{(l)}-2\i)}
                                     {q_{l-1}(u_j^{(l)})}
                                \frac{q_{l+1}(u_j^{(l)})}
                                     {q_{l+1}(u_j^{(l)}+2\i)} \quad 
                                          (2\le l\le N-2),\nn \\
& \e^{-2\eta(\sum_{j=1}^{N-1}(n_{j-1}-n_j)\theta_{j N}-
                           \sum_{j=1}^{N-2}(n_{j-1}-n_j)\theta_{jN-1})}  \nn \\
& \qquad
         =-\e^{\eta (-n_{N-2}+2\overline{n}_{N-1})}          \frac{q_{N-1}(u_j^{(N-1)}+2\i)}
                                     {q_{N-1}(u_j^{(N-1)}-2\i)}
                                \frac{q_{N-2}(u_j^{(N-1)}-2\i)}
                                     {q_{N-2}(u_j^{(N-1)})}.
\end{align}
The spectrum of the Hamiltonian $H$ is then determined by
\begin{align}
E=2\sqrt{p q}\sum_{j=1}^{n_1} \frac{\sh^2\eta}{\ch(\eta u_j^{(1)})-\ch \eta}.
\end{align}

\section*{Acknowledgements}
This work is partially supported by Grants-in-Aid for 
Scientific Research No.~19$\cdot$7744,
(B) No.~18340112 and (C) No.~19540393 from
JSPS.
The author C. A. is grateful 
to Professor H. Hinrichsen,
Professor J. Lebowitz,
Professor A. Schadschneider
and Professor E. R. Speer
for fruitful discussion.

\appendix
\def\thesection{\Alph{section}}
\def\reference{\relax\refpar}

\section{\mathversion{bold}Proof of 
theorem \ref{k:th:dd}}\label{app:dd}

\subsection{M\"obius inversion}
The power set ${\mathcal S}$ 
(\ref{k:eq:I}) is equipped with the natural poset
structure whose partial order is just $\subseteq$.
In this appendix 
the partial order in ${\mathcal M}$  \eqref{k:eq:M}
induced via (\ref{k:eq:bi}) 
will be denoted by $\preceq$.
Thus one has 
$(4) \preceq (1,3) \preceq (1,2,1) 
\preceq (1,1,1,1)$, etc. for $L=4$.
The description of 
$\preceq$ in ${\mathcal M}$ is pretty simple.
In fact, those 
$m'$ satisfying $m' \preceq m=(m_1,\ldots, m_n)$ 
are obtained  from $m$ 
by successive contractions
\begin{align}\label{k:eq:cont}
(\ldots, m_j, m_{j+1},\ldots) \mapsto (\ldots, m_j+m_{j+1},\ldots).
\end{align}
 
Let $\zeta = \bigl(\zeta(\mathfrak{s}',\mathfrak{s})
\bigr)_{\mathfrak{s}', \mathfrak{s} \in {\mathcal S}}$ be the 
$\vert {\mathcal S} \vert \times \vert {\mathcal S} \vert$ 
matrix defined by
\begin{align}
\zeta(\mathfrak{s}', \mathfrak{s}) 
= \begin{cases}1 & \mathfrak{s}' \subseteq \mathfrak{s},\\
0 & \text{otherwise}.
\end{cases}
\end{align}
Since $\zeta$ is a triangular matrix whose diagonal elements are
all 1, it has the inverse 
$\mu = \bigl(\mu(\mathfrak{s}',\mathfrak{s})
\bigr)_{\mathfrak{s}', \mathfrak{s} \in {\mathcal S}}$.
$\mu$ is called the M\"obius function of ${\mathcal S}$, and 
is again a triangular (i.e., $\mu(\mathfrak{s}',\mathfrak{s}) = 0$ 
unless $\mathfrak{s}' \subseteq \mathfrak{s}$) 
integer matrix.
 
Suppose $f, g: {\mathcal S} \rightarrow {\mathbb C}$
are the functions on ${\mathcal S}$.
By the definition, the two relations
\begin{align}\label{k:eq:fg}
f(\mathfrak{s}) = \sum_{\mathfrak{s}' \subseteq \mathfrak{s}}
g(\mathfrak{s}'),\quad
g(\mathfrak{s}) = \sum_{\mathfrak{s}' \subseteq 
{\mathcal \mathfrak{s}}}
\mu(\mathfrak{s}',\mathfrak{s})f(\mathfrak{s}')
\quad (\mathfrak{s} \in {\mathcal S})
\end{align}
are equivalent, where the latter is the M\"obius inversion 
formula.  
In a matrix notation, they are just 
$f = g\zeta$ and $g = f \mu$.
In particular the sum involving $\mu(\mathfrak{s}',\mathfrak{s})$ 
can be restricted to 
$\mathfrak{s}' \subseteq \mathfrak{s}$.
The M\"obius function contains all the information on the
poset structure.
In our case of the power set ${\mathcal S}$, it is a classical result 
(the inclusion-exclusion principle) that 
\begin{align}\label{k:eq:mu}
\mu(\mathfrak{s}',\mathfrak{s}) 
= (-1)^{\sharp \mathfrak{s}' - \sharp \mathfrak{s}},
\end{align}
where $\sharp \mathfrak{s}$ denotes the cardinality of $\mathfrak{s}$. 

The M\"obius inversion formula (\ref{k:eq:fg}) and (\ref{k:eq:mu}) on 
${\mathcal S}$ can be translated into those on ${\mathcal M}$ via 
the bijective correspondence (\ref{k:eq:bi}). The result reads as follows:
\begin{align}
f(m) &= \sum_{m' \preceq m}g(m')\quad (m \in {\mathcal M})
,\label{k:eq:fgm}\\
g(m_1,\ldots, m_n) &
= \sum
(-1)^{n-l}f(i_1,\ldots, i_l)\quad 
((m_1,\ldots, m_n) \in {\mathcal M}),
\label{k:eq:gfm}
\end{align}
where the sum in (\ref{k:eq:gfm}) extends over 
$(i_1,\ldots, i_l) \in {\mathcal M}$ such that 
$(i_1,\ldots, i_l) \preceq (m_1,\ldots, m_n)$.
(We have written $g(m)$ with 
$m=(m_1,\ldots, m_n)$ as $g(m_1,\ldots, m_n)$
rather than $g((m_1,\ldots, m_n))$, and similarly for $f$.)

For $m=(m_1,\ldots, m_n) \in {\mathcal M}$ corresponding to 
$\mathfrak{s} \in {\mathcal S}$, we let  
$\overline{m}$ denote the element in ${\mathcal M}$
that corresponds to the complement 
$\overline{\mathfrak{s}} =\Omega\setminus \mathfrak{s} 
\in {\mathcal S}$.
Thus for $L=4$, $\overline{\phantom{a}}$ acts as the involution 
\begin{align*}
\emptyset \rightleftarrows \{1,2,3\},\;\;
\{1\} \rightleftarrows \{2,3\},\;\;
\{2\} \rightleftarrows \{1,3\},\;\;
\{3\} \rightleftarrows \{1,2\}
\end{align*} 
on ${\mathcal S}$, and similarly 
\begin{align*}
(4) \rightleftarrows (1,1,1,1),\;\;
(1,3) \rightleftarrows (2,1,1),\;\;
(2,2) \rightleftarrows (1,2,1), \;\;
(3,1) \rightleftarrows (1,1,2)
\end{align*}   
on ${\mathcal M}$.
It is an easy exercise to check
\begin{align}\label{k:eq:di}
\overline{(m_1,\ldots, m_n)} = 
(1^{m_1-1}21^{m_2-2}21^{m_3-2}\ldots 
21^{m_{n-1}-2}21^{m_n-1})
\in {\mathcal M},
\end{align}
where ``$a1^{-1}b$" should be understood as $a+b-1$.

\subsection{Theorem \ref{k:th:dd}}

We keep assuming the one to one correspondence 
(\ref{k:eq:bi}) of the labels $m \in {\mathcal M}$ 
and $\mathfrak{s} \in {\mathcal S}$ and use the former. 
In view of (\ref{k:eq:ddv}) we have
$\dim V_\mathfrak{s}^\ast = f(m)$ by the choice:
\begin{align}\label{k:eq:fdef}
f(m) = \mn{L}{m_1,\ldots, m_n}
:=\frac{L!}{m_1!\cdots m_n!}
\quad \text{for}\;\;
m = (m_1,\ldots, m_n) \in {\mathcal M}.
\end{align}
Denote the $g(m)$ determined from this and (\ref{k:eq:gfm}) by
\begin{align}\label{k:eq:gdef}
g(m) = \mc{L}{m_1,\ldots, m_n}.
\end{align}
Namely, we define 
\begin{align}\label{k:eq:gdef2}
\mc{L}{m_1,\ldots, m_n} = \sum (-1)^{n-l}
\mn{L}{i_1,\ldots, i_l}\quad \text{for }
(m_1,\ldots, m_n) \in {\mathcal M},
\end{align}
where the sum runs over 
$(i_1,\ldots, i_l) \in {\mathcal M}$ such that 
$(i_1,\ldots, i_l) \preceq (m_1,\ldots, m_n)$.
{}From (\ref{k:eq:ddv}), 
(\ref{k:eq:dvx}), (\ref{k:eq:fgm}) and (\ref{k:eq:fdef}), 
we find $\dim X_\mathfrak{s}^\ast = g(m)$.
The function $g(m)$ has the invariance
\begin{align*}
\mc{L}{m_1,\ldots, m_n} = 
\mc{L}{m_n,\ldots, m_1},
\end{align*}
but it is {\em not} symmetric under general
permutations of $m_1,\ldots, m_n$ in contrast to $f(m)$.

Now theorem \ref{k:th:dd} is translated into
\begin{theorem}\label{k:th:add}
\begin{align*}
g(m) = g(\overline{m}) \;\;\text{for any } m \in {\mathcal M}.
\end{align*}
\end{theorem}

\begin{example}\label{k:ex:dd1}
Take $m=(L)$ hence $\ol{m}=(1^L)$. 
Then $g(m)=\mc{L}{L}=\mn{L}{L}=1$.
On the other hand, $g(\ol{m})$ with $L=3$ and $4$ 
are calculated as
\begin{align*}
&\left(\!\!\begin{array}{c} 3 \\  \! 1,1,1\! \end{array}\!\!\right)
-\left(\!\!\begin{array}{c} 3 \\  \! 2,1\! \end{array}\!\!\right)
-\left(\!\!\begin{array}{c} 3 \\  \! 1,2 \! \end{array}\!\!\right)
+\left(\!\!\begin{array}{c} 3 \\  \! 3\! \end{array}\!\!\right) = 1,\\
&\left(\!\!\begin{array}{c} 4 \\  \! 1,1,1,1\! \end{array}\!\!\right)
-\left(\!\!\begin{array}{c} 4 \\  \! 2,1,1\! \end{array}\!\!\right)
-\left(\!\!\begin{array}{c} 4 \\  \! 1,2,1 \! \end{array}\!\!\right)
-\left(\!\!\begin{array}{c} 4 \\  \! 1,1,2 \! \end{array}\!\!\right)
+\left(\!\!\begin{array}{c} 4 \\  \! 3,1\! \end{array}\!\!\right)
+\left(\!\!\begin{array}{c} 4 \\  \! 2,2\! \end{array}\!\!\right)
+\left(\!\!\begin{array}{c} 4 \\  \! 1,3\! \end{array}\!\!\right)
-\left(\!\!\begin{array}{c} 4 \\  \! 4\! \end{array}\!\!\right)=1.
\end{align*}
\end{example}

\begin{example}
Take $L=5$ and 
$m=(1,2,1,1)$ hence $\overline{m} = (2,3)$. 
Then one has 
\begin{align}
\begin{split}
\mc{5}{1,2,1,1}&=\mn{5}{1,2,1,1} 
- \mn{5}{1,3,1}-\mn{5}{1,2,2}+\mn{5}{1,4}\\
&-\mn{5}{3,1,1} + \mn{5}{3,2} + \mn{5}{4,1} - \mn{5}{5},
\end{split}\label{k:eq:wake}\\
\mc{5}{2,3}&=\mn{5}{2,3}-\mn{5}{5}.
\end{align}
The both of these sums yield $9$.
\end{example}

\subsection{Proof}
We first generalize example \ref{k:ex:dd1} to 
\begin{lemma}\label{k:le:mae}
\begin{align*}
\mc{L}{L} = \mc{L}{1^{L}}\;\;\text{for any} \;\;L \ge 1.
\end{align*}
\end{lemma}
\begin{proof}
The left hand side is $1$.
{}From (\ref{k:eq:gdef2}), the right hand side is given by
\begin{align*}
\mc{L}{1^L} = \sum_{l=1}^LA_{L,l},\qquad
A_{L,l} = (-1)^{L-l}\sum_{i_1,\ldots, i_l}
\mn{L}{i_1,\ldots, i_l},
\end{align*}
where the latter sum 
extends over $i_1,\ldots, i_l \in {\mathbb Z}_{\ge 1}$ 
such that $i_1+\cdots + i_l = L$.
Thus we have the following evaluation of the generating functions:
\begin{align*}
\sum_{L\ge l}\frac{A_{L,l}}{L!}z^L &
=(-1)^{l}\sum_{i_1,\ldots, i_l}
\frac{(-z)^{i_1}\cdots(-z)^{i_l}}{i_1!\cdots i_l!}= (1-\e^{-z})^{l},\\
\sum_{L\ge 1}\frac{z^L}{L!}\sum_{l=1}^LA_{L,l} &= 
\sum_{l\ge 1}\sum_{L\ge l}\frac{A_{L,l}}{L!}z^L
= \sum_{l\ge 1}(1-\e^{-z})^l = \e^z-1.
\end{align*}
The last relation tells that $\sum_{l=1}^LA_{L,l}=1$ for $L \ge 1$.
\end{proof}

In (\ref{k:eq:wake}), note that the first line of the right hand side equals 
$\mn{5}{1,4}\mc{4}{2,1,1}$, whereas the second line is nothing but
$-\mc{5}{3,1,1}$, therefore one has
\begin{align*}
\mc{5}{1,2,1,1} = \mn{5}{1,4}\mc{4}{2,1,1} -\mc{5}{3,1,1}.
\end{align*}
The following lemma shows that such a decomposition holds generally.
\begin{lemma}\label{k:le:decom}
\begin{align}\label{k:eq:dec}
\mc{L}{m_1,m_2,\ldots, m_n} = 
\binom{L}{m_1}\mc{L-m_1}{m_2,\ldots, m_n} 
-\mc{L}{m_1+m_2,m_3,\ldots, m_n},
\end{align} 
where $\binom{L}{m_1}=\binom{L}{m_1,L-m_1}$ 
denotes the binomial coefficient.
\end{lemma}
\begin{proof}
In making $(i_1,\ldots, i_l)$ from $(m_1,\ldots, m_n)$ by 
the successive contractions (\ref{k:eq:cont}),
classify the summands in (\ref{k:eq:gdef2}) according to whether 
$m_1$ has been contracted to $m_2$ or not.  
If it is not contracted, then $i_1=m_1$ always holds and the 
corresponding summands yield 
$\mn{L}{m_1}\mc{L-m_1}{m_2,\ldots, m_n}$.
The other summands correspond to the contracted case 
$i_1\ge m_1+m_2$, whose contribution is  
$-\mc{L}{m_1+m_2,m_3,\ldots, m_n}$.
\end{proof}

\begin{lemma}\label{k:le:rec}
\begin{align}\label{k:eq:rec}
\mc{L}{1^t, \pi} = \sum_{s=0}^a(-1)^{s+a}\mn{L}{s}
\mc{L-s}{a+1-s,1^{t-a-1},\pi}\quad (1 \le a \le t-1),
\end{align}
where $\pi$ is an arbitrary array of positive integers 
summing up to $L-t$.
\end{lemma}
\begin{proof}
We employ the induction on $a$.
The case $a=1$, i.e.,
\begin{align*}
\mc{L}{1^t,\pi} = -\mc{L}{2,1^{t-2},\pi} 
+ \mn{L}{1}\mc{L-1}{1^{t-1},\pi}
\end{align*}
follows from lemma \ref{k:le:decom}. 
Suppose that (\ref{k:eq:rec}) holds for $a=a$.  
Substitution of (\ref{k:eq:dec}) gives 
\begin{align}\label{k:eq:sp1}
\mc{L}{1^t, \pi} =  
\sum_{s=0}^a(-1)^{s+a}\mn{L}{s}
\left[\mn{L-s}{a+1-s}\mc{L-a-1}{1^{t-a-1},\pi} - 
\mc{L-s}{a+2-s,1^{t-a-2},\pi}\right].
\end{align}
In the first term, the $s$-sum is taken as
\begin{align}\label{k:eq:wa}
\sum_{s=0}^a(-1)^{s+a}\mn{L}{s}\mn{L-s}{a+1-s} = 
\mn{L}{a+1}\sum_{s=0}^a(-1)^{s+a}\mn{a+1}{s} = \mn{L}{a+1}.
\end{align}
Thus (\ref{k:eq:sp1}) implies 
the $a\rightarrow a\!+\!1$ case of (\ref{k:eq:rec}).
\end{proof}

\begin{lemma}\label{k:le:owari}
\begin{align*}
\mc{L}{1^{a+1},\pi} + \mc{L}{1^{a-1},2,\pi}
= \mn{L}{a}\mc{L-a}{1,\pi},
\end{align*}
where $\pi$ is an arbitrary array of positive integers 
summing up to $L-a-1$.
\end{lemma}
\begin{proof}
The case $t=a+1$ in 
(\ref{k:eq:rec}) reads
\begin{align*}
\mc{L}{1^{a+1},\pi} = \sum_{s=0}^a(-1)^{s+a}\mn{L}{s}
\mc{L-s}{a+1-s,\pi}.
\end{align*}
Similarly, by setting $t\rightarrow a-1, a\rightarrow a-2$ and 
$\pi \rightarrow (2,\pi)$ in (\ref{k:eq:rec}), we get
\begin{align*}
\mc{L}{1^{a-1},2,\pi} &= \sum_{s=0}^{a-2}(-1)^{s+a}\mn{L}{s}
\mc{L-s}{a-1-s,2,\pi}\\
&=\sum_{s=0}^{a-2}(-1)^{s+a}\mn{L}{s}
\left[\mn{L-s}{a-1-s}\mc{L-a+1}{2,\pi}
-\mc{L-s}{a+1-s,\pi}\right],\label{k:eq:aho}
\end{align*}
where (\ref{k:eq:dec}) has been applied in the second equality.
The sum of the two expressions gives
\begin{align*}
\mc{L}{1^{a+1},\pi} + \mc{L}{1^{a-1},2,\pi} 
&=-\mn{L}{a-1}\mc{L-a+1}{2,\pi}+
\mn{L}{a}\mc{L-a}{1,\pi}\\
&+\left[
\sum_{s=0}^{a-2}(-1)^{s+a}\mn{L}{s}
\mn{L-s}{a-1-s}\right]\mc{L-a+1}{2,\pi}.
\end{align*}
The $s$-sum in the last term is $\binom{L}{a-1}$ 
due to (\ref{k:eq:wa}),
finishing the proof.
\end{proof}

\noindent{\em Proof of Theorem \ref{k:th:add}}.
For $(m_1,\ldots, m_n) \in {\mathcal M}$, we are to show 
\begin{align*}
\mc{L}{m_1,\ldots, m_n} = 
\mc{L}{1^{m_1-1}21^{m_2-2} \ldots1^{m_{n-1}-2}21^{m_n-1}},
\end{align*}
where the explicit form of the dual is taken from (\ref{k:eq:di}). 
We invoke the double induction on $(L, n)$.
The case $L=1$ is trivially true.
In addition, the case $(m_1,\ldots, m_n) = (L)$ has already been 
verified for all $L$ in lemma \ref{k:le:mae}.
We assume that the assertion for $(L',n')$ is true 
for $L'<L$ with any $n'$ 
and also for $(L,n')$ with $n'<n$.
Consider the decomposition (\ref{k:eq:dec}).
By the induction assumption, 
the two quantities $\mc{\cdot}{\cdot }$ on the right hand side
can be replaced with their dual.
Then the relation to be proved becomes
\begin{align*}
&\mc{L}{1^{m_1-1}21^{m_2-2} 
\ldots1^{m_{n-1}-2}21^{m_n-1}}\\
&=\mn{L}{m_1}
\mc{L-m_1}{1^{m_2-1}21^{m_3-2} 
\ldots1^{m_{n-1}-2}21^{m_n-1}}
-\mc{L}{1^{m_1+m_2-1}21^{m_3-2} 
\ldots1^{m_{n-1}-2}21^{m_n-1}}.
\end{align*}
In terms of 
$\pi=(1^{m_2-2}21^{m_3-2} \ldots1^{m_{n-1}-2}21^{m_n-1})$,
this is expressed as 
\begin{align*}
\mc{L}{1^{m_1-1},2,\pi}
= \mn{L}{m_1}\mc{L-m_1}{1,\pi}
-\mc{L}{1^{m_1+1},\pi}.
\end{align*}
The proof is finished by lemma \ref{k:le:owari}. \qed

\section{Derivation of the stationary state}\label{sec:stationary}

Here we sketch a procedure to derive the stationary state in the 
framework of the algebraic Bethe ansatz. Throughout this appendix,
we set the nesting order as a standard one: $a_j=j\,(1\le j \le N)$.

As seen in section~\ref{sec:stationary-eigenvalue}, the eigenvalue
of the transfer matrix for the stationary state can be simply calculated
by setting all the Bethe roots equal to $1/p$. In contrast to the
eigenvalue problem, the evaluation of the eigenstate is not trivial.
This is caused by the $B$-operators such as \eqref{s:B-op1} and 
\eqref{s:B-op2} that 
approach zero as $\lambda\to1/p$. Thus to obtain the state,  
we must normalize $B$-operators as
\begin{align}
\widetilde{B}^{(j)}_{\alpha_k}(\lambda):=\frac{B^{(j)}_{\alpha_k}(\lambda)}{1-p \lambda},
\quad \alpha_k\in \{j+2,\dots,N\}
\quad (1\le k \le n_{j+1};\,0\le j \le N-2).
\end{align}
Note that  $B^{(0)}(\lambda):=B(\lambda)$.
First we consider the eigenstate $|\{\lambda^{(N-1)}\}\ket$
of $T^{(N-2)}(\lambda|\{\lambda^{(N-2)}\})$.
As shown in section \ref{sec:5-1},
this state is constructed by a multiple action of 
$\widetilde{B}^{(N-2)}_{\alpha}$ on $|\vac^{(N-2)}\ket=|N-1\ket_1\otimes\cdots
\otimes|N-1\ket_{n_{N-2}}$:
\begin{align}
|\{\lambda^{(N-1)}\}\ket&=\lim_{\lambda_j^{(N-1)}\to 1/p}
\widetilde{B}^{(N-2)}_{N}(\lambda_1^{(N-1)})\cdots\widetilde{B}^{(N-2)}_{N}
(\lambda_{n_{N-1}}^{(N-1)})|\vac^{(N-2)}\ket  \nn \\
&={\rm const.}\sum_{1\le \gamma_1<\dots <\gamma_{n_{N-1}}\le n_{N-2}}\left\{\prod_{j=1}^{n_{N-1}}
\frac{1-q\lambda_{\gamma_j}^{(N-2)}}{1-p\lambda_{\gamma_j}^{(N-2)}}
\right\}
\left(
\bigotimes_{k=1}^{n_{N-2}}
|N-1+\sum_{j=1}^{n_{N-1}}\delta_{k \gamma_j}\ket_k
\right)   \nn \\
&=:F^{(N-3)\alpha_1\cdots \alpha_{n_{N-2}}}|\alpha_1\ket_1\otimes\cdots
   \otimes |\alpha_{n_{N-2}}\ket_{n_{N-2}}, \,\,
\alpha_j\in\{N-1,N\}\,\, (1\le j \le n_{N-2}),
\label{s:1st}
\end{align}
where $F^{(N-3)\alpha_1 \cdots \alpha_{n_{N-2}}}$ is the 
element of  $|\{\lambda^{(N-1)}\}\ket$, from 
which the eigenstate $|\{\lambda^{(N-2)}\}\ket$ 
of $T^{(N-3)}(\lambda|\{\lambda^{(N-3)}\})$ can
be constructed as
\begin{align}
|\{\lambda^{(N-2)}\}\ket=\lim_{\lambda_j^{(N-2)}\to 1/p}F^{(N-3)\alpha_1 \cdots \alpha_{n_{N-2}}}
\widetilde{B}^{(N-3)}_{\alpha_1}(\lambda_1^{(N-2)}) \cdots 
\widetilde{B}^{(N-3)}_{\alpha_{n_{N-2}}}(\lambda_{n_{N-2}}^{(N-2)})|\vac^{(N-3)}\ket.
\label{s:2nd}
\end{align}
Note that $F^{(0)\alpha_1\cdots\alpha_{n_{1}}}$ denotes
$F^{\alpha_1\cdots\alpha_{n_{1}}}$ defined by \eqref{s:BV}.
Since the coefficient  $F^{(N-3)\alpha_1 \cdots \alpha_{n_{N-2}}}$ contains
the term $\prod (1-p \lambda_j^{(N-2)})^{-1}$, we cannot take the limit
$\lambda_j^{(N-2)}\to 1/p$ ($1 \le j \le n_{N-2})$ independently of $j$. 
To take this limit correctly,
we solve for the roots $\lambda_j^{(N-2)}$ ($2\le j \le n_{N-2}$) in terms of $\lambda_1^{(N-2)}$
by using the Bethe equations \eqref{s:BAE}:
\begin{alignat}{2}
& d(\lambda_j^{(1)}) =
    -\prod_{k=1}^{n_1} 
       \frac{f(\lambda_k^{(1)},\lambda_j^{(1)})}{f(\lambda_j^{(1)},\lambda_k^{(1)})}
                             &\text{ for $N=3$}, \nn \\
&\frac{1}{\prod_{k=1}^{n_{N-3}}f(\lambda_j^{(N-2)},\lambda_k^{(N-3)})} = -
    \prod_{k=1}^{n_{N-2}}\frac{f(\lambda_k^{(N-2)},\lambda_j^{(N-2)})}
                          {f(\lambda_j^{(N-2)},\lambda_k^{(N-2)})}
            \quad
                                         & \text{ for $N >3$}.
\label{s:BAE-expnd}
\end{alignat}
In the above, we have put $\lambda_j^{(N-1)}=1/p$. To extract the behavior  
of $\lambda_j^{(N-2)}$ around the point $1/p$, we expand them in terms of $\lambda_1^{(N-2)}-1/p$
as
\begin{align}
\lambda_j^{(N-2)}=\frac{1}{p}+\sum_{k=1}^{n_{N-2}}g_{j}^{(k)}(\lambda_1^{(N-2)}-1/p)^{k}
+O((\lambda_1^{(N-2)}-1/p)^{n_{N-2}+1}) \quad \text{ for $1\le j \le n_{N-2}$},
\label{s:expansion}
\end{align}
where $g_{1}^{(k)}=\delta_{1k}$.
Inserting them into \eqref{s:BAE-expnd} and comparing the coefficients of each order,
one obtains the set of equations determining the coefficients $g_{j}^{(k)}$.
In the following, as an example, we  write down the equations determining the first 
three coefficients.
\begin{align}
\begin{split}
&b^{(1)}=(-1)^{n_{N-2}+1}\prod_{k=1}^{n_{N-2}}\frac{c_{j k}^{(1)}}{c_{k j}^{(1)}}, \quad
b^{(2)}=b^{(1)}\sum_{k=1}^{n_{N-2}}\left\{\frac{c_{j k}^{(2)}}{c_{j k}^{(1)}}-
                                   \frac{c_{k j}^{(2)}}{c_{k j}^{(1)}}\right\},\\
&b^{(3)}=b^{(1)}\sum_{1\le k<l \le n_{N-2}}\left\{\frac{c_{j k}^{(2)}}{c_{j k}^{(1)}}-
                                   \frac{c_{k j}^{(2)}}{c_{k j}^{(1)}}\right\}
                                   \left\{\frac{c_{j l}^{(2)}}{c_{j l}^{(1)}}-
                                   \frac{c_{l j}^{(2)}}{c_{l j}^{(1)}}\right\}\\
&\qquad \qquad \quad+b^{(1)}\sum_{k=1}^{n_{N-2}}\left\{
                        \frac{c_{j k}^{(3)}}{c_{j k}^{(1)}}-
                        \frac{c_{j k}^{(2)}c_{k j}^{(2)}}{c_{j k}^{(1)}c_{k j}^{(1)}}+
                        \frac{(c_{k j}^{(2)})^2}{(c_{k j}^{(1)})^2}-
                        \frac{c_{k j}^{(3)}}{c_{k j}^{(1)}}
                     \right\},
\label{s:det}
\end{split}
\end{align}
where $c_{jk}$ and $b$ are defined as
\begin{align}
\begin{split}
&c_{jk}^{(1)}=q g_{k}^{(1)}-p g_{j}^{(1)},\quad
c_{jk}^{(2)}=q g_{k}^{(2)}-p g_{j}^{(2)}+
             p q g_{j}^{(1)}g_{k}^{(1)} \\
&c_{jk}^{(3)}=q g_{k}^{(3)}-p g_{j}^{(3)}+
             p q(g_{j}^{(1)}g_{k}^{(2)}+g_{j}^{(2)}g_{k}^{(1)}),
\end{split}
\end{align}
and, for $N=3$,
\begin{align}
b^{(1)}=1, \quad
b^{(2)}=L p g_{j}^{(1)}, \quad
b^{(3)}=L p g_{j}^{(2)}+\frac{1}{2}L(L-1)(p g_{j}^{(1)})^2,
\end{align}
and, for $N>3$,
\begin{align}
\begin{split}
&b^{(1)}=1, \quad
b^{(2)}=\sum_{k=1}^{n_{N-3}}\frac{1-q \lambda_k^{(N-3)}}{1-p \lambda_k^{(N-3)}}p g_j^{(1)},  \\
&b^{(3)}=\sum_{1\le k < l \le n_{N-3}} \frac{1-q \lambda_k^{(N-3)}}{1-p \lambda_k^{(N-3)}}
                                       \frac{1-q \lambda_l^{(N-3)}}{1-p \lambda_l^{(N-3)}}
 (p g_{j}^{(1)})^2                                \\
&\qquad \qquad \quad +\sum_{k=1}^{n_{N-3}}\left\{
                   \frac{1-q \lambda_k^{(N-3)}}{1-p \lambda_k^{(N-3)}} p g_{j}^{(2)}
           -\frac{(1-q \lambda_k^{(N-3)})q\lambda_k^{(N-3)}}{(1-p \lambda_k^{(N-3)})^2}
                  (p g_{j}^{(1)})^2
                     \right\}.
\end{split}
\end{align}
By solving these equations, the coefficients $g_{j}^{(k)}$ are uniquely determined.
For instance, $g_{j}^{(1)}$ is simply given by roots of unity:
\begin{align}
g_{j}^{(1)}=\exp\left\{\frac{2\pi \i}{n_{N-2}}(j-1)\right\} \quad
(1\le j\le n_{N-2}).
\end{align}
Substituting \eqref{s:expansion} together with the explicit form 
of the coefficients $g_{j}^{(k)}$
into \eqref{s:2nd} and then taking the limit $\lambda_1^{(N-2)}\to 1/p$,
one obtains the eigenstate $|\{\lambda^{(N-2)}\}\ket$ whose elements
give $F^{(N-4)\alpha_1 \cdots \alpha_{n_{N-3}}}$.
Repeating this procedure, one calculates the stationary state.

As a simple example, let us demonstrate this procedure
for the maximal sector of $L=N=3$ ($n_j=L-j$\, $(1\le j\le 2)$).
From \eqref{s:1st} $|\lambda^{(2)}\ket$ is given by
\begin{align}
|\lambda^{(2)}\ket=-p\frac{1-q \lambda_2^{(1)}}{1-p \lambda_2^{(1)}}|2\ket_1\otimes|3\ket_2
-p\frac{1-q \lambda_1^{(1)}}{1-p \lambda_1^{(1)}}|3\ket_1\otimes|2\ket_2.
\end{align} 
Namely $F^{2 3}$ and $F^{3 2}$  are, respectively, given by
\begin{align}
F^{2 3}=-p\frac{1-q \lambda_2^{(1)}}{1-p \lambda_2^{(1)}},
\quad
F^{3 2}=-p\frac{1-q \lambda_1^{(1)}}{1-p \lambda_1^{(1)}}.
\label{s:F}
\end{align}
To take the limit in \eqref{s:2nd}, we 
expand $\lambda_2^{(1)}$ by solving \eqref{s:det}. The resultant
expression reads
\begin{align}
\lambda_2^{(1)}=\frac{1}{p}-(\lambda_1^{(1)}-1/p)+\frac{p (3 p+q)}{p-q}(\lambda_1^{(1)}-1/p)^2
+O((\lambda_1^{(1)}-1/p)^3).
\end{align}
Inserting this and \eqref{s:F} into \eqref{s:1st}, and taking
the limit $\lambda_1^{(1)}\to 1/p$, we finally arrive at
\begin{align}
&|\{\lambda^{(1)}\}\rangle
=-p (p+q)C^{\alpha_1\alpha_2\alpha_3}
              |\alpha_1\ket_1\otimes |\alpha_2\ket_2\otimes |\alpha_3\ket_3, \nn \\
&C^{123}=C^{231}=C^{312}=2p+q, \quad
 C^{132}=C^{213}=C^{321}=p+2q.
\end{align}

We leave it as a future task to extend the concrete calculation 
as above to the general case and derive 
an explicit formula for the stationary state.
For an alternative approach by the matrix product ansatz, 
see \cite{PEM}.

\section{Eigenvalues and Bethe roots for $N=L=4$}\label{sec:appC}
We list the spectrum of the Hamiltonian \eqref{k:hdef} and
the transfer matrix \eqref{s:TM}
in the basic sectors
for $N=L=4$ and 
$(p,q)=(2/3,1/3)$.
The corresponding Bethe roots
with the standard nesting order
are also listed here.
The spectrum of Hamiltonian is also obtained 
by specializing the result in figure \ref{spec}. 
The second and the third examples from the last
demonstrate that there are two Bethe roots that yield the 
same eigen-polynomial.
These Bethe roots are not regular and 
the eigen-polynomial is not genuine in the sense of 
section \ref{k:subs:cp}. 
The sectors are specified by 
$(m_1,\ldots, m_4)$ according to \eqref{k:eq:sort}.

\begin{table}
\caption{\label{table1}Table of eigenvalues of
the transfer matrix and the Hamiltonian and
the corresponding Bethe roots I.}
\begin{tabular}{@{}crrrrr}
\br
 sector & $\Lambda(\lambda)$ & $E$ &
$\{\lambda^{(1)}\}$  & $\{\lambda^{(2)}\}$ & $\{\lambda^{(3)}\}$
\\
\mr
(4,0,0,0) &\scriptsize{$\frac{16 \lambda^4}{27}+1$} & $0 $ & $\emptyset$ & $\emptyset$ &$\emptyset$  
                                                                                  \\[-0.5mm] \mr 
$(1,3,0,0)$ &\scriptsize{$\frac{34 \lambda^4}{81}+1$} & 0 &
{\scriptsize $\{1.5,1.5,1.5\}$}        & $\emptyset$ & $\emptyset$ 
\\[-0.5mm] \mr 
\multicolumn{2}{r}{\scriptsize{$\frac{10 \lambda^4}{27}+\frac{2 \lambda^3}{9}-\frac{2 \lambda^2}{3}+2 \lambda -1 $ }}
      &$-2$ & {\tiny $\left\{
                                          \begin{array}{@{\,}r@{\,}}
                                             -2.51978 \\
                                              1.00989-0.565265 \i \\
                                              1.00989+0.565265 \i
                                           \end{array}
                                           \right\}$} & $\emptyset$ &$\emptyset$   
 \\[-0.5mm] \mr
\multicolumn{2}{r}{\scriptsize{$\begin{array}{@{}r@{}}
        \left(\frac{32}{81}\pm \frac{2 \i}{81}\right)\lambda^4-
        \left(\frac{1}{27}\pm \frac{\i}{9}\right)\lambda^3  \\[0.5mm]
        + \left(\frac{1}{3}\mp \frac{\i}{9}\right)\lambda^2   
       + \left(\frac{1}{3}\pm \i\right) \lambda \mp\i 
       \end{array} $ }}&$-1\pm \frac{\i}{3}$ & {\tiny $\left\{
                                          \begin{array}{@{\,}r@{\,}}
                                            -0.809148\pm 2.91994 \i \\
                                             0.824307\pm 0.182529 \i \\ 
                                              1.16517\mp 0.618862 \i 
                                           \end{array}
                                           \right\}$}   &   $\emptyset$ &$\emptyset$   
 \\[-0.5mm] \mr
 $(2,2,0,0)$ &\scriptsize{$\frac{4 \lambda^4}{9}+1$}& 0
&{\scriptsize $\{1.5,1.5\}$}        & $\emptyset$ &$\emptyset$   
 \\[-0.5mm] \mr
\multicolumn{2}{r}{\scriptsize{$\begin{array}{@{}r@{}}
    \left(\frac{32}{81}\pm\frac{4 \i}{81}\right) 
     \lambda^4\mp\frac{2 \i \lambda^3}{9}        
     +\frac{5 \lambda^2}{9}         \\
       \pm \i \lambda \mp \i
   \end{array}  $ }} & $(-1)^2$ & {\tiny $\left\{
                                          \begin{array}{@{\,}r@{\,}}
                                             -0.241158\pm1.94173 \i \\
                                              1.14116\mp0.141729 \i
                                           \end{array}
                                           \right\}$} & $\emptyset$ &$\emptyset$   
 \\[-0.5mm] \mr
\multicolumn{2}{r}{\scriptsize{$\frac{4 \lambda^4}{9}-\frac{2 \lambda^3}{3}+3 \lambda^2-3 \lambda +1$}} & $-3$ & {\tiny $\left\{
                                          \begin{array}{@{\,}r@{\,}}
                                              -0.75-1.29904 \i \\
                                              -0.75+1.29904 \i 
                                           \end{array}
                                           \right\}$} & $\emptyset$ &$\emptyset$   
 \\[-0.5mm] \mr
\multicolumn{2}{r}{\scriptsize{$\frac{28 \lambda^4}{81}+\frac{8 \lambda ^3}{27}-
   \frac{2 \lambda^2}{9}+\frac{4 \lambda }{3}-1$}} &$-\frac{4}{3}$
                            & {\tiny $\left\{
                                          \begin{array}{@{\,}r@{\,}}
                                              -3.62132 \\
                                               0.62132 
                                           \end{array}
                                           \right\}$} & $\emptyset$ &$\emptyset$   
 \\[-0.5mm] \mr
\multicolumn{2}{r}{\scriptsize{$\frac{28 \lambda^4}{81}+\frac{10 \lambda^3}{27}-
 \frac{5 \lambda^2}{9}+\frac{5 \lambda }{3}-1$ }}&$-\frac{5}{3}$
                            & {\tiny $\left\{
                                          \begin{array}{@{\,}r@{\,}}
                                              -2.42705 \\
                                               0.927051 \\
                                           \end{array}
                                           \right\}$} & $\emptyset$ &$\emptyset$   
\\[-0.5mm] \mr
$(3,1,0,0)$ &\scriptsize{$\frac{40 \lambda ^4}{81}+1$}& $0$
&\scriptsize{$\{1.5\}$} &  $\emptyset$ &$\emptyset$   
 \\[-0.5mm] \mr
\multicolumn{2}{r}{\scriptsize{$\frac{8 \lambda^4}{27}+\frac{8 \lambda^3}{9}-
\frac{4 \lambda^2}{3}+2 \lambda -1$}} & $-2$
&\scriptsize{$\{-1.5\}$}  &  $\emptyset$ &$\emptyset$ \\[-0.5mm] \mr
\multicolumn{2}{r}{\scriptsize{$\begin{array}{@{}r@{}}
   \left(\frac{32}{81}\mp\frac{8 \i}{81}\right) \lambda^4+
   \left(\frac{4}{27}\pm\frac{4 \i}{9}\right) \lambda^3 \\[0.5mm]+
   \left(\frac{2}{3}\mp\frac{2 \i}{9}\right) \lambda^2-
   \left(\frac{1}{3}\pm\i\right) \lambda \pm\i
   \end{array}$ }} & $-1\pm \frac{\i}{3}$
&\scriptsize{$\{\mp 1.5\i\}$}&  $\emptyset$ &$\emptyset$
\\[-0.5mm] \mr
$(1,1,2,0)$&\scriptsize{$\frac{22 \lambda^4}{81}+1$ }& 0 &{\scriptsize $\{1.5,1.5,1.5\}$}       &{\scriptsize $\{1.5,1.5\}$}        &$\emptyset$   
 \\[-0.5mm] \mr
\multicolumn{2}{r}{\scriptsize{$\frac{2 \lambda^4}{9}+\frac{2 \lambda^3}{9}-\frac{2 \lambda^2}{3}+2 \lambda -1 $}} 
      &$-2$ & {\tiny $\left\{
                                          \begin{array}{@{\,}r@{\,}}
                                              1.00989-0.565265 \i \\
                                              1.00989+0.565265 \i \\
                                              -2.51978 \\
                                           \end{array}
                                           \right\}$}
                          & {\scriptsize $\{1.5,1.5\}$}   &$\emptyset$   
 \\[-0.5mm] \mr
\multicolumn{2}{r}{\scriptsize{$\begin{array}{@{}r@{}}
\left(\frac{20}{81}\pm\frac{2 \i}{81}\right) \lambda^4-
\left(\frac{1}{27}\pm\frac{\i}{9}\right) \lambda^3 \\
+\left(\frac{1}{3}\mp\frac{\i}{9}\right) \lambda^2
+\left(\frac{1}{3}\pm\i\right) \lambda \mp \i
   \end{array}$}}  & $-1\pm \frac{\i}{3}$ &{\tiny $\left\{
                                          \begin{array}{@{\,}r@{\,}}
                                                -0.809148\pm2.91994 \i \\
                                                 0.824307\pm0.182529 \i \\
                                                 1.16517\mp0.618862 \i
                                           \end{array}
                                           \right\}$} 
                                           &{\scriptsize $\{1.5,1.5\}$} 
 &$\emptyset$   
 \\[-0.5mm] \mr
\multicolumn{2}{r}{\scriptsize{$\begin{array}{@{}r@{}}
    \left(\frac{2}{9}\pm\frac{4 \i}{81}\right) 
     \lambda^4\mp\frac{2 \i \lambda^3}{9}        
     +\frac{5 \lambda^2}{9}         \\
       \pm \i \lambda \mp \i
   \end{array}  $}} & $(-1)^2$ & {\tiny $\left\{
                                          \begin{array}{@{\,}r@{\,}}
                                             -0.241158\pm1.94173 \i \\
                                              1.14116\mp0.141729 \i \\
                                              1.5
                                           \end{array}
                                           \right\}$} 
                          & {\tiny $\left\{
                                          \begin{array}{@{\,}r@{\,}}
                                             0.69207\pm0.648316 \i \\
                                             1.19884\mp0.161286 \i \\ 
                                           \end{array}
                                           \right\}$} &$\emptyset$   
 \\[-0.5mm] \mr
\multicolumn{2}{r}{\scriptsize{$\frac{22 \lambda^4}{81}-\frac{2 \lambda^3}{3}+3 \lambda^2-3 \lambda +1$}}&$-3$& {\tiny $\left\{
                                          \begin{array}{@{\,}r@{\,}}
                                             1.5, -0.75-1.29904 \i \\
                                             -0.75+1.29904 \i
                                           \end{array}
                                           \right\}$}
                          & {\tiny $\left\{
                                          \begin{array}{@{\,}r@{\,}}
                                               0.789474-0.957186 \i \\
                                               0.789474+0.957186 \i
                                           \end{array}
                                           \right\}$}  &$\emptyset$   
 \\[-0.5mm] \mr
\multicolumn{2}{r}{\scriptsize{$\frac{14 \lambda^4}{81}+\frac{8 \lambda^3}{27}-
   \frac{2 \lambda^2}{9}+\frac{4 \lambda }{3}-1$}}&$-\frac{4}{3}$ & {\tiny $\left\{
                                          \begin{array}{@{\,}r@{\,}}
                                       1.5, -3.62132 \\ 
                                           0.62132
                                           \end{array}
                                           \right\}$}
                          & {\tiny $\left\{
                                          \begin{array}{@{\,}r@{\,}}
                                            0.686441-0.503364 \i \\
                                            0.686441+0.503364 \i
                                           \end{array}
                                           \right\}$}  &$\emptyset$   
 \\[-0.5mm] \mr
\multicolumn{2}{r}{\scriptsize{$\frac{14 \lambda^4}{81}+\frac{10 \lambda^3}{27}-\frac{5 \lambda^2}{9}
   +\frac{5 \lambda }{3}-1$}} &$-\frac{5}{3}$  & {\tiny $\left\{
                                          \begin{array}{@{\,}r@{\,}}
                                           1.5, -2.42705 \\
                                             0.927051 
                                           \end{array}
                                           \right\}$}
                          & {\tiny $\left\{
                                          \begin{array}{@{\,}r@{\,}}
                                             0.765957-0.188811 \i \\ 
                                             0.765957+0.188811 \i
                                           \end{array}
                                           \right\}$}  &$\emptyset$   
 \\[-0.5mm] \mr
\multicolumn{2}{r}{\scriptsize{$\frac{10 \lambda^4}{81}+\frac{2 \lambda^3}{9}+\frac{4 \lambda^2}{3}-2 \lambda +1$}} &$-2$ & {\tiny $\left\{
                                          \begin{array}{@{\,}r@{\,}}
                                              -0.61352-2.06536 \i \\
                                              -0.61352+2.06536 \i \\
                                                0.72704 
                                           \end{array}
                                           \right\}$} 
                          & {\tiny $\left\{
                                          \begin{array}{@{\,}r@{\,}}
                                             -9.97723 \\
                                              0.977226 
                                           \end{array}
                                           \right\}$}  &$\emptyset$   
 \\[-0.5mm] \mr
\multicolumn{2}{r}{\scriptsize{$\begin{array}{@{}r@{}}
    \left(\frac{8}{81}\pm\frac{10 \i}{81}\right) \lambda^4+
    \left(\frac{13}{27}\mp\i\right) \lambda^3 \\
    +\left(\frac{1}{3}\pm\frac{25 \i}{9}\right)\lambda^2-
    \left(\frac{1}{3}\pm3 \i\right) \lambda \pm\i
   \end{array}  $ }} & $-3\pm \frac{\i}{3}$ & {\tiny $\left\{
                                          \begin{array}{@{\,}r@{\,}}
                                              -1.97293\mp0.68627 \i \\
                                              -0.405293\pm1.74466 \i \\
                                               0.762834\mp0.481463 \i 
                                           \end{array}
                                           \right\}$}
                          & {\tiny $\left\{
                                          \begin{array}{@{\,}r@{\,}}
                                              -3.45362\pm7.73728 \i \\
                                               1.07431\mp0.289005 \i
                                           \end{array}
                                           \right\}$}  &$\emptyset$   
\\[-0.5mm]
\mr
$(1,2,1,0)$ &\scriptsize{$\frac{26 \lambda^4}{81}+1$}& 0 & {\scriptsize $\{1.5,1.5,1.5\}$}   &{\scriptsize $\{1.5\}$}& $\emptyset$   
 \\[-0.5mm] \mr
\multicolumn{2}{r}{\scriptsize{$\frac{22 \lambda^4}{81}+\frac{2 \lambda^3}{9}-\frac{2 \lambda^2}{3}+2 \lambda -1 $}}
      &$-2$ & {\tiny $\left\{
                                          \begin{array}{@{\,}r@{\,}}
                                              1.00989-0.565265 \i \\
                                              1.00989+0.565265 \i \\
                                              -2.51978 \\
                                           \end{array}
                                           \right\}$} & {\scriptsize $\{1.5\}$}    &$\emptyset$   
 \\[-0.5mm] \mr
\multicolumn{2}{r}{\scriptsize{$\begin{array}{@{}r@{}}
\left(\frac{8}{27}\pm\frac{2 \i}{81}\right) \lambda^4-
\left(\frac{1}{27}\pm\frac{\i}{9}\right) \lambda^3 \\
+\left(\frac{1}{3}\mp\frac{\i}{9}\right) \lambda^2
+\left(\frac{1}{3}\pm\i\right) \lambda \mp \i
   \end{array}$}} & $-1\pm \frac{\i}{3}$ &{\tiny $\left\{
                                          \begin{array}{@{\,}r@{\,}}
                                                -0.809148\pm2.91994 \i \\
                                                 0.824307\pm0.182529 \i \\
                                                 1.16517\mp0.618862 \i
                                           \end{array}
                                           \right\}$} 
& {\scriptsize $\{1.5\}$}  &$\emptyset$   
 \\[-0.5mm] \mr
\multicolumn{2}{r}{\scriptsize{$\frac{10 \lambda ^4}{81}+\frac{8 \lambda^3}{9}-\frac{4 \lambda^2}{3}+2 \lambda -1$}} & $-2$ &{\scriptsize $\{-1.5,1.5,1.5\}$}  & {\scriptsize $\{-0.5\}$}  
                                                                              &$\emptyset$   
 \\[-0.5mm] \mr
\multicolumn{2}{r}{\scriptsize{$\begin{array}{@{}r@{}}
\left(\frac{2}{9}\mp\frac{8 \i}{81}\right) \lambda^4+
\left(\frac{4}{27}\pm\frac{4\i}{9}\right) \lambda^3 \\
+\left(\frac{2}{3}\mp\frac{2 \i}{9}\right) \lambda^2-
\left(\frac{1}{3}\pm\i\right) \lambda \pm \i   \end{array}$}}
& $-1\pm \frac{\i}{3}$ &{\scriptsize $\{\mp 1.5 \i , 1.5 ,1.5\}$}
& {\tiny $\{ 0.253846\mp0.969231 \i\}$} & $\emptyset$   
 \\[-0.5mm] \br
\end{tabular}                                                                  
\end{table}

\begin{table}
\caption{\label{table2}Table of eigenvalues of
the transfer matrix and the Hamiltonian and
the corresponding Bethe roots II.}
\begin{tabular}{@{}crrrrr}
\br
sector & $\Lambda(\lambda)$ & $E$ &
$\{\lambda^{(1)}\}$  & $\{\lambda^{(2)}\}$ & $\{\lambda^{(3)}\}$
\\[-0.5mm] \mr
(1,2,1,0) &
\scriptsize{$\frac{16 \lambda ^4}{81}+\frac{\lambda ^3}{3}+\lambda -1$}& $-1$ &{\tiny $\left\{
                                          \begin{array}{@{\,}r@{\,}}
                                                 -7.5494 \\
                                                  0.383413 \\
                                                  1.16599 
                                           \end{array}
                                           \right\}$} &
{\scriptsize $\{1.5\}$}    &$\emptyset$   
 \\[-0.5mm] \mr
\multicolumn{2}{r}{\scriptsize{$\begin{array}{@{}r@{}}
\left(\frac{2}{81}\mp\frac{14 \i}{81}\right) \lambda^4+
\left(\frac{10}{9}\pm\frac{11 \i}{9}\right) \lambda^3 \\
-
\left(\frac{5}{9}\pm3 \i\right) \lambda ^2\pm3 \i \lambda \mp\i
\end{array}$ }}  & $(-3)^2$ &{\tiny $\left\{
                                          \begin{array}{@{\,}r@{\,}}
                                                 -1.9214\pm0.238512 \i \\
                                                 -0.315515\mp1.62413 \i \\
                                                  1.00162\pm0.326798 \i
                                           \end{array}
                                           \right\}$}
& {\tiny $\{-2.19231\mp3.46154 \i\}$} & $\emptyset$  
 \\[-0.5mm] \mr
\multicolumn{2}{r}{\scriptsize{$\frac{4 \lambda^4}{27}+\frac{7 \lambda^3}{27}+
  \frac{14 \lambda^2}{9}-\frac{7 \lambda }{3}+1$}}& $-\frac{7}{3}$ &{\tiny $\left\{
                                          \begin{array}{@{\,}r@{\,}}
                                                 -0.511746-1.74231 \i \\ 
                                                 -0.511746+1.74231 \i \\
                                                  1.02349 
                                           \end{array}
                                           \right\}$} 
&{\scriptsize $\{1.5\}$}  &$\emptyset$   
\\[-0.5mm] \mr
\multicolumn{2}{r}{\scriptsize{$
\frac{2 \lambda ^4}{9}-\frac{4 \lambda^3}{27}+\frac{20 \lambda^2}{9}
-\frac{8 \lambda }{3}+1 $}}& $-\frac{8}{3}$ &{\tiny $\left\{
                                          \begin{array}{@{\,}r@{\,}}
                                                -0.717003-1.5691 \i \\
                                                -0.717003+1.5691 \i \\
                                                 1.13401
                                           \end{array}
                                           \right\}$} & {\tiny $\{-0.214286\}$}  
                                                                              &$\emptyset$   
\\[-0.5mm]
\mr
$(2,1,1,0)$ &\scriptsize{$\frac{28 \lambda^4}{81}+1$}& 0 
& {\scriptsize $\{1.5,1.5\}$}
 &{\scriptsize $\{1.5\}$}& $\emptyset$   
 \\ \mr
\multicolumn{2}{r}{\scriptsize{$\begin{array}{@{}r@{}}
    \left(\frac{8}{27}\pm\frac{4 \i}{81}\right) 
     \lambda^4\mp\frac{2 \i \lambda^3}{9}        
     +\frac{5 \lambda^2}{9}         \\
       \pm \i \lambda \mp \i
   \end{array}  $}} & $(-1)^2$ & {\tiny $\left\{
                                          \begin{array}{@{\,}r@{\,}}
                                             -0.241158\pm1.94173 \i \\
                                              1.14116\mp0.141729 \i
                                           \end{array}
                                           \right\}$} & 
{\scriptsize $\{1.5\}$} &$\emptyset$   
 \\[-0.5mm] \mr
\multicolumn{2}{r}{\scriptsize{$\frac{28 \lambda^4}{81}-\frac{2 \lambda^3}{3}+3 \lambda^2-
   3\lambda +1$ }} & $-3$ &{\tiny $\left\{
                                          \begin{array}{@{\,}r@{\,}}
                                                 -0.75-1.29904 \i \\
                                                 -0.75+1.29904 \i 
                                           \end{array}
                                           \right\}$} & {\scriptsize $\{1.5\}$}   &$\emptyset$   
 \\[-0.5mm] \mr
\multicolumn{2}{r}{\scriptsize{$\frac{20 \lambda^4}{81}+\frac{8 \lambda^3}{27}-
\frac{2 \lambda^2}{9}+\frac{4 \lambda }{3}-1$}} & $-\frac{4}{3}$ &{\tiny $\left\{
                                          \begin{array}{@{\,}r@{\,}}
                                                 -3.62132\\ 
                                                  0.62132
                                           \end{array}
                                           \right\}$} & {\scriptsize $\{1.5\}$}    &$\emptyset$   
 \\[-0.5mm] \mr
\multicolumn{2}{r}{\scriptsize{$\frac{20 \lambda^4}{81}+\frac{10 \lambda^3}{27}-
  \frac{5 \lambda^2}{9}+\frac{5 \lambda }{3}-1$ }} & $-\frac{5}{3}$ &{\tiny $\left\{
                                          \begin{array}{@{\,}r@{\,}}
                                                 -2.42705 \\
                                                  0.927051
                                           \end{array}
                                           \right\}$} & {\scriptsize $\{1.5\}$}   &$\emptyset$   
 \\[-0.5mm] \mr
\multicolumn{2}{r}{\scriptsize{$\frac{4 \lambda^4}{27}+\frac{8 \lambda^3}{9}
  -\frac{4 \lambda^2}{3}+2 \lambda -1$}} & $-2$
&{\scriptsize $\{-1.5,1.5\}$} &{\scriptsize $\{0.3\}$}     &$\emptyset$   
 \\[-0.5mm] \mr
\multicolumn{2}{r}{\scriptsize{$\begin{array}{@{}r@{}}
\left(\frac{20}{81}\mp\frac{8 \i}{81}\right) \lambda^4+
\left(\frac{4}{27}\pm\frac{4 \i}{9}\right) \lambda^3 \\
+
\left(\frac{2}{3}\mp\frac{2 \i}{9}\right) \lambda^2-
\left(\frac{1}{3}\pm \i\right) \lambda \pm \i
\end{array}$ }} & $-1\pm \frac{\i}{3}$
&{\scriptsize $\{\mp 1.5 \i ,1.5\}$}
& {\tiny $\{ 0.617647\mp 0.529412 i\}$}  
                                                                              & $\emptyset$   
 \\[-0.5mm] \mr
\multicolumn{2}{r}{\scriptsize{$\frac{4 \lambda^4}{81}+\frac{8 \lambda^3}{9}+
   \frac{2 \lambda^2}{3}-2 \lambda +1$}}& $-2$
&{\scriptsize $\{-1.5 \i,1.5 \i\}$} & {\scriptsize $\{-4.5\}$}   &$\emptyset$   
 \\[-0.5mm] \mr
\multicolumn{2}{r}{\scriptsize{$\begin{array}{@{}r@{}}
\left(-\frac{4}{81}\mp \frac{16 \i}{81}\right) \lambda^4+
\left(\frac{44}{27}\pm \frac{4 \i}{3}\right) \lambda^3 \\
-\left(\frac{4}{3}\pm \frac{28 \i}{9}\right) \lambda^2+
\left(\frac{1}{3}\pm 3 \i\right) \lambda \mp \i
\end{array}$ }} & $-3\pm \frac{\i}{3}$
&{\scriptsize $\{- 1.5,\mp 1.5 \i\}$}
& {\tiny $\{ -2.55882\mp4.76471 \i \}$}  
& $\emptyset$
\\[-0.5mm] \mr
(1,1,1,1) &\scriptsize{$\frac{14 \lambda ^4}{81}+1$} & 0
&{\scriptsize $\{1.5,1.5,1.5\}$} & {\scriptsize $\{1.5,1.5\}$}
 &{\scriptsize $\{1.5\}$}  
\\[-0.5mm] \mr 
\multicolumn{2}{r}{
\scriptsize{$\frac{10 \lambda^4}{81}+\frac{2 \lambda^3}{9}-\frac{2 \lambda^2}{3}+2 \lambda -1 $}} 
      &$-2$ & {\tiny $\left\{
                                          \begin{array}{@{\,}r@{\,}}
                                              1.00989-0.565265 \i \\
                                              1.00989+0.565265 \i \\
                                              -2.51978 \\
                                           \end{array}
                                           \right\}$}
 & {\scriptsize $\{1.5,1.5\}$} &{\scriptsize $\{1.5\}$}  
\\[-0.5mm] \mr 
\multicolumn{2}{r}{\scriptsize{$\begin{array}{@{}r@{}}
\left(\frac{4}{27}\pm\frac{2 \i}{81}\right) \lambda^4-
\left(\frac{1}{27}\pm\frac{\i}{9}\right) \lambda^3 \\
+\left(\frac{1}{3}\mp\frac{\i}{9}\right) \lambda^2
+\left(\frac{1}{3}\pm\i\right\} \lambda \mp \i
   \end{array}$}}   & $-1\pm \frac{\i}{3}$ &{\tiny $\left\{
                                          \begin{array}{@{\,}r@{\,}}
                                                -0.809148\pm2.91994 \i \\
                                                 0.824307\pm0.182529 \i \\
                                                 1.16517\mp0.618862 \i
                                           \end{array}
                                           \right\}$} 
                                           &{\scriptsize $\{1.5,1.5\}$}
&{\scriptsize $\{1.5\}$}   
\\[-0.5mm] \mr 
\multicolumn{2}{r}{\scriptsize{$\begin{array}{@{}r@{}}
    \left(\frac{10}{81}\pm\frac{4 \i}{81}\right)
     \lambda^4\mp\frac{2 \i \lambda^3}{9}        
     +\frac{5 \lambda^2}{9}         \\
       \pm \i \lambda \mp \i
   \end{array}  $ }} & $(-1)^2$ & {\tiny $\left\{
                                          \begin{array}{@{\,}r@{\,}}
                                             -0.241158\pm1.94173 \i \\
                                              1.14116\mp0.141729 \i \\
                                              1.5
                                           \end{array}
                                           \right\}$} 
                          & {\tiny $\left\{
                                          \begin{array}{@{\,}r@{\,}}
                                             0.69207\pm0.648316 \i \\
                                             1.19884\mp0.161286 \i \\ 
                                           \end{array}
                                           \right\}$} &{\scriptsize $\{1.5\}$}   
\\[-0.5mm] \mr
\multicolumn{2}{r}{\scriptsize{$\frac{14 \lambda^4}{81}
-\frac{2 \lambda^3}{3}+3 \lambda^2-3 \lambda +1$}}
    &$-3$& {\tiny $\left\{
                                          \begin{array}{@{\,}r@{\,}}
                                             -0.75-1.29904 \i \\
                                             -0.75+1.29904 \i \\ 
                                              1.5
                                           \end{array}
                                           \right\}$}
                          & {\tiny $\left\{
                                          \begin{array}{@{\,}r@{\,}}
                                               0.789474-0.957186 \i \\
                                               0.789474+0.957186 \i
                                           \end{array}
                                           \right\}$}  &{\scriptsize $\{1.5\}$}  
\\[-0.5mm] \mr
\multicolumn{2}{r}{\scriptsize{$\frac{2 \lambda^4}{27}+\frac{8 \lambda^3}{27}-
   \frac{2 \lambda^2}{9}+\frac{4 \lambda }{3}-1$}}
&$-\frac{4}{3}$ & {\tiny $\left\{
                                          \begin{array}{@{\,}r@{\,}}
                                         1.5,  -3.62132 \\ 
                                           0.62132
                                           \end{array}
                                           \right\}$}
                          & {\tiny $\left\{
                                          \begin{array}{@{\,}r@{\,}}
                                            0.686441-0.503364 \i \\
                                            0.686441+0.503364 \i
                                           \end{array}
                                           \right\}$}  &{\scriptsize $\{1.5\}$}   
\\[-0.5mm]
\mr
\multicolumn{2}{r}{\scriptsize{$\frac{2 \lambda^4}{27}+\frac{10 \lambda^3}{27}-\frac{5 \lambda^2}{9}
   +\frac{5 \lambda }{3}-1$}}&$-\frac{5}{3}$  & {\tiny $\left\{
                                          \begin{array}{@{\,}r@{\,}}
                                          1.5,   -2.42705 \\
                                             0.927051
                                           \end{array}
                                           \right\}$}
                          & {\tiny $\left\{
                                          \begin{array}{@{\,}r@{\,}}
                                             0.765957-0.188811 \i \\ 
                                             0.765957+0.188811 \i
                                           \end{array}
                                           \right\}$}  &{\scriptsize $\{1.5\}$}
                                                                               \\[-0.5mm]\mr
\multicolumn{2}{r}{\scriptsize{$\frac{2 \lambda^4}{81}+\frac{2 \lambda^3}{9}+\frac{4 \lambda^2}{3}-2 \lambda +1$ }}
      &$-2$ & {\tiny $\left\{
                                          \begin{array}{@{\,}r@{\,}}
                                              -0.61352-2.06536 \i \\
                                              -0.61352+2.06536 \i \\
                                                0.72704 
                                           \end{array}
                                           \right\}$} 
                          & {\tiny $\left\{
                                          \begin{array}{@{\,}r@{\,}}
                                             -9.97723 \\
                                              0.977226 
                                           \end{array}
                                           \right\}$}
&{\scriptsize $\{1.5\}$}                \\[-0.5mm]\mr
\multicolumn{2}{r}{
\scriptsize{$\begin{array}{@{}r@{}}
    \pm\frac{10 \i}{81} \lambda^4+
    \left(\frac{13}{27}\mp\i\right) \lambda^3 \\
    +\left(\frac{1}{3}\pm\frac{25 \i}{9}\right)\lambda^2-
    \left(\frac{1}{3}\pm3 \i\right) \lambda \pm\i
   \end{array}  $}}& $-3\pm \frac{\i}{3}$ & {\tiny $\left\{
                                          \begin{array}{@{\,}r@{\,}}
                                              -1.97293\mp0.68627 \i \\
                                              -0.405293\pm1.74466 \i \\
                                               0.762834\mp0.481463 \i 
                                           \end{array}
                                           \right\}$}
                          & {\tiny $\left\{
                                          \begin{array}{@{\,}r@{\,}}
                                              -3.45362\pm7.73728 \i \\
                                               1.07431\mp0.289005 \i
                                           \end{array}
                                           \right\}$}  &{\scriptsize $\{1.5\}$}  
\\[-0.5mm] \mr
\multicolumn{2}{r}{
\scriptsize{$\frac{4 \lambda^4}{81}+\frac{\lambda^3}{3}+\lambda -1$}}
    &$-1$& {\tiny $\left\{
                                          \begin{array}{@{\,}r@{\,}}
                                              1.16599,-7.5494 \\
                                              0.383413
                                           \end{array}
                                           \right\}$}
                          & {\scriptsize $\{0.5,1.5\}$}  &{\tiny $\{0.954545\}$}   
\\[-0.5mm] \mr
\multicolumn{2}{r}{\scriptsize{$\begin{array}{@{}r@{}}
   \left(-\frac{10}{81}\mp\frac{14 \i}{81}\right) \lambda^4+
   \left(\frac{10}{9}\pm\frac{11 \i}{9}\right) \lambda^3 \\ -
   \left(\frac{5}{9}\pm3 \i\right) \lambda^2\pm3 \i \lambda \mp\i
   \end{array}  $}} & $(-3)^2$ & {\tiny $\left\{
                                          \begin{array}{@{\,}r@{\,}}
                                              -1.9214\pm0.238512 \i \\
                                              -0.315515\mp1.62413 \i \\
                                               1.00162\pm0.326798 \i
                                           \end{array}
                                           \right\}$}
                          & {\tiny $\left\{
                                          \begin{array}{@{\,}r@{\,}}
                                              -2.19231\mp3.46154 \i \\
                                              1.5
                                           \end{array}
                                           \right\}$}  &\tiny{$\{-0.101751\mp0.59081 \i\}$}   
\\[-0.5mm]
\br
\end{tabular}
\end{table}

\begin{table}
\caption{\label{table3}Table of eigenvalues of
the transfer matrix and the Hamiltonian and
the corresponding Bethe roots III.}
\begin{tabular}{@{}crrrrr}
\br
sector &
$\Lambda(\lambda)$ & $E$ &
$\{\lambda^{(1)}\}$  & $\{\lambda^{(2)}\}$ & $\{\lambda^{(3)}\}$
\\ \mr
(1,1,1,1)
\\[-0.5mm]
\multicolumn{2}{r}{\scriptsize{$\frac{7 \lambda^3}{27}+\frac{14 \lambda^2}{9}-\frac{7 \lambda }{3}+1$}}
      &$-\frac{7}{3}$ & {\tiny $\left\{
                                          \begin{array}{@{\,}r@{\,}}
                                             -0.511746-1.74231 \i \\
                                             -0.511746+1.74231 \i \\
                                              1.02349
                                           \end{array}
                                           \right\}$} 
                          & {\scriptsize $\{-1.5,1.5\}$}
&{\scriptsize $\{0.3\}$}
 \\[-0.5mm] \mr
\multicolumn{2}{r}{\scriptsize{$\frac{2 \lambda^4}{27}-\frac{4 \lambda^3}{27}+\frac{20 \lambda^2}{9}-\frac{8 \lambda }{3}+1$}} 
      &$-\frac{8}{3}$ & {\tiny $\left\{
                                          \begin{array}{@{\,}r@{\,}}
                                              -0.717003-1.5691 \i \\
                                              -0.717003+1.5691 \i \\
                                               1.13401
                                           \end{array}
                                           \right\}$}
                          & {\tiny $\left\{
                                          \begin{array}{@{\,}r@{\,}}
                                              -0.214286\\
                                              1.5 
                                           \end{array}
                                           \right\}$}  &{\tiny $\{0.672414\}$}   
\\[-0.5mm] \mr
\multicolumn{2}{r}{\scriptsize{$-\frac{10 \lambda^4}{81}+\frac{8 \lambda^3}{9}+\frac{2 \lambda^2}{3}-2 \lambda +1$}}
      &$-2$ & {\scriptsize $\{-1.5 \i,1.5 \i, 1.5\}$} 
 & {\tiny $\{ -2.06427, 0.778553\}$}
&{\tiny $\{-1.77273\}$}   
\\[-0.5mm] \mr
\multicolumn{2}{r}{\scriptsize{$\begin{array}{@{}r@{}}
   \left(-\frac{2}{9}\mp \frac{16 \i}{81}\right) \lambda^4+
   \left(\frac{44}{27}\pm \frac{4 \i}{3}\right) \lambda^3 \\-
   \left(\frac{4}{3}\pm \frac{28 \i}{9}\right) \lambda^2+
   \left(\frac{1}{3}\pm 3 \i\right) \lambda \mp \i
   \end{array}  $}} & $-3\pm \frac{\i}{3}$
& {\scriptsize $\{-1.5,\mp 1.5 \i ,1.5\}$}
                          & {\tiny $\left\{
                                          \begin{array}{@{\,}r@{\,}}
                                              -2.35155\mp3.02941 \i \\
                                              0.874627\mp0.15521 \i
                                           \end{array}
                                           \right\}$}
&\tiny{$\{-1.63706\mp1.91878 \i\}$}   
\\[-0.5mm] \mr
\multicolumn{2}{r}{\scriptsize{$-\frac{2 \lambda^4}{81}+\frac{8 \lambda^3}{9}-\frac{4 \lambda^2}{3}+2 \lambda -1$ }}&
$-2$ & \scriptsize
{$\{-1.5,1.5,1.5 \}$ } &
\tiny{$\{-0.5,1.5\}$} & \tiny{$\{0.576923\}$}
\\[-0.5mm] \cline{5-6}
& & & & \tiny{$\{1.5,1.5\}$} & \tiny{$\{-1.5\}$}
\\[-1mm] \mr
\multicolumn{2}{r}{\scriptsize{$\begin{array}{@{}r@{}}
\left(\frac{2}{27}\mp \frac{8 \i}{81}\right) \lambda^4+
\left(\frac{4}{27}\pm \frac{4 \i}{9}\right) \lambda^3 \\
+
\left(\frac{2}{3}\mp \frac{2 \i}{9}\right) \lambda^2-
\left(\frac{1}{3}\pm \i\right) \lambda  \pm \i
   \end{array}  $}} & $-1\pm\frac{\i}{3}$ 
                          & {\scriptsize $\{\mp 1.5 \i,1.5,1.5 \} $} 
& \tiny{$\left\{\begin{array}{@{\,}r@{\,}}0.253846\mp 0.969231 \i \\ 1.5 \end{array}\right\}$} &
\tiny{$\{0.784404\mp0.385321 \i\}$}\\
\cline{5-6}
& & & & \tiny{$\{1.5,1.5\}$}& \tiny{$\{\mp 1.5 \i\}$} 
\\[-0.5mm] \mr
\multicolumn{2}{r}{\scriptsize{$-\frac{50 \lambda^4}{81}+4 \lambda^3-6 \lambda^2+4 \lambda -1$}}  
& $-4$ & {\scriptsize $\{-1.5,-1.5\i,1.5\i\}$}
 & {\tiny $\left\{
                                          \begin{array}{@{\,}r@{\,}}
                                              -0.5-3.74166 \i \\
                                              -0.5+3.74166 \i
                                           \end{array}
                                           \right\}$}
 &\scriptsize{$\{25.5\}$}   
\\[-0.5mm] \br
\end{tabular}
\end{table}

\end{document}